\documentclass[a4paper,11pt]{article} 
\usepackage[a4paper,hmargin=2.8cm,vmargin=3cm]{geometry}
\usepackage[utf8x]{inputenc} 
\usepackage[english]{babel}
\usepackage{amsmath,amsthm,amsfonts,amssymb,mathrsfs}
\usepackage{authblk}
\usepackage{lineno}
\usepackage[dvipsnames]{xcolor}

\definecolor{BeauBlue}{rgb}{0, 0.2, .9}
\definecolor{BeauOrange}{rgb}{.8, .1, 0}
\usepackage{hyperref}
\hypersetup{
	colorlinks = true,
	linkcolor = BeauBlue,
	urlcolor = cyan,
	citecolor = BeauOrange,
}

\numberwithin{equation}{section}

\usepackage{tikz}
\usetikzlibrary{calc,shapes,patterns,decorations.pathreplacing}

\makeatletter
\DeclareRobustCommand\widecheck[1]{{\mathpalette\@widecheck{#1}}}
\def\@widecheck#1#2{%
    \setbox\z@\hbox{\m@th$#1#2$}%
    \setbox\tw@\hbox{\m@th$#1%
       \widehat{%
          \vrule\@width\z@\@height\ht\z@
          \vrule\@height\z@\@width\wd\z@}$}%
    \dp\tw@-\ht\z@
    \@tempdima\ht\z@ \advance\@tempdima2\ht\tw@ \divide\@tempdima\thr@@
    \setbox\tw@\hbox{%
       \raise\@tempdima\hbox{\scalebox{1}[-1]{\lower\@tempdima\box
\tw@}}}%
    {\ooalign{\box\tw@ \cr \box\z@}}}
\makeatother

\makeatletter
\newcommand{\leqnos}{\tagsleft@true\let\veqno\@@leqno}
\newcommand{\reqnos}{\tagsleft@false\let\veqno\@@eqno}
\reqnos
\makeatother

\newcommand{\triple}[1]{{\left\vert\kern-0.25ex\left\vert\kern-0.25ex\left\vert #1 
    \right\vert\kern-0.25ex\right\vert\kern-0.25ex\right\vert}}

\newcommand{\im}{\operatorname{Im}}
\newcommand{\ii}{\mathrm{i}}
\newcommand{\ee}{\mathrm{e}}

\newcommand{\R}{\mathbb{R}}

\newcommand{\tr}{\mathrm{tr}}
\newcommand{\op}{\mathrm{op}}

\newcommand{\veps}{\varepsilon}
\newcommand{\W}{\mathcal{W}}

\newcommand{\F}{\mathcal{F}}
\newcommand{\cN}{\mathcal{N}}

\def\tr{\mathrm{tr}}

\def\bR{\mathbb{R}}

\def\cE{\mathcal{E}}
\def\cH{\mathcal{H}}
\def\cN{\mathcal{N}}
\def\cU{\mathcal{U}}

\newtheorem{theorem}{Theorem}[section] 
\newtheorem{proposition}[theorem]{Proposition} 
\newtheorem{corollary}[theorem]{Corollary}
\newtheorem{lemma}[theorem]{Lemma}

\newtheorem{remark}[theorem]{Remark}
\newtheorem{assumption}[theorem]{Assumption}

\title{Effective Dynamics of Extended Fermi Gases in the High-Density Regime}

\author[1]{Luca Fresta}
\author[2]{Marcello Porta}
\author[3]{Benjamin Schlein}
\affil[1]{Hausdorff Center for Mathematics,
University of Bonn, Endenicher Allee 60
53115 Bonn, Germany}
\affil[2]{SISSA, Via Bonomea 265, 34136 Trieste, Italy}
\affil[3]{Institute of Mathematics, University of Zurich, Winterthurerstrasse 190, 8057 Zurich, Switzerland}

\begin{document}

\maketitle

\begin{abstract}
We study the quantum evolution of many-body Fermi gases in three dimensions, in arbitrarily large domains. We consider both particles with non-relativistic and with relativistic dispersion. We focus on the high-density regime, in the semiclassical scaling, and we consider a class of initial data describing zero-temperature states. In the non-relativistic case we prove that, as the density goes to infinity, the many-body evolution of the reduced one-particle density matrix converges to the solution of the time-dependent Hartree equation, for short macroscopic times. In the case of relativistic dispersion, we show convergence of the many-body evolution to the relativistic Hartree equation for all macroscopic times. With respect to previous work, the rate of convergence does not depend on the total number of particles, but only on the density: in particular, our result allows us to study the quantum dynamics of extensive many-body Fermi gases.
\end{abstract}

\tableofcontents     

\section{Introduction}

In the last years there has been substantial progress in the derivation of effective equations for interacting fermions in the mean-field regime. In this scaling limit, we consider systems of $N$ particles, confined in a region $\Lambda \subset \bR^3$ with volume of order one, interacting through a weak two-body potential with range comparable to the size of $\Lambda$. Denoting by $V_\text{ext}$ the trapping potential and by $V$ the interaction, the Hamilton operator takes the form 
\begin{equation}\label{eq:HNmf} H_N^\text{mf} (V_\text{ext}) = \sum_{j=1}^N \big[ -\varepsilon^2 \Delta_{x_j} + V_\text{ext} (x_j) \big] + \frac{1}{N} \sum_{i<j}^N V (x_i - x_j) \end{equation}  
and, in accordance with the Pauli principle, it acts on $L^2_a (\bR^{3N})$, the subspace of $L^2 (\bR^{3N})$ consisting of functions that are antisymmetric with respect to permutations. In (\ref{eq:HNmf}), we set $\varepsilon = N^{-1/3}$. Together with the factor $N^{-1}$ in front of the potential energy, this choice guarantees that all terms in the Hamilton operator are, typically, of order $N$. In fact, because of the fermionic statistics, the expectation of $\sum_{j=1}^N -\Delta_{x_j}$ on states trapped in a volume of order one is at least of order $N^{5/3}$; this can be verified with the Lieb-Thirring inequality, see {\it e.g.} \cite[Chapter 4]{LiSa}.

To describe low-energy properties of (\ref{eq:HNmf}), we introduce Hartree--Fock theory, defined by restricting (\ref{eq:HNmf}) to Slater determinants, {\it i.e.}, to wave functions of the form 
\begin{equation}\label{eq:slater}  \psi_\text{Slater} (x_1, \dots , x_N) = \frac{1}{\sqrt{N!}}\text{det } (f_i (x_j))_{1 \leq i,j \leq N}  \end{equation} 
where $\{ f_j \}_{j=1}^N$ is an orthonormal family in the one-particle space $L^2 (\bR^3)$. Slater determinants are an example of quasi-free states: they are completely characterized by their one-particle reduced density matrix 
\[ \omega_N = N \tr_{2,\dots , N} |\psi_\text{Slater} \rangle \langle \psi_\text{Slater}| = \sum_{j=1}^N |f_j \rangle \langle f_j|\;, \]
coinciding with the orthogonal projection onto the $N$-dimensional subspace of $L^2 (\bR^3)$, spanned by the orbitals $\{ f_j \}_{j=1}^N$. In particular, the energy of the Slater determinant (\ref{eq:slater}) is given by the Hartree--Fock energy functional 
\begin{equation}\label{eq:HFfun} \begin{split} \cE_\text{HF} (\omega_N) &= \langle \psi_\text{Slater}, H_N^\text{mf} (V_\text{ext}) \psi_\text{Slater} \rangle \\ &= \tr \, \big[ -\varepsilon^2 \Delta + V_\text{ext} \big] \omega_N + \frac{1}{2N} \int dx dy \, V(x-y) \big[ \omega_N (x;x) \omega_N (y;y) - |\omega_N (x;y)|^2 \big]\;. \end{split} \end{equation} 
The interaction contributes to (\ref{eq:HFfun}) through the direct term, proportional to the product of the particle densities $\omega_N (x;x)$ and $\omega_N (y;y)$ and through the exchange term, proportional to $|\omega_N (x;y)|^2$. The Hartree--Fock energy $E_N^\text{HF}$, obtained minimizing (\ref{eq:HFfun}) over all rank-$N$ orthogonal projections $\omega_N$, provides a good approximation for the ground state energy of (\ref{eq:HNmf}) as $N\gg 1$; see \cite{Bach,GS} for the case of Coulomb systems. Recently, the large $N$ asymptotics of the correlation energy, defined as the difference between the many-body ground state energy minus the Hartree--Fock ground state energy, has been determined in \cite{HPR,BNPSSa,BNPSSb,BPSScorr,CHN}.  

It is natural to ask what happens when the external traps are switched off; does the Hartree--Fock theory also describe the resulting many-body evolution $\psi_{N,t}^\text{mf} = e^{-i H_N^\text{mf} (0) t/\varepsilon} \psi_N$, generated by the translation invariant Hamiltonian $H_N^\text{mf} (0)$? Here, the presence of the parameter $\varepsilon = N^{-1/3}$ guarantees that $\psi^\text{mf}_{N,t}$ undergoes macroscopic changes for times $t$ of order one. The convergence of the one-particle reduced density matrix 
\begin{equation}\label{eq:gamma1}
\gamma_{N,t}^{(1)} = N \tr_{2,\dots ,N} |\psi_{N,t} \rangle \langle \psi_{N,t} |
\end{equation}
associated with $\psi^\text{mf}_{N,t}$ towards the solution of the time-dependent Hartree--Fock equation 
\begin{equation}\label{eq:HFtd}  i \varepsilon \partial_t \omega_{N,t} = \Big[ -\varepsilon^2 \Delta + (V * \rho_t) -X_t , \omega_{N,t} \Big] \end{equation} 
has been established in \cite{EESY}, for analytic potentials and for short times. Here $\rho_t (x) = N^{-1} \omega_{N,t} (x;x)$ is the density associated with $\omega_{N,t}$ and $X_t$ is the exchange operator, defined by its integral kernel $X_t (x;y) = N^{-1} V(x-y) \omega_t (x;y)$. More recently, in \cite{BPS} (and later in \cite{PP}, following a different approach), this convergence has been generalized to a much larger class of interaction potentials and to all times. This result and the techniques that were used to derive it provide the starting point for the present work. The initial data considered in \cite{BPS} are assumed to satisfy suitable semiclassical estimates, which appear as a natural characterization of trapped equilibrium states, in the mean-field regime. Furthermore, in \cite{BPS} it is also shown that, for bounded potentials, the exchange term in (\ref{eq:HFtd}) is subleading, compared with the direct term, and that the many-body dynamics can also be approximated by the Hartree equation 
\begin{equation}\label{eq:Ha0} i \varepsilon \partial_t \omega_{N,t} = \big[ -\varepsilon^2 \Delta + (V * \rho_t) , \omega_{N,t} \big]\;. \end{equation} 
Notice that the result of \cite{BPS} holds in the sense of convergence of density matrices; convergence in $L^{2}$-norm for homogeneous Fermi gases has been recently obtained in \cite{BNPSS}, via the rigorous bosonization techniques developed in \cite{BNPSSa,BNPSSb,BPSScorr} (in this case, $\omega_N$ is translation invariant which implies, in particular, that $\omega_{N,t} = \omega_N$ is stationary).

The result of \cite{BPS} has been extended to fermions with relativistic dispersion (known as pseudo-relativistic fermions) in \cite{BPS2} and to quasi-free mixed states in \cite{BJPSS}. See also \cite{BPSbook} for a review. All these works consider bounded interaction potentials. As for unbounded potentials, the time-dependent Hartree--Fock equation for particles interacting through a Coulomb potential has been derived in \cite{PRSS}, under the assumption that a suitable semiclassical structure of the initial datum propagates along the flow of the Hartree--Fock equation. Recently, the propagation of the semiclassical structure has been proven in \cite{CLS}, for mixed states and for a class of singular potentials that includes a suitably regularized version of the Coulomb interaction. In the absence of semiclassical scaling, that is, setting $\varepsilon = 1$ in the previous discussion, convergence to the time-dependent Hartree--Fock equation has been shown in \cite{BGGM} for bounded potentials, and then extended to Coulomb potentials in \cite{FK} (see also \cite{BBPPT}).

Notice that both the Hartree--Fock equation (\ref{eq:HFtd}) and the Hartree equation (\ref{eq:Ha0}) still depend on the number of particles $N$. In the limit $N \to \infty$, the Hartree--Fock and the Hartree dynamics are known to converge to the Vlasov equation, a classical effective evolution equation. The first proof of convergence from the quantum many-body dynamics to the Vlasov dynamics has been obtained in \cite{NS} for analytic potentials, and then extended in \cite{Sp} to a much larger class of interactions. Next, convergence from the Hartree--Fock to the Vlasov equation has been proved in \cite{LP, MM}. All these results hold in a weak sense. Bounds on the rate of convergence from the Hartree--Fock equation to the Vlasov equation have been first obtained in \cite{APPP}, and more recently in \cite{BPSS} for a larger class of initial data and of interaction potentials. Unbounded interaction potentials, including the Coulomb interaction, have been considered in \cite{LS}. Finally, let us mention the result \cite{LeSa}, where convergence from the Hartree equation to the Vlasov equation is proven for local perturbations of the equilibrium state of extended Fermi gases at fixed density, in the high density regime. This last setting will be related to the one considered in the present work.

The results described above (with the exception of \cite{LeSa}) apply to the mean-field limit, where particles are initially trapped in a volume of order one. To describe the physically important case of extended gases, let us now consider $N$ fermions moving in a large region $\Lambda \subset \bR^3$, at high density $\varrho = N / |\Lambda| \gg 1$. If the potential has range of order one, each particle interacts, at any given time, with order $\varrho$ other particles. Furthermore, the kinetic energy of the $N$ particles is now of the order $\varrho^{2/3} N$. Therefore, to describe an extended Fermi gas at high density, we consider the Hamilton operator 
\begin{equation}\label{eq:HNdefintro}
H_{N} = \sum_{j=1}^{N} - \varepsilon^{2}\Delta_{j} + \varepsilon^{3} \sum_{i<j}^{N} V(x_{i} - x_{j})\qquad \text{on $L^{2}_{\text{a}}(\mathbb{R}^{3N})$}
\end{equation}
where we set $\varepsilon = \varrho^{-1/3}$ to make sure that both terms are of order $N$. In contrast with the mean-field regime (where we had $\varepsilon = N^{-1/3}$), $\varepsilon$ is now small but independent of $N$. We will be interested in the many-body evolution governed by the Schr\"odinger equation
\begin{equation}\label{eq:schr0} i \varepsilon\partial_t \psi_{N,t} = H_N \psi_{N,t} \end{equation} 
for initial data $\psi_{N,0} = \psi_N$ that are close (in an appropriate sense) to a Slater determinant with reduced density matrix $\omega_N$, describing a quasi-free state of $N$ fermions in a large domain $\Lambda \subset \bR^3$, with density of particles bounded everywhere by $\varepsilon^{-3}$ (up to a multiplicative constant), in the sense that 
\begin{equation}\label{eq:noblowup} 
\sup_{z\in \mathbb{R}^{3}} \int dy\, \frac{\omega_N (y;y)}{1 + |y - z|^{4}} \leq C\varepsilon^{-3}\;.
\end{equation}
Additionally, we will assume the initial data to exhibit a local semiclassical structure, captured by localized commutator bounds for $\omega_N$ with the position operator and with the momentum operator, that are expected to hold true  for equilibrium states of confined systems at (or close to) zero temperature and at density of order $\varepsilon^{-3}$. 

The present setting can also be viewed as a {\it Kac regime}, see \cite{EP} for a review. In this situation, the density of particles is order $1$, the range of the potential is $\gamma \gg 1$ uniformly in the size of the system, and the strength of the potential is fixed so that its $L^{1}$-norm is order $1$. Therefore, the many-body Hamiltonian is:
\begin{equation*}
H_{N}^{\text{Kac}} = \sum_{j=1}^{N} -\Delta_{j} + \gamma^{-3} \sum_{i<j}^{N} V((x_{i} - x_{j})/\gamma)\;.
\end{equation*}
We look at the evolution of the system for times $\tau = O(\gamma)$, so that every particle covers a distance of the order of the range of the interaction potential. Thus, writing $\tau = \gamma t$, the time evolution of the system is described by the Schr\"odinger equation:
\begin{equation}\label{eq:kacevo}
i \gamma^{-1} \partial_{t} \psi_{N,t} = H_{N}^{\text{Kac}} \psi_{N,t}\;.
\end{equation}
Let us denote by $U_{\gamma}$ the unitary operator on $L^{2}(\mathbb{R}^{3N})$ implementing the space rescaling: $U_{\gamma} \psi_{N}(x_{1}, \ldots, x_{N}) = \gamma^{3N/2} \psi_{N}(\gamma x_{1}, \ldots, \gamma x_{N})$. Applying the transformation to both sides of (\ref{eq:kacevo}), we have:
\begin{equation}
\begin{split}\label{eq:resckac}
i \gamma^{-1} \partial_{t} U_{\gamma} \psi_{N,t} &= U_{\gamma} H_{N}^{\text{Kac}} U_{\gamma}^{*} U_{\gamma}\psi_{N,t} \\
&= \Big( \sum_{j=1}^{N} -\gamma^{-2} \Delta_{j} + \gamma^{-3} \sum_{i<j}^{N} V(x_{i} - x_{j}) \Big) U_{\gamma}\psi_{N,t}\;,
\end{split}
\end{equation}
where now $U_{\gamma}\psi_{N,t}$ is a wave function describing a quantum system with density $O(\gamma^{3})$. Thus, the dynamics (\ref{eq:resckac}) is equivalent to the one generated by (\ref{eq:HNdefintro}), after setting $\varepsilon = \gamma^{-1}$.

Although (\ref{eq:HNdefintro}) does not describe a mean-field regime (particles typically interact with $\varepsilon^{-3}$ other particles and the size of the potential is of the order $\varepsilon^3$, with $\varepsilon$ now independent of $N$), for small $\varepsilon > 0$ we can still expect a local averaging mechanism to take place and thus that the many-body dynamics (\ref{eq:schr0}) can be approximated by the time-dependent Hartree equation 
\begin{equation}\label{eq:Hintro}
i\varepsilon \partial_{t} \omega_{N,t} = [ -\varepsilon^{2} \Delta + (V *\rho_{t}), \omega_{N,t} ] \, , 
\end{equation}
with $\omega_{N,0} = \omega_N$ and where now
\begin{equation*}
\rho_{t}(x) := \varepsilon^{3} \omega_{N,t}(x;x)\;.
\end{equation*}
Motivated by the results of \cite{BPS}, we are neglecting here the exchange term, appearing on the r.h.s.~of (\ref{eq:HFtd}), since it is expected to be small for the class of smooth interaction potentials that we are going to consider. In our main theorem we compare the one-particle reduced density matrix $\gamma^{(1)}_{N,t}$ associated with the solution of (\ref{eq:schr0}) with the solution of (\ref{eq:Hintro}) and we show that 
\begin{equation}\label{eq:maininf}
\frac{\| \gamma^{(1)}_{N,t} - \omega_{N,t} \|_{\text{HS}}}{N^{1/2}} \leq C\varepsilon^{1/2}\;,
\end{equation}
for short macroscopic times $t$ of order $1$ in $\varepsilon$, and for a constant $C$ independent of $\varepsilon$ and of $N, |\Lambda|$. This result should be compared with the trivial estimates 
\begin{equation*}
\| \gamma^{(1)}_{N,t}\|_{\text{HS}} \leq N^{1/2},\qquad  \| \omega_{N,t} \|_\text{HS} = N^{1/2} \,.
\end{equation*}
It is important to notice that the rate of convergence, on the r.h.s.~of (\ref{eq:maininf}), only depends on the parameter $\varepsilon$, and not on the volume of the system: in particular, our theorem applies to the setting in which the limit $|\Lambda| \to \infty$ is taken {\it before} the limit $\varepsilon \to 0$. To the best of our knowledge, this is the first derivation of the time-dependent Hartree equation for an interacting, extended Fermi gas (notice, however, that the dynamics of tracer particles or impurities moving through an extended ideal gas has been considered in \cite{DFPP,MP}).

Furthermore, we also consider the case of massive pseudo-relativistic fermions, evolving with Hamiltonian:
\begin{equation}\label{eq:Hrel}
H_{N}^{\text{rel}} = \sum_{j=1}^{N} \sqrt{1 - \varepsilon^{2} \Delta_{j}} + \varepsilon^{3} \sum_{i<j}^{N} V(x_{i} - x_{j})\;.
\end{equation}
In this case, the relevant effective evolution equation is the pseudo-relativistic Hartree equation,
\begin{equation}
i\varepsilon \partial_{t} \omega_{N,t} = \big[ \sqrt{1 - \varepsilon^{2}\Delta} + (V *\rho_{t}), \omega_{N,t} \big] \;.
\end{equation}
For mean-field fermions, the validity of the pseudo-relativistic Hartree equation has been proved in \cite{BPS2}. Here, we show the validity of the bound (\ref{eq:maininf}) for extended pseudo-relativistic fermions, for all times.

Technically, the main challenge we have to face consists in showing that the nonlinear Hartree equation, both in the non-relativistic and in the relativistic case, preserves the local semiclassical structure assumed on the initial data $\omega_N$. The fact that the commutator bounds for $\omega_{N,t}$ are localized in space makes the proof of propagation of the semiclassical structure much more involved than in the mean-field regime. In particular, to achieve this we need to propagate the bound (\ref{eq:noblowup}) on the density of particles along the solution of the Hartree equation. For non-relativistic fermions we are able to propagate this estimate for small times of order $1$ in $\varepsilon$. For the pseudo-relativistic case, on the other hand, we can take advantage of the boundedness of the group velocity of the particles to show that (\ref{eq:noblowup}) remains true for all times (of order $1$ in $\varepsilon$). 

The article is organized as follows. In Section \ref{sec:main} we state our main results, Theorem~\ref{thm:main} (non-relativistic case) and Theorem \ref{thm:mainrel} (relativistic case). Both theorems are stated for initial data satisfying the estimates of Assumption \ref{ass:sc}. These assumptions are verified for the free Fermi gas and for coherent states in Appendix \ref{app:localsc}. In Section \ref{sec:fock} we introduce the basic tools of our analysis, namely the fermionic Fock space and Bogoliubov transformations. In Section \ref{sec:proofmain} we prove our main result; the proof is based on the adaptation of the method of \cite{BPS} to extended systems, which crucially relies on the propagation of the local semiclassical structure along the flow of the Hartree equation, as stated in Theorem \ref{thm:propcomm} for the non-relativistic case. The proof of Theorem \ref{thm:propcomm} is given in 
Section \ref{sec:propsc}. Finally, in Section \ref{sec:proprel} we extend the propagation of the local semiclassical structure to the pseudo-relativistic case.

\section{Main Result}\label{sec:main}

Let $\Lambda \subset \R^{3}$ denote a Lebesgue measurable domain of volume $|\Lambda|$, such that $\mathrm{diam}(\Lambda) \leq C |\Lambda|^{1/3}$. We consider an initial wave function $\psi_N \in L^{2}_{\text{a}}(\mathbb{R}^{3N})$, $\| \psi_{N} \|_{2} = 1$, describing $N$ particles localized in $\Lambda$ (in a sense that will be made precise below). We set $\veps = \varrho^{-1/3}$, with the average density $\varrho = N / |\Lambda|$. We are interested in the high-density regime, where $\varrho \gg 1$ or, equivalently, $\veps \ll 1$ (independently of $N, |\Lambda|$). 

Let us denote by $\gamma^{(1)}_{N}$ the reduced one-particle density matrix of $\psi_{N}$, with the integral kernel 
\begin{equation}\label{eq:gamma11}
\gamma^{(1)}_{N}(x;y) = N \int dx_{2}\ldots dx_{N}\, \psi_{N} (x, x_2, \dots , x_N) \overline{\psi_N (y, x_2, \dots , x_N)}\;.
\end{equation}
In the high-density regime, the ground state of the system is expected to be well approximated by a Slater determinant 
\begin{equation}\label{eq:Slaterdef}
\psi_{\text{Slater}} = f_{1} \wedge \cdots \wedge f_{N}\;,\qquad \langle f_{i}, f_{j} \rangle = \delta_{ij}\;,
\end{equation}
for a suitable choice of orthonormal functions $f_{i}$. The closeness of $\psi_{N}$ to a Slater determinant is expressed in terms of closeness of the reduced one-particle density matrix. The reduced one-particle density matrix of the Slater determinant (\ref{eq:Slaterdef}) is given by a rank-$N$ orthogonal projector 
\begin{equation*}
\omega_{N} = \sum_{i=1}^{N} |f_{i}\rangle \langle f_{i} |\;,
\end{equation*}
and we shall suppose that 
\begin{equation*}
\frac{1}{N} \| \gamma^{(1)}_{N} - \omega_{N} \|_{\text{HS}}^{2} \ll 1 \;.
\end{equation*}
We shall now introduce some important assumptions on $\omega_{N}$, which are expected to hold for a wide class of confined systems at equilibrium. We shall check them in Appendix \ref{app:localsc}, for the free Fermi gas and for coherent states. Strictly speaking, coherent states do not really fit our setting, since they are not projections. However, in Appendix \ref{app:localsc} we will consider a family of coherent states that can be viewed as approximate projections. 

For any $z\in \mathbb{R}^{3}$, any $t\in \mathbb{R}$ and any $n\in \mathbb{N}$, let us define the localization operator as:
\begin{equation}\label{eq:defW}
\W_{z}^{(n)}(t):= \frac{1}{1+|\hat{x}(t) - z|^{4n}} \;, \qquad \W_{z}^{(n)}:=\W_{z}^{(n)}(0)\;,
\end{equation}
where $\hat{x}(t)$ denotes the free evolution of the position operator
\begin{equation}\label{free_evolution_x}
\hat{x}(t) =\ee^{-i \veps t \Delta}\hat{x} \ee^{+i \veps t \Delta} = \hat{x}  - 2i \veps t \nabla \;.
\end{equation}
Let us also introduce the weight function:
\begin{equation}\label{eq:chilambda}
X_{\Lambda}(z):= 1+\mathrm{dist}(z,\Lambda)^{4}\;,
\end{equation}
with $\text{dist}(x,\Lambda) = \inf_{y\in \Lambda} | x - y |$. Next, we collect the assumptions we shall make on the reference Slater determinant.

\begin{assumption}[Assumptions on the Initial Datum]\label{ass:sc} There exist $n\in \mathbb{N}$ and $T_{1}\geq 0$ such that the following holds true:
\begin{equation}
\sup _{t \in [0,T_{1}]} \sup_{z\in \mathbb{R}^{3}} \sup_{\substack{p\in\mathbb{R}^{3} \\ |p| \leq \varepsilon^{-1}}} \frac{X_{\Lambda}(z)}{1+|p|} \big \| \W_{z}^{(n)}(t) \big[ e^{i p \cdot \hat{x}},\omega_{N}\big]\big\|_{\tr} \leq C \veps^{-2} \;,
\label{eq:propcomm1}
\end{equation}
and
\begin{equation}
\sup _{t \in [0,T_{1}]} \sup_{z\in \R^{3}} X_{\Lambda}(z) \big \| \W_{z}^{(n)}(t) \big[ \veps \nabla,\omega_{N}\big] \big\|_{\tr}  \leq C \veps^{-2}\;.
\label{eq:propcomm2}
\end{equation}
Furthermore,
\begin{equation}\label{eq:H2}
\sup _{t \in [0,T_{1}]} \sup _{z \in \R^{3}} X_{\Lambda}(z) \big\|\W_{z}^{(n)}(t)  \, \omega_{N} \big\|_{\tr}  \leq C \veps^{-3} \;,
\end{equation}
and
\begin{equation}
\label{eq:assump_density}
\sup_{t\in [0,T_{1}]} \sup_{z\in \mathbb{R}^{3}} \tr\, \mathcal{W}^{(1)}_{z}(t) \omega_{N} \leq  C\varepsilon^{-3}\;.
\end{equation}
\end{assumption}
We shall refer to the first two estimates in Assumption \ref{ass:sc} as the local semiclassical structure of the initial datum.

We are now ready to state our main result. We shall separate the cases of non-relativistic and pseudo-relativistic fermions.
\begin{theorem}[Main Result: Non-Relativistic Case] \label{thm:main} Let $\omega_{N}$ be a rank-$N$ orthogonal projector on $L^{2}(\mathbb{R}^{3})$, satisfying Assumption \ref{ass:sc} for some $n\in \mathbb{N}$ and $T_{1} > 0$. Suppose that $V \in L^{1}(\R^{3})$ is such that
\begin{equation}
\label{eq:assV}
\sup _{\alpha: |\alpha| \leq 8 n}
\int _{\R^{3}} d p \, (1+|p|^{\mathrm{max}(4n,7)})\, \big| \partial_{p}^{\alpha} \widehat{V}(p)\big| < \infty \;.
\end{equation}
Let $\psi_{N} \in L^{2}_{\text{a}}(\mathbb{R}^{3N})$, such that:
\begin{equation}\label{eq:trdatum}
\| \gamma^{(1)}_{N} - \omega_{N} \|_{\tr} \leq C\varepsilon^{\delta} N \;,\qquad \text{for some $\delta >0$.}
\end{equation}
Let $\psi_{N,t} = e^{-iH_{N} t/\varepsilon} \psi_{N}$, with $H_{N}$ given by (\ref{eq:HNdefintro}), and let $\gamma^{(1)}_{N,t}$ be the reduced one-particle density matrix of $\psi_{N,t}$. Let $\omega_{N,t}$ be the solution of the time-dependent Hartree equation:
\begin{equation}\label{eq:H}
i\varepsilon \partial_{t} \omega_{N,t} = [ -\varepsilon^{2}\Delta + \rho_{t} * V, \omega_{N,t} ]\;,\qquad \omega_{N,0} = \omega_{N}\;.
\end{equation}
%
%
%
Then, there exist $0< T < T_1$ and $C > 0$, independent of $\veps$ and $N$ such that, for all $t\in [0,T]$:
\begin{equation}\label{eq:mainclaim}
\| \gamma^{(1)}_{N,t} - \omega_{N,t} \|_{\mathrm{HS}} \leq C \max\{ \varepsilon^{\frac{1}{2}}, \varepsilon^{\frac{\delta}{2}} \} N^{\frac{1}{2}}\;.
\end{equation}
%
\end{theorem}
\begin{remark}\label{rmk:sizes_thm_one}
\begin{itemize}
\item[(i)] The $N$-dependence of the estimates (\ref{eq:trdatum}), (\ref{eq:mainclaim}) is the natural one. In particular, these estimates allow to quantify the closeness of the expectation values of extensive operators.

\item[(ii)] As discussed in the introduction, the result should be compared with the trivial estimates $\| \gamma_{N,t}^{(1)} \|_{\mathrm{HS}} \leq N^{1/2}$, $\|\omega_{N,t}\|_{\mathrm{HS}} = N^{1/2}$. Normalizing the Hilbert-Schmidt norm by $|\Lambda|^{1/2}$, and recalling that $N = \varrho |\Lambda|$ with $\varrho = O(\varepsilon^{-3})$, our main theorem proves convergence of the many-body evolution towards the nonlinear Hartree equation as the density goes to infinity, uniformly in the system size $|\Lambda|$.

\item[(iii)] We expect the Hartree--Fock equation to give a better approximation of the many-body quantum dynamics. However, the difference between the Hartree and the Hartree-Fock dynamics is smaller than the error 
on the r.h.s.~of (\ref{eq:mainclaim}) and thus it cannot be resolved with our present techniques. 
\end{itemize}
\end{remark}
The time $T>0$ appearing in Theorem \ref{thm:main} is related to the validity of a suitable non-concentration estimate for the solution of the time-dependent Hartree equation, which quantifies the number of particles in bounded regions of space. We prove this bound in Proposition \ref{prp:density} for short macroscopic times of order $1$ in $\varepsilon$.

Next, let us consider pseudo-relativistic fermions. In this case, we only need assumptions on norms of the initial projector $\omega_N$ and of its commutators, multiplied with the multiplication operator $\W_z^{(n)}$ (in the non-relativistic case, the assumptions involved instead the free evolution $\W_z^{(n)} (t)$ of $\W_z^{(n)}$). For this reason, in the next theorem we will only require Assumption \ref{ass:sc} to hold with $T_1=0$ (which implies $t=0$ in \eqref{eq:propcomm1} - \eqref{eq:assump_density}). The other important difference, compared with the non-relativistic case, is that, thanks to the boundedness of the group velocity of the particles, we can establish convergence towards Hartree dynamics, for all fixed times $t \in \bR$ (rather than only for short times).  
\begin{theorem}[Main Result: Pseudo-Relativistic Case]\label{thm:mainrel}
Let $\omega_{N}$ be a rank-$N$ orthogonal projector on $L^{2}(\mathbb{R}^{3})$, satisfying Assumption \ref{ass:sc} for some $n \in \mathbb{N}$ and for $T_{1} = 0$. Assume that $V \in L^{1}(\R^{3})$ is such that
\begin{equation}
\label{eq:assV-ps}
\sup _{\alpha: |\alpha| \leq 8 n}
\int _{\R^{3}} d p \, (1+|p|^{7})\, \big| \partial_{p}^{\alpha} \widehat{V}(p)\big| < \infty \;.
\end{equation}
%
%
Let $\psi_{N}$ be as in Theorem \ref{thm:main}, let $\psi_{N,t} = e^{-iH_{N}^{\text{rel}} t/\varepsilon} \psi_{N}$, with $H_{N}^{\text{rel}}$ given by Eq.~(\ref{eq:Hrel}),
and let $\gamma^{(1)}_{N,t}$ be the reduced one-particle density matrix of $\psi_{N,t}$. Let $\omega_{N,t}$ be the solution of the time-dependent pseudo-relativistic Hartree equation:
\begin{equation}
i\varepsilon \partial_{t} \omega_{N,t} = \big[ \sqrt{1-\varepsilon^{2}\Delta} + \rho_{t} * V, \omega_{N,t} \big]\;,\qquad \omega_{N,0} = \omega_{N}\;.
\end{equation}
Then, for all $t\in \mathbb{R}$:
\begin{equation}\label{eq:mainclaim2}
\| \gamma^{(1)}_{N,t} - \omega_{N,t} \|_{\mathrm{HS}} \leq C\exp(\exp Ct) \max\{ \varepsilon^{\frac{1}{2}}, \varepsilon^{\frac{\delta}{2}} \} N^{\frac{1}{2}}\;.
\end{equation}
\end{theorem}

The rest of the paper is devoted to the proof of Theorems \ref{thm:main}, \ref{thm:mainrel}. The proof of the two theorems are very similar, except for the propagation of the local semiclassical structure. To fix notations, we work with non-relativistic particles. Only in Section \ref{sec:proprel} we come back to the pseudo-relativistic case. 
In Section \ref{sec:fock} we introduce the Fock space formalism, which will allow us to efficiently describe the fluctuation of the many-body evolution around the nonlinear effective dynamics as the creation and annihilation of particles around a suitable time-dependent state. The proof of Theorem \ref{thm:main} is given in Section \ref{sec:proofmain}. It relies on a bound for the growth of the number of fluctuations, proven in Proposition \ref{prp:fluct}. In turn, this result crucially relies on the propagation of the local semiclassical structure along the flow of the Hartree equation. This is established in Section \ref{sec:propsc} for non-relativistic fermions and in Section \ref{sec:proprel} for pseudo-relativistic fermions.

\section{Fock Space Representation}\label{sec:fock}

To prove Theorem \ref{thm:main}  we switch to a Fock space formulation of the problem.

\subsection{Second Quantization}
We define the fermionic Fock space $\mathcal{F}$ over $L^{2}(\mathbb{R}^{3})$ as:
\begin{align*}
\notag
\F = \mathbb{C} \oplus \bigoplus_{n\geq 1} \mathcal{F}^{(n)}\;,\qquad \F^{(n)}:= L_{\text{a}}^2(\mathbb{R}^{3n}) \;.
\end{align*}
Vectors in the Fock space correspond to infinite sequences of functions $(\psi^{(n)})_n$ with $\psi^{(n)} \in L_{\text{a}}^2(\mathbb{R}^{3n})$. A simple example is the vacuum state, $\Omega = (1, 0, \ldots, 0, \ldots)$.

Given $\psi \in \mathcal{F}$, $\psi = (\psi^{(0)}, \psi^{(1)}, \ldots, \psi^{(n)}, \ldots)$, we define the scalar product:
\begin{equation*}
\langle \psi_{1}, \psi_{2} \rangle = \sum_{ n \geq 0 } \langle \psi_{1}^{(n)},\psi_{2}^{(n)}  \rangle_{L^{2}(\R^{3n})} \;.
\end{equation*}
Equipped with this natural inner product, the Fock space $\mathcal{F}$ is a Hilbert space. We henceforth denote by $\| \cdot \|$ the norm induced by this inner product, and with a slight abuse of notation we shall use the same notation to denote the operator norm of linear operators acting on $\mathcal{F}$.

It is convenient to introduce creation and annihilation operators, acting on $\mathcal{F}$. Let $f\in L^{2}(\mathbb{R}^{3})$. We define the creation operator $a^{*}(f)$ and the annihilation operator $a(f)$ as
\begin{equation*}
\begin{split} 
\left(a^*(f) \psi\right)^{(n)}(x_1,\ldots, x_n) & := \frac{1}{\sqrt{n}} \sum_{j=1}^n (-1)^{j-1} f(x_j)\psi^{(n-1)}(x_1, \ldots,x_{j-1}, x_{j+1}, \ldots, x_n) \\
\left(a(f)\psi\right)^{(n)}(x_1,\ldots, x_n)&:= \sqrt{n+1} \int_{\R^{3}} d x \; \overline{f(x)} \, \psi^{(n+1)} (x, x_1,\ldots, x_n)\;,
\end{split}
\end{equation*}
for any $\psi \in \mathcal{F}$. 
From a physics viewpoint, the creation operator creates a particle with wave function $f$, while the annihilation operator annihilates a particle with wave function $f$. These definitions are supplemented by the requirement $a(f) \Omega = 0$.

It is not difficult to see that $a^{*}(f) = a(f)^{*}$. Also, the creation and annihilation operators satisfy the canonical anticommutation relations:
\begin{equation*}
\{ a(f), a(g) \} = \{ a^{*}(f), a^{*}(g) \} = 0\;,\qquad \{ a(f), a^{*}(g) \} = \langle f, g \rangle_{L^{2}(\mathbb{R}^{3})}\;.
\end{equation*}
These relations imply that $\|a(f)\| \leq \|f\|_{2}$, $\|a^{*}(f)\|\leq \|f\|_{2}$ (it is not difficult to see that these bounds are sharp, that is the norms of the operators are actually equal to $\|f\|_{2}$). In the following sections, we will also make use of operator-valued distributions, e.g., $a_x^*$ and $ a_x$ for $x \in \R^{3}$, such that
\begin{equation*}
a^* (f) = \int_{\R^{3}} d x \, f(x)\, a_x^* , \quad a(f) = \int_{\R^{3}} d x  \, \overline{f (x)}\, a_x\,. 
\end{equation*}
The creation and annihilation operators can be used to define the second quantization of observables. For instance, consider the number operator $\mathcal{N}$, acting on a given Fock space vector as $(\mathcal{N} \psi)^{(n)} = n \psi^{(n)}$. In terms of the operator-distributions, it can be written as:
\begin{equation*}
\cN = \int_{\R^{3}} d x \, a_x^* a_x \;.
\end{equation*}
More generally, for a given one-particle operator $J$ on $L^{2}(\mathbb{R}^{3})$ we define its second quantization $d\Gamma(J)$ as the operator on the Fock space acting as follows:
\begin{equation*}
 d\Gamma (J) \upharpoonright_{\F^{(n)}} = \sum_{j=1}^n J^{(j)}
\end{equation*}
where $J^{(j)} = \mathbf{1}^{\otimes (n-j)} \otimes J \otimes \mathbf{1}^{\otimes (j-1)}$. If $J$ has the integral kernel $J (x;y)$, we can write $d\Gamma (J)$ as:
\begin{equation*}
d\Gamma (J) = \int_{(\R^{3})^{2}} d x d y \, J (x;y) a_x^* a_y\;.
\end{equation*}
In the next lemma we collect some bounds for the second quantization of one-particle operators, that will play an important role in the proof of our main result. Their proofs can be found in \cite[Lemma 3.1]{BPS}.
\begin{lemma}
\label{lemma:bounds}
Let $J$ be a bounded operator on $L^{2}(\mathbb{R}^{3})$. We have, for any $\psi \in \mathcal{F}$:
\begin{equation*}
|\langle \psi, d\Gamma (J) \psi \rangle| \leq \| J \|_{\mathrm{op}} \langle \psi, \cN \psi \rangle \;, \qquad \left\|d\Gamma (J) \psi \right\| \leq \left\| J \right\|_{\mathrm{op}} \left\| \cN \psi\right\|\;.
\end{equation*}
Let $J$ be a Hilbert-Schmidt operator. We then have, for any $\psi \in \mathcal{F}$:
\begin{equation*}
\begin{split} 
\left\|d\Gamma (J) \psi \right\| &\leq \left\| J \right\|_{\mathrm{HS}} \left\| \mathcal{N}^{1/2} \psi \right\| \\
\left\| \int_{(\R^{3})^{2} } d x d x'\, J (x; x') a_x a_{x'} \psi \right\| & \leq \| J \|_{\mathrm{HS}} \left\| \mathcal{N}^{1/2} \psi \right\| \\
\left\| \int_{(\R^{3})^{2}} d x d x'\, J (x; x') a^*_x a^*_{x'} \psi \right\| & \leq 2\| J \|_{\mathrm{HS}} \left\| (\mathcal{N}+1)^{1/2} \psi \right\|\;. 
\end{split}
\end{equation*}
Let $J$ be a trace class operator. We then have, for any $\psi \in \mathcal{F}$:
\begin{equation*}
\begin{split}
\left\|d\Gamma (J) \psi \right\| &\leq 2  \left\| J \right\|_{\mathrm{tr}}\|\psi\| \\
\left\| \int_{(\R^{3})^{2}} d x d y\, J (x; x') a_x a_{x'} \psi \right\| & \leq 2  \left\| J \right\|_{\mathrm{tr}} \|\psi\| \\
\left\| \int_{(\R^{3})^{2}} d x d y\,  J (x; x') a^*_x a^*_{x'} \psi \right\| & \leq 2  \left\| J \right\|_{\mathrm{tr}} \|\psi\|  \;.
\end{split}
\end{equation*}
\end{lemma}
Given a Fock space vector $\psi \in \mathcal{F}$, we define its reduced one-particle density matrix as the non-negative trace class operator $\gamma^{(1)}_{\psi}$ on $L^{2}(\mathbb{R}^{3})$ with integral kernel: 
\begin{equation}
\label{eq:gamma-FS} 
\gamma_\psi^{(1)} (x;y) = \langle \psi, a_y^* a_x \psi \rangle \;.
\end{equation}
If $\psi$ is an $N$-particle state, it is not difficult to check that this definition agrees with (\ref{eq:gamma11}). Furthermore, given a one-particle observable $J$, we have:
\begin{equation}
\label{expectation_one_body}
\langle \psi,d\Gamma (J) \psi \rangle = \int_{(\R^{3})^{2}} d x d y \, J (x;y) \, \langle \psi, a_x^* a_y \psi \rangle = \tr \, J \gamma^{(1)}_\psi\;,
\end{equation}
an identity which motivates the definition \eqref{eq:gamma-FS}. In particular, $\tr \, \gamma^{(1)}_\psi = \langle \psi , \cN \psi \rangle$ is the expected number of particles in $\psi$. 

Next, we lift the many-body Hamiltonian to the Fock space, as follows. We define the second quantization of $H_{N}$ as $\mathcal{H}_{N} = 0 \oplus \bigoplus_{n\geq 1} \mathcal{H}_{N}^{(n)}$, where
\begin{equation*}
\mathcal{H}_{N}^{(n)} = \sum_{i=1}^{n} -\varepsilon^{2} \Delta_{i} + \varepsilon^{3} \sum_{i<j}^{n} V(x_{i} - x_{j})\;.
\end{equation*}
In terms of the operator-valued distributions, we can write:
\begin{equation}
\label{eq:FocKHN} 
\cH_{N} = \veps^2 \int_{\R^{3}} d x \,  \nabla  a_{x}^{*} \nabla a_{x} + \frac{\veps^{3}}{2} \int_{(\R^{3})^{2}} d x d y \, V(x-y) a_{x}^{*} a_{y}^{*} a_{y} a_{x} \;.
\end{equation}
The time evolution of a state in the Fock space is defined as $\psi_{t} = e^{-i\mathcal{H}_{N} t / \varepsilon} \psi$. On $N$-particle states, this coincides with the solution of the Schr\"odinger equation (\ref{eq:schr0}). 
\subsection{Bogoliubov Transformations}\label{sec:Bogtra}
The Fock-space representation of the problem is particularly convenient also in view of the representation of Slater determinants via Bogoliubov transformations, defined here. 

Given an orthonormal family $(f_{i})_{i=1}^{N}$, the corresponding Slater determinant can be represented in Fock space as follows:
\begin{equation*}
a^{*}(f_{1}) \cdots a^{*}(f_{N}) \Omega = \Big( 0, 0, \ldots, 0, f_{1}\wedge \cdots \wedge f_{N}, 0, \ldots \Big) 
\end{equation*}
where the only nontrivial entry is the $N$-th. A crucial fact for our analysis, as was the case in previous work starting from \cite{BPS}, is the existence of a unitary operator $R_{\omega_{N}} : \F \to \F$ with the following two properties: 
\begin{equation*}
R_{\omega_{N}} \Omega = a^* (f_1) \dots a^* (f_N) \Omega
\end{equation*}
and, for any $g \in L^{2}(\R^{3})$,
\begin{equation}
\label{eq:bog} 
R^*_{\omega_{N}} a (g) R_{\omega_{N}} = a (u_{N} g) + a^* (\overline{v}_{N} \overline{g})
\end{equation}
where $u_{N} = 1- \omega_{N}$ and $v_{N} = \sum_{j=1}^N | \overline{f_j} \rangle \langle f_j|$. Equivalently, consider an orthonormal basis $(f_{i})_{i\geq 1}$ of $L^{2}(\mathbb{R}^{3})$ obtained completing the orthonormal family $f_{1}, \ldots, f_{N}$ in an arbitrary way. We have:
\begin{equation}\label{eq:Bogo}
\begin{split}
 R^*_{\omega_{N}} a(f_j) R_{\omega_{N}}= \begin{cases}
 a (f_j) & \text{if } \quad j >N \\
 a^{*}(f_{j})  & \text{if } \quad j  \leq N\;.
\end{cases}
\end{split}
\end{equation}
It follows that $R_{\omega_{N}}^{*} = R_{\omega_{N}} = R_{\omega_{N}}^{-1}$. The map $R_{\omega_{N}}$ is known as Bogoliubov transformation, and it acts as a particle-hole transformation. It allows us to switch to a new representation of the system, where the new vacuum is the Slater determinant associated to the reduced density $\omega_{N}$. The new creation operators, given by (\ref{eq:Bogo}), create excitations around the Slater determinant, which are either particles outside the determinant or holes in it. The proof of the existence of the unitary operator $R_{\omega_{N}}$ with the properties listed above can be found, for example, in \cite{Solovej}.

More generally, for every $t \in \bR$, we can associate a Bogoliubov transformation $R_{\omega_{N,t}}$ to the solution of the time-dependent Hartree equation (\ref{eq:H}). Then $R_{\omega_{N,t}} \Omega$ is the Slater determinant with reduced one-particle density matrix $\omega_{N,t}$ and 
\[ R_{\omega_{N,t}}^* a(g) R_{\omega_{N,t}} = a (u_{N,t} g) + a^* (\bar{v}_{N,t} \bar{g}) \]
with $u_{N,t}, v_{N,t}$ defined similarly as $u_N, v_N$ after (\ref{eq:bog}).

\section{Proof of the Main Result}\label{sec:proofmain}
Here we shall prove our main result, Theorem \ref{thm:main}. It will be a corollary of an estimate for the growth of the number operator evolved with a suitable fluctuation dynamics, Proposition \ref{thm:propcomm}, proven in the next section.
\subsection{Bound on the Growth of Fluctuations}\label{sec:bdfluct}
Let $\omega_{N,t}$ be the solution of the time-dependent Hartree equation. We introduce the {\it fluctuation dynamics} 
\begin{equation}
\label{eq:fluc} 
\mathcal{U}_{N}(t;s) := R_{\omega_{N,t}}^{*} \ee^{-\ii\cH_{N} (t - s)/\veps} R_{\omega_{N,s}}\;.
\end{equation}
Given an $N$-particle state $\psi$, we define the corresponding fluctuation vector $\xi = R^{*}_{\omega_{N}} \psi$. Then, we rewrite the many-body evolution of $\psi$ as 
\begin{equation*}
\psi_{t} = e^{-i\mathcal{H}_{N} t / \varepsilon} R_{\omega_{N}} \xi = R_{\omega_{N,t}} \mathcal{U}_{N}(t;0) \xi\;.
\end{equation*}
To show our main theorem we need to prove that, for $N$-particle initial data $\psi$ close to the Slater determinant with reduced one-particle density matrix $\omega_N$, the evolution $\psi_t$ remains close to the Slater determinant with reduced one-particle density matrix $\omega_{N,t}$. This will follow, if we can control the growth of the expectation of the number of particles 
\begin{equation}\label{eq:contro-N}
\langle \mathcal{U}_{N}(t;0) \xi, \mathcal{N} \mathcal{U}_{N}(t;0) \xi \rangle\;.
\end{equation}
To reach this goal, a key ingredient is the propagation of the semiclassical structure, introduced in Assumption \ref{ass:sc}, along the flow of the Hartree equation. This is the content of the next theorem.
\begin{theorem}[Propagation of the Local Semiclassical Structure]\label{thm:propcomm} Under the same assumptions of Theorem \ref{thm:main}, the following is true. There exists $C >0$ and $T>0$ such that:
\begin{equation}
\label{eq:propagaz2}
\sup _{t \in [0,T]} \sup _{z \in \R^{3}} X_{\Lambda}(z) \big\|\W_{z}^{(n)}  \, \omega_{N,t} \big\|_{\tr}  \leq C \veps^{-3} \;,
\end{equation}
and
\begin{equation}\label{eq:propagazione}
\sup _{t \in [0,T]} \sup_{p: |p|\leq \varepsilon^{-1}} \sup _{z\in \R^{3}}
 \; \frac{X_{\Lambda}(z)}{1+|p|}
\Big \| \W_{z}^{(n)} \,\big[e^{i p \cdot \hat{x}}, \omega_{N,t}\big] \Big \|_{\tr}  \leq  C \varepsilon^{-2} \;,
\end{equation}
\end{theorem}
\begin{remark}
\begin{itemize}
\item[(i)] The constant $C$ is a priori explicit. It is independent of $|\Lambda|$ and of $\varepsilon$, and it depends on $T$.
\item[(ii)] The requirement that $\omega_{N}$ is a rank-$N$ projection is not needed in this theorem. It could be replaced by $0\leq \omega_{N} \leq 1$, $\tr\, \omega_{N} = N$.
\end{itemize}
\end{remark}
The proof of Theorem \ref{thm:propcomm} is postponed to Section \ref{sec:propsc}. From Theorem \ref{thm:propcomm}, we obtain the following corollary, which will be used to control the growth of the expectation (\ref{eq:contro-N}) and thus to prove our main result, Theorem \ref{thm:main}. 
\begin{corollary}[Bounds for Commutators with Regular Functions]
\label{cor:commF}
Under the same assumptions of Theorem \ref{thm:propcomm}, the following is true. Let:
\begin{equation}\label{eq:Fass}
F(x) = \int_{\mathbb{R}^{3}} dp\, e^{ip \cdot x} \hat F(p)\;,\qquad  \int dp\, (1 + |p|) \big| \partial_{p}^{k} \hat F(p) \big| \leq C \quad \text{for all $k \leq 8 n$.}
\end{equation}
Let $F_{z}(x) = F(x-z)$. Then, the following bound holds true:
\begin{equation}\label{eq:commF}
\sup _{t \in [0,T]} \sup _{z \in \R^{3}} X_{\Lambda}(z) \Big\| \big[ \omega_{N,t}, F_{z}(\hat x) \big]\Big\|_{\tr} \leq C \veps^{-2} \;.
\end{equation}
\end{corollary}
\begin{proof} Let $\chi(p)$ be a smooth, non-increasing, compactly supported function, equal to $1$ for $|p| \leq \varepsilon^{-1} - 1$ and equal to $0$ for $|p| > \varepsilon^{-1}$. We write:
\begin{equation}
\label{cut-off_potential_sharp}
\begin{split}
F(x) & =  F^{(\leq)}(x) + F^{(>)}(x)\;,
\\
F^{(\leq)}(x) & = \int_{\R^{3}} d p  \, e^{i p \cdot x} \chi(p) \hat F(p) \;, 
\\
F^{(>)}(x) & = \int_{\R^{3}} d p \, e^{i p \cdot x} (1 - \chi(p)) \hat F(p) \;,
\end{split}
\end{equation}
so that
\begin{equation}
\label{splitting_wVj}
\Big\| \big[ \omega_{N,t},F_{z}(\hat x) \big]\Big\|_{\tr}  \leq \Big\| \big[ \omega_{N,t}, F^{(\leq)}_{z}(\hat x) \big]\Big\|_{\tr} + \Big\| \big[ \omega_{N,t}, F^{(>)}_{z}(\hat x) \big]\Big\|_{\tr} \;.
\end{equation}
Let us define, for $\sharp = \leq, >$:
\begin{equation}
\label{def_fj_gj}
f^{(\sharp)}(x) := (1+|x|^{4n}) F^{(\sharp)}(x)\;,\qquad  g^{(\sharp)}(x):= (1+|x|^{4n})^{2} F^{(\sharp)}(x)\;,
\end{equation}
together with $f_z^{\sharp}(x):=f^{\sharp}(x-z)$ and $g^{\sharp}_z(x):=g^{\sharp}(x-z)$. By the assumptions (\ref{eq:Fass}) and the smoothness of $\chi$, 
\begin{equation}\label{eq:fgsharp}
\| (1 + |\cdot|) \hat f^{(\sharp)} \|_{1} \leq C\;,\qquad \| \hat g^{(\sharp)} \|_{1} \leq C\;.
\end{equation}
Consider the first term on the r.h.s.~of (\ref{splitting_wVj}). We have:
\begin{equation}\label{eq:Wf}
\begin{split}
 \Big\| \big[ \omega_{N,t}, F^{(\leq)}_{z}(\hat{x}) \big]\Big\|_{\tr}&=  \Big\| \big[ \omega_{N,t}, \mathcal{W}_{z}^{(n)} f^{(\leq)}_{z}(\hat{x}) \big]\Big\|_{\tr} \\
 &\leq \Big\| \big[ \omega_{N,t}, \mathcal{W}_{z}^{(n)} \big] f^{(\leq)}_{z} (\hat{x}) \Big\|_{\tr} + \Big\| \mathcal{W}_{z}^{(n)} \big[ \omega_{N,t}, f^{(\leq)}_{z} (\hat{x}) \big]\Big\|_{\tr}\;.
 \end{split}
\end{equation}
Consider the second term on the r.h.s.~of (\ref{eq:Wf}). We estimate it as:
\begin{equation}\label{eq:Wfsec}
\Big\| \mathcal{W}_{z}^{(n)} \big[ \omega_{N,t}, f^{(\leq)}_{z} (\hat{x})\big]\Big\|_{\tr} \leq \int dp\, |\hat f^{(\leq)}(p)| \Big\| \mathcal{W}_{z}^{(n)} \big[ \omega_{N,t}, e^{ip\cdot\hat{x}} \big]\Big\|_{\tr}\;.
\end{equation}
Therefore, using (\ref{eq:propagazione}), the contribution of this term to the final bound (\ref{eq:commF}) is:
\begin{equation}\label{eq:Fcomm1}
\begin{split}
\sup _{t \in [0,T]} \sup _{z \in \R^{3}} &X_{\Lambda}(z) \Big\| \mathcal{W}_{z}^{(n)} \big[ \omega_{N,t}, f^{(\leq)}_{z}(\hat{x}) \big]\Big\|_{\tr} \\&\leq \sup _{t \in [0,T]} \sup _{z \in \R^{3}}\int dp\, |\hat f^{(\leq)}(p)| (1 + |p|) \frac{X_{\Lambda}(z)}{1 + |p|} \Big\| \mathcal{W}_{z}^{(n)} \big[ \omega_{N,t}, e^{ip\cdot \hat{x}} \big]\Big\|_{\tr} \\
&\leq C \sup _{t \in [0,T]} \sup _{z \in \R^{3}} \sup_{p: |p| \leq |\varepsilon|^{-1}}\frac{X_{\Lambda}(z)}{1 + |p|} \Big\| \mathcal{W}_{z}^{(n)} \big[ \omega_{N,t}, e^{ip \cdot \hat{x}} \big]\Big\|_{\tr} \\
&\leq C\varepsilon^{-2}\;,
\end{split}
\end{equation}
where we used (\ref{eq:fgsharp}). Consider now the first term in the r.h.s.~of (\ref{eq:Wf}). We estimate it as:
\begin{equation*}
\Big\| \big[ \omega_{N,t}, \mathcal{W}_{z}^{(n)} \big] f^{(\leq)}_{z} (\hat{x}) \Big\|_{\tr} = \Big\| \big[ \omega_{N,t}, \mathcal{W}_{z}^{(n)} \big] \mathcal{W}_{z}^{(n)} g^{(\leq)}_{z} (\hat{x})\Big\|_{\tr} \leq C\Big\| \big[ \omega_{N,t}, \mathcal{W}_{z}^{(n)} \big] \mathcal{W}_{z}^{(n)}\Big\|_{\tr}\;,
\end{equation*}
where we used that $g^{(\sharp)}$ is bounded, which follows from (\ref{eq:fgsharp}). We then have:
\begin{equation}\label{eq:Wcomm}
\begin{split}
&\Big\| \big[ \omega_{N,t}, \mathcal{W}_{z}^{(n)} \big] \mathcal{W}_{z}^{(n)}\Big\|_{\tr} \\
& \leq \int dp\, |\widehat{\mathcal{W}^{(n)}}(p)|\,  \Big\| \big[ \omega_{N,t}, e^{ip\cdot \hat{x}} \big] \mathcal{W}_{z}^{(n)}\Big\|_{\tr} \\
& \leq \int_{|p|\leq \varepsilon^{-1}} dp\, |\widehat{\mathcal{W}^{(n)}}(p)|\,  \Big\| \big[ \omega_{N,t}, e^{ip\cdot\hat{x}} \big] \mathcal{W}_{z}^{(n)}\Big\|_{\tr} + \int_{|p| > \varepsilon^{-1}} dp\, |\widehat{\mathcal{W}^{(n)}}(p)|\,  \Big\| \big[ \omega_{N,t}, e^{ip\cdot \hat{x}} \big] \mathcal{W}_{z}^{(n)}\Big\|_{\tr}\;,
\end{split}
\end{equation}
where we denoted the Fourier transform of $(1+|\cdot|^{4n})^{-1}$ by $\widehat{\W^{(n)}}$.
The contribution of the first term to the final bound (\ref{eq:commF}) is estimated exactly as before. We get:
\begin{equation}\label{eq:Fcomm2}
\sup _{t \in [0,T]} \sup _{z \in \R^{3}} X_{\Lambda}(z) \int_{|p|\leq \varepsilon^{-1}} dp\, |\widehat{\mathcal{W}^{(n)}}(p)|\,  \Big\| \big[ \omega_{N,t}, e^{ip\hat{x}} \big] \mathcal{W}_{z}^{(n)}\Big\|_{\tr} \leq C\varepsilon^{-2}\;.
\end{equation}
Consider now the second term in (\ref{eq:Wcomm}). We estimate it as:
\begin{equation}\label{eq:bla}
\int_{|p| > \varepsilon^{-1}} dp\, |\widehat{\mathcal{W}^{(n)}}(p)|\,  \Big\| \big[ \omega_{N,t}, e^{ip\cdot \hat{x}} \big] \mathcal{W}_{z}^{(n)}\Big\|_{\tr} \leq 2 \int_{|p| > \varepsilon^{-1}} dp\, |\widehat{\mathcal{W}^{(n)}}(p)| \Big\| \omega_{N,t}\mathcal{W}_{z}^{(n)}\Big\|_{\tr} \;.
\end{equation}
Using that 
\begin{equation*}
\int_{|p| > \varepsilon^{-1}} dp\, |\widehat{\mathcal{W}^{(n)}}(p)| \leq \varepsilon \int dp\, |p| |\widehat{\mathcal{W}^{(n)}}(p)| \leq C\varepsilon\;,
\end{equation*}
and using \eqref{eq:propagaz2} to estimate the last trace norm in (\ref{eq:bla}), we get:
\begin{equation}\label{eq:Fcomm3}
\sup _{t \in [0,T]} \sup _{z \in \R^{3}} X_{\Lambda}(z)\int_{|p| > \varepsilon^{-1}} dp\, |\widehat{\mathcal{W}^{(n)}}(p)|\,  \Big\| \big[ \omega_{N,t}, e^{ip\cdot\hat{x}} \big] \mathcal{W}_{z}^{(n)}\Big\|_{\tr} \leq C\varepsilon^{-2}\;.
\end{equation}
Putting together (\ref{eq:Fcomm1}), (\ref{eq:Fcomm2}), (\ref{eq:Fcomm3}), we find:
\begin{equation}\label{eq:Fcomm4}
\sup _{t \in [0,T]} \sup _{z \in \R^{3}} X_{\Lambda}(z) \Big\| \big[ \omega_{N,t}, \mathcal{W}_{z}^{(n)} \big] f^{(\leq)}_{z}(\hat{x})\Big\|_{\tr} \leq C\varepsilon^{-2}\;.
\end{equation}
Together with (\ref{eq:Wf}), the bounds (\ref{eq:Fcomm1}), (\ref{eq:Fcomm4}) imply:
\begin{equation}\label{eq:Fmin}
\sup _{t \in [0,T]} \sup _{z \in \R^{3}} X_{\Lambda}(z) \Big\| \big[ \omega_{N,t}, F^{(\leq)}_{z}(\hat{x}) \big]\Big\|_{\tr} \leq C\varepsilon^{-2}\;.
\end{equation}
This proves the desired estimate for the first term on the r.h.s.~of (\ref{splitting_wVj}). Consider now the second term on the r.h.s.~of (\ref{splitting_wVj}). By opening the commutator and by using the invariance of the norm under hermitian conjugation, we get:
\begin{equation}\label{eq:wF>comm}
\Big\| \big[ \omega_{N,t}, F^{(>)}_{z}(\hat{x}) \big]\Big\|_{\tr} \leq 2 
\Big\| \omega_{N,t} F^{(>)}_{z}(\hat{x}) \Big\|_{\tr} \leq 2 
\Big\| \omega_{N,t} \W_{z}^{(n)} f^{(>)}_{z}(\hat{x}) \Big\|_{\tr}.
\end{equation}
By (\ref{eq:fgsharp}), together with the fact that $\hat f^{(>)}(p) = 0$ for $|p| < \varepsilon^{-1} - 1$:
\begin{equation*}
\|f^{(>)}(\hat{x}) \|_{\mathrm{op}} \leq 
\int dp\, |\hat f^{(>)}(p)|  \leq C\varepsilon\int dp\, \big| \hat f^{(>)}(p)\big| |p|\leq C\varepsilon\;;
\end{equation*}
using \eqref{eq:propagaz2}, we easily get:
\begin{equation}\label{eq:Fbig1}
\sup _{t \in [0,T]} \sup _{z \in \R^{3}} X_{\Lambda}(z) \Big\|  \omega_{N,t} \mathcal{W}_{z}^{(n)} f^{(>)}_{z} (\hat{x})\Big\|_{\tr} \leq C\varepsilon^{-2}\;.
\end{equation}
Plugging this estimate in \eqref{eq:wF>comm}, we get:
\begin{equation*}
\sup _{t \in [0,T]} \sup _{z \in \R^{3}} X_{\Lambda}(z) \Big\| \big[ \omega_{N,t}, F^{(>)}_{z} (\hat{x})\big]\Big\|_{\tr} \leq C\varepsilon^{-2}\;.
\end{equation*}
Combined with (\ref{eq:Fmin}) and with (\ref{splitting_wVj}), this concludes the proof of Corollary \ref{cor:commF}.
\end{proof}
The next result is the key to control the distance between the many-body and the effective dynamics. It relies on the propagation of the local semiclassical structure, Theorem \ref{thm:propcomm}, and on its Corollary \ref{cor:commF}.
\begin{proposition}[Bound on the Growth of Fluctuations] \label{prp:fluct} Under the same assumptions of Theorem \ref{thm:main}, the following is true. Let $\psi$ be the Fock space vector associated to $\psi_{N}$, and let $\xi = R^{*} \psi$. Then, there exists $C>0$ such that:
\begin{equation}
\sup _{t \in [0,T]}
\langle \xi, \mathcal{U}_{N}(t;0)^{*} \cN \mathcal{U}_{N}(t;0) \xi \rangle \leq C ( N \varepsilon + \langle \xi, \cN \xi \rangle )\;.\label{eq:propN}
\end{equation}
\end{proposition}
The proof of Theorem \ref{thm:main} is a corollary of Proposition \ref{prp:fluct}, and will be given in Section \ref{sec:proofthm}. Let us now prove Proposition \ref{prp:fluct}
\begin{proof}[Proof of Proposition \ref{prp:fluct}.] The proof is based on a Gronwall-type argument, as in \cite{BPS, BPS2}. For convenience, we shall use the following notations 
\begin{equation*}
u_{t;x}(\cdot):= u_{N,t}(\,\cdot\,;x)\;,\quad v_{t;x}(\cdot):= v_{N,t}(\,\cdot\,;x)\;,\quad \overline{v}_{t;x}(\cdot):= \overline{v_{N,t}(\,\cdot\,;x)}\;.
\end{equation*}
The starting point is the following identity, for any $\xi \in \mathcal{F}$:
\begin{equation}
\label{eq:derN}
\begin{split}
i \veps \partial_{t}  & \Big \langle \xi ,
\mathcal{U}_{N}(t;0)^{*} \cN \mathcal{U}_{N}(t;0) \xi \Big \rangle 
\\
&
 = -4 i \veps^{3} \im \,\int d x d y \, V(x-y) \,
\Big\langle \xi, \mathcal{U}_{N}(t;0)^{*}\Big( 
a(\overline{v}_{t;x}) a(\overline{v}_{t;y}) a(u_{t;y}) a(u_{t;x}) 
\\
&
\quad
+
 a^{\ast}(u_{t;x}) a(\overline{v}_{t;y}) a(u_{t;y}) a(u_{t;x}) 
+
  a^{*}(u_{t;y}) a^{*}(\overline{v}_{t;y}) a^{*}(\overline{v}_{t;x}) a(\overline{v}_{t;x})
 \Big) \,\mathcal{U}_{N}(t;0) \xi \Big\rangle
 \\ 
 &\quad + 4 i \veps^{3} \im \int d x d y \, V(x-y)\, \Big \langle \xi, \mathcal{U}_{N}(t;0)^{*} \Big( \omega_{N,t}(y;x)   a^{*}(u_{t,y})a^{*}(\overline{v}_{t,x})\Big) \mathcal{U}_{N}(t;0) \xi \Big \rangle 
 \\
 & = \mathrm{I}+\mathrm{II} + \mathrm{III} + \mathrm{IV}\;,
\end{split}
\end{equation}
with a natural identification of each term.
The proof of this identity follows  \cite[Proof of Proposition 3.3]{BPS}. The difference with respect to \cite{BPS} is that now the fluctuation dynamics is defined starting from the Hartree equation rather than the Hartree--Fock equation. This is the reason for the extra quadratic term in the right-hand side of (\ref{eq:derN}), which in \cite{BPS} is cancelled by the presence of the exchange term in the time-dependent Hartree--Fock equation. Let us briefly sketch the arguments and refer to \cite[Proof of Proposition
3.3]{BPS} for further details.
By using the definition of the Bogoliubov transformation, one can see that
\begin{equation}
\label{fluctuation_derivative}
\ii\veps \partial_{t} \mathcal{U}_{N}(t;0)^{*} \cN \mathcal{U}_{N}(t;0) = -2 \, \mathcal{U}_{N}(t;0)^{*} R_{\omega_{N,t}}^* \Big(  d\Gamma (\ii\veps \partial_{t} \omega_{N,t}) - [\cH_N , d\Gamma (\omega_{N,t}) ] \Big) R_{\omega_{N,t}} \,\mathcal{U}_{N}(t;0) \;. 
\end{equation}
Plugging the Hartree equation, and using that 
\begin{equation*}
d \Gamma \big([-\veps^{2}\Delta, \omega_{N,t}]\big) = [d\Gamma(-\veps^{2}\Delta), d\Gamma(\omega_{N,t})]\;, 
\end{equation*}
we easily get: 
\begin{equation*}
\begin{split}
\ii\veps \partial_{t} \, &\mathcal{U}_{N}(t;0)^{*} \cN \mathcal{U}_{N}(t;0)
\\
& 
= -2 \, \mathcal{U}_{N}(t;0)^{*} R_{\omega_{N,t}}^{*} \Big(  d\Gamma \big( [V \ast \rho_{\Lambda,t}, \omega_{N,t}]\big) - [\mathcal{V} , d\Gamma (\omega_{N,t}) ] \Big) R_{\omega_{N,t}} \,\mathcal{U}_{N}(t;0)\;,
\end{split}
\end{equation*}
where $\mathcal{V}$ is the second quantization of the many-body interaction. After conjugation with the Bogoliubov transformation, see \eqref{eq:bog}, and normal ordering, we get:
\begin{equation*}
R_{\omega_{N,t}}^{*} 
d\Gamma \big( [V \ast \rho_{\Lambda,t}, \omega_{N,t}]\big) R_{\omega_{N,t}}  = \veps^{3} \int d x d y \, V(x-y) \omega_{N,t}(x,x) a^{*}(u_{t,y})a^{*}(\overline{u}_{t,y}) - \text{h.c.} \;,
\end{equation*}
and
\begin{equation}
\begin{split}\label{eq:V2}
R&_{\omega_{N,t}}^{*}[\mathcal{V} , d\Gamma (\omega_{N,t}) ] R_{\omega_{N,t}} = \veps^{3} \int d x\, d y  \,V(x-y) \Big(a(\overline{v}_{t;x}) a(\overline{v}_{t;y}) a(u_{t;y}) a(u_{t;x}) 
\\
& +
a^{\ast}(u_{t;x}) a(\overline{v}_{t;y}) a(u_{t;y}) a(u_{t;x})
+
a^{*}(u_{t;y}) a^{*}(\overline{v}_{t;y}) a^{*}(\overline{v}_{t;x}) a(\overline{v}_{t;x})  \Big)
\\
& +\veps^{3} \int d x d y \, V(x-y) \Big( \omega_{N,t}(x;x) a^{*}(u_{t,y})a^{*}(\overline{u}_{t,y}) - \omega_{N,t}(y;x) a^{*}(u_{t,y})a^{*}(\overline{v}_{t,x})\Big)\nonumber\\
& - \text{h.c.}
\end{split}
\end{equation}
which gives the claim. 

Let us now bound the terms appearing on the r.h.s.~of (\ref{eq:derN}). For brevity, we set $\xi_{t}:= \mathcal{U}_{N}(t;0) \xi$.

\medskip

\noindent{\underline{Bound for the term $\text{I}$}.} 
We rewrite the interaction potential as:
\begin{equation}\label{eq:Vsplit}
V(x-y) = \int_{\R^{3}} d z \, V^{(1)}(x-z) V^{(2)}(z-y)\;,
\end{equation}
where:
\begin{equation}
\label{def_Vj}
V^{(1)}(x)  := \int_{\R^{3}} d p \,\frac{e^{i p\cdot x}}{1+|p|^{6}} \;, \qquad V^{(2)}(x)  := \int_{\R^{3}} d p \, e^{i p \cdot x}\,  (1+|p|^{6})\widehat{V}(p) \;.
\end{equation}
The function $V^{(2)}$ is bounded, and its regularity can be inferred from the assumption \eqref{eq:assV} on the potential. Accordingly, the term $\mathrm{I}$ is re-written as
\begin{equation*}
\begin{split}
\mathrm{I} & = \varepsilon^{3} \int_{(\R^{3})^{2}} d x d y \, V(x-y) 
\left \langle  \xi_{t},
a(\overline{v}_{t;x}) a(\overline{v}_{t;y}) a(u_{t;y}) a(u_{t;x}) \xi_{t}\right \rangle
 \\
 & = \varepsilon^{3} \int_{\R^{3}} d z \, \left \langle \int_{\R^{3}} d x 
\, V^{(1)}_{z}(x) \,a^{*}(u_{t;x}) a^{*}(\overline{v}_{t;x})\xi _{t}\, , \int_{\R^{3}} d y \,V^{(2)}_{z}(y) a(\overline{v}_{t;y}) a(u_{t;y}) \xi_{t} \right  \rangle \;,
\end{split}
\end{equation*}
where we used the notation $f_{z}(x) = f(x-z)$. Next, we notice that 
\begin{equation}
\label{represent_uVv}
\begin{split}
\int_{\R^{3}} d x 
\, V^{(1)}_{z}(x) \,a^{\ast}(\overline{v}_{t;x})a^{\ast}(u_{t;x}) &= \int_{(\R^{3})^{2}} d r d s \, \Big(\int_{\R^{3}} d x \,  u_{N,t}(s;x) V_{z}^{(1)}(x) \overline{v}_{N,t}(r;x)\Big)a^{\ast}_{r}a^{\ast}_{s}
\\
& =\int_{(\R^{3})^{2}} d r d s \, \big(u_{N,t}V_{z}^{(1)}(\hat{x}) \overline{v}_{N,t} \big)(r;s)a^{\ast}_{r}a^{\ast}_{s}
\end{split}
\end{equation}
and that
\begin{equation}
\label{representation_uVv}
\begin{split}
\int_{\R^{3}} d y
\, V^{(2)}_{z}(y) \,a(\overline{v}_{t;y})a(u_{t;y}) & = \int_{(\R^{3})^{2}} d r d s \, \Big(\int_{\R^{3}} d y \,  \overline{u_{N,t}(s;y)} V_{z}^{(2)}(y) v_{N,t}(r;y)\Big)a_{r}a_{s}
\\
& =\int_{(\R^{3})^{2}} d r d s \, \big( v_{N,t} V_{z}^{(2)}(\hat{x})  u_{N,t} \big)(r;s)a_{r}a_{s} \;.
\end{split}
\end{equation}
By the Cauchy-Schwarz inequality and by Lemma \ref{lemma:bounds} we can bound the term $\mathrm{I}$ as:
\begin{equation*}
\begin{split}
| \mathrm{I}| &\leq  \varepsilon^{3} \bigg( \int_{\R^{3}} d z \, \Big\|\int_{\R^{3}} d x 
\, V^{(1)}_{z}(x) \,a^{*}(\overline{v}_{t;x})a^{*}(u_{t;x}) \xi _{t} \Big\|^{2}\bigg)^{1/2} \, \cdot
\\
& \qquad \qquad \qquad\cdot \;\bigg( \int_{\R^{3}} d z \, \Big\|\int_{\R^{3}} d x 
\, V^{(2)}_{z}(x) \,a(\overline{v}_{t;x})a(u_{t;x}) \xi_{t} \Big\|^{2}\bigg)^{1/2} 
\\
& \leq  \varepsilon^{3} \bigg( \int_{\R^{3}} d z \, \Big\| u_{N,t}V_{z}^{(1)}(\hat{x}) \overline{v}_{N,t} \Big\|^{2}_{\mathrm{tr}}\bigg)^{1/2} \bigg( \int_{\R^{3}} d z \, \Big\| v_{N,t} V_{z}^{(2)}(\hat{x}) u_{N,t} \Big\|^{2}_{\mathrm{tr}}\bigg)^{1/2}
\\
&
 \leq \varepsilon^{3} \bigg( \int_{\R^{3}} d z \, \Big\| \big[ \omega_{N,t},V_{z}^{(1)}(\hat{x}) \big]\Big\|^{2}_{\mathrm{tr}}\bigg)^{1/2}  \bigg( \int_{\R^{3}} d z \, \Big\| \big[ \omega_{N,t},V_{z}^{(2)}(\hat{x}) \big]\Big\|^{2}_{\mathrm{tr}}\bigg)^{1/2}
\\ 
 &
\leq \varepsilon^{3} \bigg(
 \int_{\R^{3}} d z\, X_{\Lambda}(z)^{-2}\bigg)\prod_{j =1,2} \sup _{z \in \R^{3}} X_{\Lambda}(z)
  \Big\| \big[ \omega_{N,t},V_{z}^{(j)}(\hat{x}) \big]\Big\|_{\tr}\;,
\end{split}
\end{equation*}
where in the third step we used that $u_{N,t} =1 - \omega_{N,t}$ together with the orthogonality condition $u_{N,t}\overline{v}_{N,t} = v_{N,t} u_{N,t} = 0$, and that $\| \overline{v}_{N,t}\|_{\mathrm{op}} = 1$. The trace norm of the commutator can be estimated using Corollary \ref{cor:commF}. In fact, the function $V^{(1)}$ satisfies the assumptions of Corollary \ref{cor:commF}, and the same is true for the function $V^{(2)}$, thanks to the assumptions on the potential $V$, Eq.~(\ref{eq:assV}). Therefore, we get:
\begin{equation*}
\sup_{t\in [0,T]}\prod_{j =1,2} \sup _{z \in \R^{3}} X_{\Lambda}(z) \Big\| \big[ \omega_{N,t},V_{z}^{(j)}(\hat{x}) \big]\Big\|_{\tr} \leq C\varepsilon^{-4}\;.
\end{equation*}
Using that $\int_{\R^{3}} d z\, X_{\Lambda}(z)^{-2} \leq C|\Lambda|$, we find:
\begin{equation}\label{eq:estI}
|\text{I}| \leq C|\Lambda| \varepsilon^{-1} \leq C\varepsilon^{2} N\;. 
\end{equation}

\noindent{\underline{Bound for the term $\text{II}$}.} We write:
\begin{equation*}
\begin{split}
\mathrm{II} 
& = \varepsilon^{3} \int_{\R^{3}} d x \, \Big\langle a(u_{t;x})
 \xi_{t}, \Big(\int_{\R^{3}} d y \,V_{x}(y)  a(\overline{v}_{t;y}) a(u_{t;y}) \Big) a(u_{t;x})\xi_{t} \Big\rangle 
\\
& = \varepsilon^{3} \int_{\R^{3}} d x \, \Big\langle a(u_{t;x})
 \xi_{t}, \Big(\int_{(\R^{3})^{2}} d r \, d s
 \Big(v_{N,t} V_{x}(\hat{x})u_{N,t} \Big)(r,s)a_{r} a_{s}  \Big) a(u_{t;x})\xi_{t} \Big\rangle  \;,
\end{split}
\end{equation*}
recall \eqref{representation_uVv}. By the Cauchy-Schwarz inequality and by Lemma \ref{lemma:bounds} we obtain
\begin{equation}
\label{bound_II}
\begin{split}
| \mathrm{II}|&  \leq \varepsilon^{3} \int_{\R^{3}} d x \, \big\| a(u_{t;x}) \xi_{t} \big\| \, \big\|v_{N,t} V_{x}(\hat{x})u_{N,t}\big\|_{\tr} \, \big\| a(u_{t;x})\xi_{t} \big\|
\\
& 
\leq \varepsilon^{3} \sup _{z \in \R^{3}} \big\|[ V_{z}(\hat{x}),\omega_{N,t} ]\big\|_{\tr} \int_{\R^{3}} d x \, \big\| a(u_{t;x}) \xi_{t} \big\|^{2} \;.
\end{split}
\end{equation}
The trace norm of the commutator can be estimated using Corollary \ref{cor:commF}. The last term is estimated in terms of the number operator:
\begin{equation*}
\begin{split}
\int_{\R^{3}} d x \, \big\| a(u_{t;x}) \xi_{t} \big\|^{2} 
& =
\langle \xi_{t}, d\Gamma(u_{N,t}) \xi_{t}\rangle \leq \langle \xi_{t}, \cN \xi_{t} \rangle \;,
\end{split}
\end{equation*}
where we used that $u_{N,t}^{2} = u_{N,t}$. Thus, we get, for $t\in [0,T]$:
\begin{equation}\label{eq:estII}
|\mathrm{II}| \leq C \veps \langle \xi_{t}, \cN \xi_{t} \rangle \;.
\end{equation}

\medskip

\noindent{\underline{Bound for the term $\text{III}$.}} We write
\begin{equation*}
\begin{split}
\mathrm{III} 
 = \varepsilon^{3} \int_{\R^{3}} d x \, \Big\langle a(\overline{v}_{t;x})
 \xi_{t}, \Big(\int_{(\R^{3})^{2}} d r \, d s\,
 \Big( u_{N,t} V_{x}(\hat{x})\overline{v}_{N,t} \Big)(r;s)a_{r} a_{s}  \Big)a(v_{t;x})\xi_{t} \Big\rangle  \;.
\end{split}
\end{equation*}
Proceeding as we did for the term $\text{II}$, we find, for $t \in [0;T]$, 
\begin{equation}
\label{eq:estIII}
\begin{split}
|\mathrm{III}| \leq \varepsilon^{3} \sup _{z \in \R^{3}} \big\|[ V_{z}(\hat{x}),\omega_{N,t} ]\big\|_{\tr} \int_{\R^{3}} d x \, \big\| a(v_{t;x}) \xi_{t} \big\|^{2} 
\leq C \veps   \langle \xi_{t}, \cN \xi_{t} \rangle \;,
\end{split}
\end{equation}
where we used that $\overline{v}_{N,t} v_{N,t} \leq 1$.

\medskip

\noindent{\underline{Bound for the term $\text{IV}$}.} Finally, we consider the term $\mathrm{IV}$, containing the quadratic contributions. We rewrite the potential as in (\ref{eq:Vsplit}). We then get:
\begin{equation*}
\begin{split}
\mathrm{IV} & =\varepsilon^{3}  \int_{\R^{3}} d z\, \Big\langle \xi_{t}, \int_{(\R^{3})^{2}} d x d y\, V^{(1)}_{z}(x) V^{(2)}_{z}(y) \omega_{N,t}(y,x) a^{*}(u_{t,y}) a^{*}(\overline{v}_{t,x})  \xi_{t} \Big\rangle
\\
& = \varepsilon^{3} \int_{\R^{3}} d z\, \Big\langle \xi_{t}, \int_{(\R^{3})^{2}} d r d s \,\Big( u_{N,t} V^{(2)}_{z} \omega_{N,t} V^{(1)}_{z} \overline{v}_{N,t} \Big)(r;s)  a^{*}_{r} a^{*}_{s}  \xi_{t} \Big\rangle
\end{split}
\end{equation*}
where we used that $v_{N,t}(s,x) = v_{N,t}(x,s)$. By the Cauchy-Schwarz inequality and by Lemma \ref{lemma:bounds} we obtain
\begin{equation}
\label{eq:bound-IV}
\begin{split}
|\mathrm{IV}| &\leq \varepsilon^{3} \int_{\R^{3}} d z \, \Big\| \int_{(\R^{3})^{2}} d r d s \,\big( u_{N,t} V^{(2)}_{z} \omega_{N,t} V^{(1)}_{z} \overline{v}_{N,t} \Big)(r,s)  a^{*}_{r} a^{*}_{s} \,  \xi_{t} \Big\|
\\
& \leq 
\varepsilon^{3} \int_{\R^{3}} d z \, \Big\| u_{N,t} V^{(2)}_{z}(\hat{x}) \omega_{N,t} V^{(1)}_{z}(\hat{x}) \overline{v}_{N,t}  \Big\|_{\tr}
\\
&
\leq \varepsilon^{3} \bigg(\int_{\R^{3}} d z\, X_{\Lambda}(z)^{-1}\bigg) \sup _{z \in \R^{3}} X_{\Lambda}(z) \Big\|  V^{(2)}_{z}(\hat{x}) \omega_{N,t} V^{(1)}_{z}(\hat{x}) \Big\|_{\tr}
\end{split}
\end{equation}
where we used that $\| u_{N,t}\|_{\mathrm{op}} = \| \overline{v}_{N,t}\|_{\mathrm{op}} = 1$. Next, we estimate:
\begin{equation*}
\Big\|  V^{(2)}_{z}(\hat{x}) \omega_{N,t} V^{(1)}_{z}(\hat{x}) \Big\|_{\tr} \leq C \Big\| \omega_{N,t} V^{(1)}_{z} \Big\|_{\tr} \leq C \Big\| \omega_{N,t} \mathcal{W}_{z}^{(n)} \Big\|_{\tr}
\end{equation*}
where we used that $\| V^{(2)}_{z} \|_{\infty} \leq C$ and that $\| (\mathcal{W}_{z}^{(n)})^{-1} V^{(1)}_{z} \|_{\infty} \leq C$. Hence, using the bound (\ref{eq:H2}), we get:
\begin{equation}\label{eq:estIV}
|\mathrm{IV}| \leq C |\Lambda| \leq C\varepsilon^{3} N \;.
\end{equation}
Notice that, by using the orthogonality between $u_{N,t}$ and $\omega_{N,t}$, we could have written the bound  in \eqref{eq:bound-IV} with $\Big\|  [V^{(2)}_{z}(\hat{x}), \omega_{N,t}] V^{(1)}_{z}(\hat{x}) \Big\|_{\tr}$ instead and eventually improved the estimate \eqref{eq:estIV} by $\veps$.

\medskip

\noindent{\underline{Conclusion.}} Putting together (\ref{eq:derN}), (\ref{eq:estI}), (\ref{eq:estII}), (\ref{eq:estIII}), (\ref{eq:estIV}), we have:
\begin{equation*}
\partial_{t}\langle \xi_{t}, \cN \xi_{t} \rangle \leq C N\varepsilon + C\langle \xi_{t}, \cN \xi_{t} \rangle\;.
\end{equation*}
By the Gronwall lemma, for all $t\in [0,T]$, for a constant $K$ depending on $T$:
\begin{equation*}
\langle \xi_{t}, \cN \xi_{t} \rangle \leq K ( N\varepsilon  + \langle \xi, \cN \xi \rangle )
\end{equation*}
which concludes the proof.
\end{proof}
\subsection{Proof of Theorem \ref{thm:main}}\label{sec:proofthm}
\begin{proof}[Proof of Theorem \ref{thm:main}.] We now prove our main result. It turns out that the distance between the many-body evolution and the Hartree equation can be quantified by the average number of particles in the fluctuation vector $\xi_{t} = \cU_N (t;0) \xi = R_{\omega_{N,t}}^* \psi_{N,t}$, associated with the solution $\psi_{N,t} = e^{-i H_N t /\veps} \psi_N$ of the Schr\"odinger equation (\ref{eq:schr0}), with initial data $\psi_N = R_{\omega_N} \xi$. In fact 
\begin{equation}
\begin{split}
\langle \xi_{t}, \mathcal{N} \xi_{t} \rangle &= \langle \psi_{N,t}, R_{\omega_{N,t}} \mathcal{N} R^{*}_{\omega_{N,t}} \psi_{N,t}  \rangle\\
&= \langle \psi_{N,t}, ( N - 2d\Gamma(\omega_{N,t}) + \mathcal{N}) \psi_{N,t}\rangle\\
&= 2\tr\, \gamma^{(1)}_{N,t} (1 - \omega_{N,t})\;,\label{eq:Ndist}
\end{split}
\end{equation}
where we used that $\tr\,\gamma^{(1)}_{N,t} = N$. Thus, we find 
\begin{equation}\label{eq:gamma2N}
\begin{split}
\| \gamma^{(1)}_{N,t} - \omega_{N,t} \|_{\text{HS}}^{2} &= \tr\, |\gamma^{(1)}_{N,t} - \omega_{N,t}|^{2} \\
&= \tr ( \gamma^{(1)2}_{N,t} + \omega_{N,t}^{2} - 2\gamma^{(1)}_{N,t} \omega_{N,t})\\
&\leq 2\tr\, \gamma^{(1)}_{N,t} (1 - \omega_{N,t}) \equiv \langle \xi_{t}, \mathcal{N} \xi_{t} \rangle\;.
\end{split}
\end{equation}
In the second step we used the cyclicity of the trace, while in the last step we used that $\gamma^{(1)}_{N,t} \leq 1$, $\omega_{N,t} \leq 1$ and that $\tr\, \gamma^{(1)}_{N,t} = \tr\, \omega_{N,t} = N$. On the other hand, the quantity $\langle \xi_{t}, \mathcal{N} \xi_{t} \rangle$ is controlled by the distance between $\gamma^{(1)}_{N,t}$ and $\omega_{N,t}$, in the trace norm topology. In fact:
\begin{equation*}
\tr\, \gamma^{(1)}_{N,t} (1 - \omega_{N,t}) = \tr\, (\gamma^{(1)}_{N,t} - \omega_{N,t}) (1 - \omega_{N,t}) \leq \| \gamma^{(1)}_{N,t} - \omega_{N,t} \|_{\text{tr}}\;.
\end{equation*}
In particular, the assumption \eqref{eq:trdatum}  implies that 
\begin{equation}\label{eq:estxi0}
\langle \xi, \mathcal{N} \xi \rangle \leq C\varepsilon^{\delta} N\;.
\end{equation}
Hence, thanks to Proposition \ref{prp:fluct}, we have:
\begin{equation*}
\langle \xi_{t}, \mathcal{N} \xi_{t} \rangle \leq C \max\{ \veps^{\delta}, \veps \} N\;;
\end{equation*}
plugging this bound in (\ref{eq:gamma2N}), the final claim (\ref{eq:mainclaim}) follows. This concludes the proof of Theorem \ref{thm:main}.
\end{proof}

\section{Propagation of the Semiclassical Structure: Nonrelativistic case}\label{sec:propsc}
\label{sec:prop_semiclassical}
Here we shall prove Theorem \ref{thm:propcomm}, which is the key technical ingredient to control the growth of the number of fluctuations around the Hartree dynamics, Proposition \ref{prp:fluct}. In what follows, we will denote multi-indices in $\mathbb{N}^{3}$ by Greek letters, and we shall denote the length of a multi-index by $|\alpha| := \sum_{j} \alpha_{j}$. Moreover, unless otherwise specified, we will denote generic constants, possibly depending on $n$ and $T$, by $C$ and $K$, with the understanding that these constants can be different on different lines.

Throughout the section, we will make use of the following elementary lemma. 
\begin{lemma}[Monotonicity Properties of the Trace Norm]\label{lem:mon} Let $A, B, C$ be bounded operators on $L^{2}(\mathbb{R}^{3})$ such that $|A|^{2}\leq |B|^{2}$. Suppose that $AC$ and $BC$ are trace class. Then:
\begin{equation*}
\| AC \|_{\tr} \leq \| BC \|_{\tr}\;.
\end{equation*}
\end{lemma}
\begin{proof}
We write:
\begin{equation*}
\| AC \|_{\text{tr}} = \tr\, \sqrt{ (AC)^{*} AC } = \tr\, \sqrt{ C^{*} |A|^{2} C } \leq \tr\, \sqrt{ C^{*} |B|^{2} C } = \| BC \|_{\text{tr}}\;,
\end{equation*}
where the inequality follows from the operator monotonicity of the square root. This concludes the proof of the lemma.
\end{proof}
\subsection{Evolution of the Localization Operators}
With respect to previous work, \cite{BPS,BPS2, BJPSS,PRSS}, one important difference is the presence of the localization operators $\mathcal{W}_{z}^{(n)}$ in the semiclassical structure defined in Assumption \ref{ass:sc}.  
To propagate these bounds, we need to control the behavior of the localization operators under the Hartree dynamics $U(t;s)$, defined by
\begin{equation*}
\label{eq:Hartree-dyn}
i\veps \partial_{t}U(t;s) = (-\veps^{2}\Delta + V \ast \rho_{t}) U(t;s) \;,
\qquad U(s;s) =1 \;,
\end{equation*}
where $\rho_{t}(x) = \veps^{3} \omega_{N,t}(x;x)$, $\omega_{N,t}$ is the solution of the Hartree equation with initial datum $\omega_{N}$. To achieve this, a key role will be played by the following proposition.
\begin{proposition}[No-Concentration Bound for Short Times]\label{prp:density} Under the assumption on the potential $V$ of Theorem \ref{thm:main}, and under the assumption of Eq.~(\ref{eq:assump_density}) on the initial datum $\omega_{N}$,
the following is true. There exists $T_{*}>0$ independent of $\varepsilon$ such that:
\begin{equation}\label{eq:densityt}
\sup_{t\in [0,T_{*}]} \sup_{z\in \mathbb{R}^{3}} \tr\, \mathcal{W}^{(1)}_{z} \omega_{N,t} \leq \varepsilon^{-3} C\;.
\end{equation}
\end{proposition}
%
%
We will postpone the proof of Proposition \ref{prp:density} after the proof of the next proposition, where we control the propagation of the localization operators. 
\begin{proposition}[Bounds on the Evolution of the Localization Operator]\label{prp:propU} Under the same assumptions of Theorem \ref{thm:main}, consider the modified Hartree generator $U_{p}(t;s)$, defined by 
\begin{equation*}
i\varepsilon \partial_{t} U_{p}(t;s) = ( -\varepsilon^{2}\Delta + V*\rho_{t} + i\varepsilon^{2} p \cdot \nabla ) U_{p}(t;s)\;,\qquad U_{p}(s;s) = 1\;.
\end{equation*}
Then, for all $z\in \mathbb{R}^{3}$, $t_{0} \in \mathbb{R}$, for any $0\leq s,t \leq T$ and any $1\leq k \leq 2n$ the following is true:
\begin{equation}\label{eq:propU1}
U_{p}(t;s)^{*} \mathcal{W}^{(k)}_{z}(t_{0}) U_{p}(t;s) \leq C\mathcal{W}^{(k)}_{z + \veps p(t-s)}(t_{0} + t - s)\;.
\end{equation}
\end{proposition}
\begin{remark}\label{rmk:prop5.1}
Using the inequalities:
\begin{equation}\label{eq:cC}
\begin{split}
c \mathcal{W}^{(k/2)2}_{z}(t_{0}) &\leq \mathcal{W}^{(k)}_{z}(t_{0}) \leq C \mathcal{W}^{(k/2)2}_{z}(t_{0})
\end{split}
\end{equation}
Eq.~(\ref{eq:propU1}) also implies:
\begin{equation*}
U_{p}(t;s)^{*} \mathcal{W}^{(k/2)2}_{z}(t_{0}) U_{p}(t;s) \leq C\mathcal{W}^{(k/2)2}_{z + \veps p(t-s)}(t_{0} + t - s) \;.
\end{equation*}
These inequalities are of the form $|A|^{2} \leq |B|^{2}$, and will be extensively used in combination with Lemma \ref{lem:mon}. 
\end{remark}
The proof of Proposition \ref{prp:propU} relies on the following technical lemma. 
\begin{lemma}\label{lem:conj} Let $k \in\mathbb{N}$ and let $F:\mathbb{R}^{3} \to \mathbb{C}$ be such that:
\begin{equation*}
\| D^{j} F \|_{\infty} := \max_{\alpha: |\alpha| = j} \| \partial^{\alpha} F \|_{\infty} <\infty \qquad \text{for all $j\leq 2 k$.}
\end{equation*}
Then, the following is true:
\begin{equation}\label{eq:conj1}
\begin{split}
\Big\| \frac{1}{1 + |\hat x(t)|^{2k}} F(\hat x) (1 + |\hat x(t)|^{2k}) - F(\hat x)\Big\|_{\op} &\leq C_{k} \Big(\max_{0 \leq j \leq 2k}\| D^{j} F \|_{\infty}  \Big) |t| \varepsilon(1 + t^{2k}\varepsilon^{2k})\;,
\end{split}
\end{equation}
for some constants $C_{k}$.
\end{lemma}
\begin{proof}[Proof of Lemma \ref{lem:conj}] We have:
\begin{equation}\label{eq:lem1}
\frac{1}{1 + |\hat x(t)|^{2k}} F(\hat x) (1 + |\hat x(t)|^{2k}) - F(\hat x) = \frac{1}{1 + |\hat x(t)|^{2k}} \big[ F(\hat x), |\hat x(t)|^{2k} \big]\;.
\end{equation}
Next, we write:
\begin{equation}\label{eq:Fxn}
\big[ F(\hat x), |\hat x(t)|^{2k} \big] = |\hat{x} (t)|^{2} \big[ F(\hat x), |\hat x(t)|^{2(k-1)} \big] + \big[ F(\hat x), |\hat x(t)|^{2} \big]  |\hat x(t)|^{2(k-1)}\;.
\end{equation}
Consider the last commutator. We have:
\begin{equation*}
\begin{split}
\big[ F(\hat x), |\hat x(t)|^{2} \big] &= \sum_{i=1}^{3} \Big( \hat x_{i}(t) \big[ F(\hat x), \hat x_{i}(t) \big] + \big[ F(\hat x), \hat x_{i}(t) \big]  \hat x_{i}(t)\Big) \\
&= \sum_{i=1}^{3} \Big( 2\hat x_{i}(t) \big[ F(\hat x), \hat x_{i}(t) \big] + \big[ \big[ F(\hat x), \hat x_{i}(t) \big], \hat x_{i}(t)\big]  \Big)\;.
\end{split}
\end{equation*}
Recalling that $\hat x_{i}(t) = \hat x_{i} - i 2 t \varepsilon \partial_{i}$, we get:
\begin{equation*}
\big[ F(\hat x), \hat x_{i}(t) \big] = i 2 t \varepsilon \partial_{i} F(\hat x)\;,\qquad \big[ \big[ F(\hat x), \hat x_{i}(t) \big], \hat x_{i}(t)\big] = -4 t^{2} \varepsilon^{2} \partial_{i}^{2} F(\hat x)\;.
\end{equation*}
Similarly, we can rewrite $\big[ F(\hat x), |\hat x(t)|^{2n} \big]$ as a sum of terms, involving $0\leq j\leq 2k-1$ operators $\hat x_{i}(t)$ on the left times a partial derivative of $F(\hat x)$ of order $2k-j$, multiplied by a factor $(2t)^{2k-j} \varepsilon^{2k-j}$.  It is not difficult to see that:
\begin{equation}\label{eq:notdif}
\begin{split}
\Big\| \frac{1}{1 + |\hat x(t)|^{2k}} \big[ F(\hat x), |\hat x(t)|^{2k} \big] \Big\|_{\op} &\leq \sum_{j=0}^{2k-1} C_{j}  |t|^{2k-j} \varepsilon^{2k-j} \sup _{\alpha: |\alpha| = 2k -j}\Big\| \frac{1}{1 + |\hat x(t)|^{2k}} | \hat x(t) |^{j} \partial^{\alpha} F(\hat x)  \Big\|_{\op} \\ &\leq C_{k} \Big(\max_{0 \leq j \leq 2k}\| D^{j} F \|_{\infty}  \Big)|t| \varepsilon(1 + t^{2k}\varepsilon^{2k})\;.
\end{split}
\end{equation}
This concludes the proof of (\ref{eq:conj1}). 
\end{proof}
\begin{remark}\label{rem:F} As a consequence of the assumption in Eq.~(\ref{eq:assump_density}), the function $V*\rho_{t}$ satisfies the hypotheses of Lemma \ref{lem:conj} for $0\leq t\leq T$, $1\leq k \leq \max (2n,3)$. In fact, for $j\leq \max (4n,7)$ we have:
\begin{equation}\label{eq:bdDkV}
\begin{split}
\| D^{j} V*\rho_{t} \|_{\infty} &\leq \| D^{j} V (1 + |\cdot|^{4}) \|_{\infty} \| \mathcal{W}^{(1)} * \rho_{t}  \|_{\infty} \\
&\leq C_{j} \| \mathcal{W}^{(1)} * \rho_{t}  \|_{\infty}\;,
\end{split}
\end{equation}
where we used the non-negativity of the density $\rho_{t}$ and the assumption (\ref{eq:assV}) on the potential $V$. Next, by Eq.~(\ref{eq:densityt}):
\begin{equation}\label{eq:bnd-W*rho}
\begin{split}
\|\mathcal{W}^{(1)} * \rho_{t}\|_{\infty} &= \sup_{z\in \mathbb{R}^{3}} \varepsilon^{3} \int dy\, \frac{1}{1+| z-y |^{4}} \omega_{N,t}(y,y) \\
&\equiv \sup_{z\in \mathbb{R}^{3}} \varepsilon^{3} \tr\, \mathcal{W}_{z}^{(1)} \omega_{N,t} \\
&\leq C_{T}\;,
\end{split}
\end{equation}
which proves the boundedness of $\| D^{j} V*\rho_{t} \|_{\infty}$ for any $j \leq \max( 4 n,7)$. 
\end{remark}
We are now ready to prove Proposition \ref{prp:propU}.
\begin{proof}[Proof of Proposition \ref{prp:propU}] Let $U_{p}^{0}(t;s) = e^{i( \varepsilon^{2}\Delta - i\varepsilon^{2} p \cdot \nabla ) (t-s)/\varepsilon}$
be the modified free dynamics and notice that
\begin{equation*}
U_{p}^{0}(t;s)^{\ast} \hat{x} U_{p}^{0}(t;s) = \hat{x}(t-s) + \veps p(t-s) \;.
\end{equation*}
We introduce the Hartree dynamics in the interaction picture as:
\begin{equation}\label{eq:intpicdef}
U^{\text{I}}_{p}(t;s) := U_{p}^{0}(t;0)^{*} U_{p}(t;s) U_{p}^{0}(s;0)\;,
\end{equation}
satisfying the following evolution equation:
\begin{equation}\label{eq:genUI}
i \veps \partial_{t} U^{\text{I}}_{p}(t;s) = U^{0}_{p}(t;0)^{*} (V*\rho_{t}) U^{0}_{p}(t;0) U^{\text{I}}_{p}(t;s), \qquad U^{\text{I}}_{p}(s;s) = \mathbf{1}\;.
\end{equation}
Let us start by proving the bound (\ref{eq:propU1}). Recalling \eqref{eq:cC}, we write:
\begin{equation}\label{eq:genUI2}
\begin{split}
U_{p}(t;s)^{*} \mathcal{W}^{(k/2)}_{z}(t_{0})^{2} U_{p}(t;s) &= U_{p}(t;s)^{*} U_{p}^{0}(t;0) U^{0}_{p}(t;0)^{*} \mathcal{W}^{(k/2)}_{z}(t_{0})^{2} U_{p}^{0}(t;0) U_{p}^{0}(t;0)^{*} U_{p}(t;s) \\ 
&\equiv U^{0}_{p}(s;0) U^{\text{I}}_{p}(t;s)^{*} \mathcal{W}^{(k/2)}_{z + \veps  pt}(t_{0} + t)^{2} U^{\text{I}}_{p}(t;s) U^{0}_{p}(s;0)^{*}\;.
\end{split}
\end{equation}
By the arbitrariness of $z$ and $t_{0}$, we will focus on the following operator:
\begin{equation*}
U^{\text{I}}_{p}(t;s)^{*} \mathcal{W}^{(k/2)}_{z}(t_{0})^{2} U^{\text{I}}_{p}(t;s)\;.
\end{equation*}
Let $\phi \in L^{2}(\mathbb{R}^{3})$. We write:
\begin{equation}\label{eq:interp}
\begin{split}
\langle \phi, U^{\text{I}}_{p}(t;s)^{*} \mathcal{W}^{(k/2)}_{z}(t_{0})^{2} U^{\text{I}}_{p}(t;s) \phi \rangle &= \langle \phi,  \mathcal{W}^{(k/2)}_{z}(t_{0})^{2}  \phi \rangle\\
& -\frac{i}{\varepsilon} \int_{s}^{t}d\tau\, i\varepsilon \partial_{\tau} \langle \phi, U^{\text{I}}_{p}(\tau;s)^{*} \mathcal{W}^{(k/2)}_{z}(t_{0})^{2} U^{\text{I}}_{p}(\tau;s) \phi \rangle\;.
\end{split}
\end{equation}
To compute the derivative, recall (\ref{eq:genUI}). We get:
\begin{equation}\label{eq:der}
\begin{split}
i\varepsilon \partial_{\tau} \langle \phi, U^{\text{I}}_{p}(\tau;s)^{*} &\mathcal{W}^{(k/2)}_{z}(t_{0})^{2} U^{\text{I}}_{p}(\tau;s) \phi \rangle \\ &= \langle U^{\text{I}}_{p}(\tau;s)\phi, [ \mathcal{W}^{(k/2)}_{z}(t_{0})^{2}, U^{0}_{p}(\tau;0)^{*} (V*\rho_{\tau}) U^{0}_{p}(\tau;0)] U^{\text{I}}_{p}(\tau;s) \phi  \rangle\\
& = \langle U^{0}_{p}(\tau;0) U^{\text{I}}_{p}(\tau;s)\phi, [ \mathcal{W}^{(k/2)}_{z - \veps p\tau}(t_{0} - \tau)^{2}, V*\rho_{\tau} ] U^{0}_{p}(\tau;0) U^{\text{I}}_{p}(\tau;s) \phi  \rangle\;.
\end{split}
\end{equation}
Let us now focus on the commutator appearing in the scalar product. We introduce the short-hand notations:
\begin{equation}
\label{notation_weight_dynamics}
\widetilde{x}:= \hat{x}(t_{0}-\tau) - z + \veps p \tau \;, \qquad \W^{(k/2)}:= \W_{z - \veps p \tau}^{(k/2)}(t_{0}-\tau) \;.
\end{equation}
We write:
\begin{equation}\label{eq:aa}
[ \big(\mathcal{W}^{(k/2)}\big)^{2}, V*\rho_{\tau} ] = \mathcal{W}^{(k/2)} [ \mathcal{W}^{(k/2)}, V*\rho_{\tau} ] + [ \mathcal{W}^{(k/2)}, V*\rho_{\tau} ]  \mathcal{W}^{(k/2)}\;.
\end{equation}
Consider the first term. We rewrite:
\begin{equation*}
\begin{split}
\big| \langle \phi, \mathcal{W}^{(k/2)} [ \mathcal{W}^{(k/2)}, V*\rho_{\tau} ] \phi\rangle \big| &\leq \big\| \mathcal{W}^{(k/2)} \phi  \big\|  \big\| [ \mathcal{W}^{(k/2)}, V*\rho_{\tau} ] \phi \big\|   \\
&= \big\| \mathcal{W}^{(k/2)} \phi  \big\| \big\| \big( \mathcal{W}^{(k/2)} (V*\rho_{\tau}) (\mathcal{W}^{(k/2)})^{-1} - V*\rho_{\tau} \big) \mathcal{W}^{(k/2)} \phi \big\| \\
&\leq \big\| \big( \mathcal{W}^{(k/2)} (V*\rho_{\tau}) (\mathcal{W}^{(k/2)})^{-1}  - V*\rho_{\tau} \big) \big\|_{\text{op}} \big\| \mathcal{W}^{(k/2)} \phi  \big\|^{2}\;.
\end{split}
\end{equation*}
The function $V*\rho_{\tau}$ satisfies the assumptions of Lemma \ref{lem:conj}, recall Remark \ref{rem:F}. Hence, recalling also (\ref{eq:cC}):
\begin{equation}
\label{eq:use-bnd-W*rho}
\begin{split}
\big| \langle \phi, \mathcal{W}^{(k/2)} [ \mathcal{W}^{(k/2)}, V*\rho_{\tau} ] \phi\rangle \big| &\leq C \varepsilon \| \mathcal{W}^{(1)} * \rho_{\tau}\|_{\infty} |t_{0} - \tau|  \langle \phi, \mathcal{W}^{(k)} \phi  \rangle \\
&\leq K \varepsilon |t_{0} - \tau|  \langle \phi, \mathcal{W}^{(k)} \phi  \rangle\;.
\end{split}
\end{equation}
The same bound holds for the second term in (\ref{eq:aa}).
Therefore:
\begin{equation}\label{eq:estWgro}
\begin{split}
\Big|\langle \phi, \big[ \big(\W^{(k/2)}\big)^{2}, V \ast \rho_{\tau}\big] \phi \rangle \Big|  
&\leq C  \veps |t_{0} - \tau| \langle \phi, \mathcal{W}^{(k)} \phi \rangle\;.
\end{split}
\end{equation}
Let us now come back to (\ref{eq:interp}). The identity (\ref{eq:der}) implies:
\begin{equation*}
\begin{split}
\langle \phi, & U^{\text{I}}_{p}(t;s)^{*} \mathcal{W}^{(k/2)}_{z}(t_{0})^{2} U^{\text{I}}_{p}(t;s) \phi \rangle \leq  \langle \phi,  \mathcal{W}^{(k/2)}_{z}(t_{0})^{2}  \phi \rangle \nonumber\\
& + \frac{C}{\varepsilon}\int_{s}^{t} d\tau\, \Big|\langle U^{0}_{p}(\tau;0) U^{\text{I}}_{p}(\tau;s)\phi, [ \mathcal{W}^{(k/2)}_{z - \veps p\tau}(t_{0} - \tau)^2, V*\rho_{\tau} ] U^{0}_{p}(\tau;0) U^{\text{I}}_{p}(\tau;s) \phi  \rangle\Big|\;,
\end{split}
\end{equation*}
which we estimate as, using the bound (\ref{eq:estWgro}) with $\phi$ replaced by $U^{0}_{p}(\tau;0) U^{\text{I}}_{p}(\tau;s) \phi$:
\begin{equation*}
\begin{split}
\langle \phi, &U^{\text{I}}_{p}(t;s)^{*} \mathcal{W}^{(k/2)}_{z}(t_{0})^2 U^{\text{I}}_{p}(t;s) \phi \rangle  \leq  \langle \phi,  \mathcal{W}^{(k/2)}_{z}(t_{0})^2  \phi \rangle \\ &+ C\int_{s}^{t} d\tau\, |t_{0} - \tau| \langle U^{0}_{p}(\tau;0) U^{\text{I}}_{p}(\tau;s)\phi , \mathcal{W}^{(k)}_{z - \veps p\tau}(t_{0} - \tau) U^{0}_{p}(\tau;0) U^{\text{I}}_{p}(\tau;s) \phi\rangle\;.
\end{split}
\end{equation*}
Using that
\begin{equation*}
\begin{split}
\langle U^{0}_{p}(\tau;0) &U^{\text{I}}_{p}(\tau;s) \phi, \mathcal{W}^{(k)}_{z - \veps p\tau}(t_{0} - \tau) U^{0}_{p}(\tau;0) U^{\text{I}}_{p}(\tau;s)\phi \rangle  = \langle U^{\text{I}}_{p}(\tau;s)\phi, \mathcal{W}^{(k)}_{z}(t_{0}) U^{\text{I}}_{p}(\tau;s) \phi\rangle\;,
\end{split}
\end{equation*}
and recalling (\ref{eq:cC}), we get, for all $0\leq s,t\leq T$, by the Gronwall lemma:
\begin{equation}\label{eq:groW}
\langle \phi, U^{\text{I}}_{p}(t;s)^{*} \mathcal{W}^{(k)}_{z}(t_{0}) U^{\text{I}}_{p}(t;s) \phi \rangle \leq C \langle \phi,  \mathcal{W}^{(k)}_{z}(t_{0})  \phi \rangle\;,
\end{equation}
where the constant $C$ depends on $t_0$, $T$ and $n$.
Going back to (\ref{eq:genUI2}), we have, for all $\phi \in L^{2}(\mathbb{R}^{3})$:
\begin{equation*}
\begin{split}
\langle \phi, U_{p}(t;s)^{*} \mathcal{W}^{(k)}_{z}(t_{0}) U_{p}(t;s) \phi \rangle & = \langle \phi, U^{0}_{p}(s;0) U^{\text{I}}_{p}(t;s)^{*} \mathcal{W}^{(k)}_{z + \veps pt}(t_{0} + t) U^{\text{I}}_{p}(t;s) U^{0}_{p}(s;0)^{*} \phi \rangle \\
& \leq C \langle \phi, U^{0}_{p}(s;0) \mathcal{W}^{(k)}_{z +\veps  pt}(t_{0} + t) U^{0}_{p}(s;0)^{*} \phi \rangle \\
& = C \langle \phi, \mathcal{W}^{(k)}_{z + \veps p(t-s)}(t_{0} + t - s)  \phi \rangle\;,
\end{split}
\end{equation*}
where the inequality follows from (\ref{eq:groW}) with $z$ replaced by $z + \veps pt$ and $t_{0}$ replaced by $t_{0} + t$. This proves the bound in (\ref{eq:propU1}). 
\end{proof}
To conclude this section, we prove Proposition \ref{prp:density}. The proof is a simple adaptation of the one of Proposition \ref{prp:propU}.
\begin{proof}[Proof of Proposition \ref{prp:density}.] Let $\omega_{N,t} = \sum_{j=1}^{N} |\phi_{j,t}\rangle \langle \phi_{j,t}|$, and take $0\leq t \leq T_{*}$, with $T_{\ast}\leq T_{1}$ to be suitably chosen. We start by writing:
\begin{equation}
\label{eq:trWw-Unit}
\begin{split}
\tr\, \mathcal{W}^{(1)}_{z} \omega_{N,t} &= \sum_{j=1}^{N} \langle \phi_{j,t}, \mathcal{W}^{(1)}_{z}  \phi_{j,t}\rangle \\
&\equiv \sum_{j=1}^{N} \langle \phi_{j}, U(t;0)^{*} \mathcal{W}^{(1)}_{z} U(t;0) \phi_{j} \rangle\;,
\end{split}
\end{equation}
where $U(t;0) = U_{p=0}(t;0)$ is the Hartree dynamics, see \eqref{eq:Hartree-dyn}.
Let $t_{0}\in [0, T_{\ast}]$. Consider the following quantity:
\begin{equation*}
\langle \phi_{j}, U^{\text{I}}(t;0)^{*} \mathcal{W}^{(1)}_{z}(t_{0}) U^{\text{I}}(t;0) \phi_{j} \rangle\;.
\end{equation*}
Recall that, by definition of the evolution in the interaction picture, Eq.~(\ref{eq:intpicdef}), 
\begin{equation}
\label{eq:manip-UWU}
\langle \phi_{j}, U^{\text{I}}(t;0)^{*} \mathcal{W}^{(1)}_{z}(t) U^{\text{I}}(t;0) \phi_{j} \rangle \equiv \langle \phi_{j}, U(t;0)^{*} \mathcal{W}^{(1)}_{z} U(t;0) \phi_{j} \rangle\;.
\end{equation}
Next, by the proof of Proposition \ref{prp:propU}, it is not difficult to see that, for a suitable constant $K>0$:
\begin{equation}\label{eq:phij}
\begin{split}
&\langle \phi_{j}, U^{\text{I}}(t;0)^{*} \mathcal{W}^{(1)}_{z}(t_{0}) U^{\text{I}}(t;0) \phi_{j} \rangle \leq \langle \phi_{j},  \mathcal{W}^{(1)}_{z}(t_{0}) \phi_{j} \rangle \\
&\qquad + \int_{0}^{t} d\tau\, K T_{\ast} \varepsilon^{3} \Big( \sup_{z\in \mathbb{R}^{3}} \tr\, \mathcal{W}^{(1)}_{z} \omega_{N,\tau} \Big) \langle \phi_{j}, U^{\text{I}}(\tau;0)^{*} \mathcal{W}^{(1)}_{z}(t_{0}) U^{\text{I}}(\tau;0) \phi_{j} \rangle\;.
\end{split}
\end{equation}
Defining
\begin{equation*}
\alpha(t;t_{0}) := \varepsilon^{3} \sup_{z\in \mathbb{R}^{3}}\sum_{j=1}^{N} \langle \phi_{j}, U^{\text{I}}(t;0)^{*} \mathcal{W}^{(1)}_{z}(t_{0}) U^{\text{I}}(t;0) \phi_{j} \rangle\;,
\end{equation*}
Eq.~(\ref{eq:phij}) implies that:
\begin{equation}\label{eq:alpha}
\begin{split}
\alpha(t;t_{0}) &\leq \alpha(0;t_{0}) + K T_{\ast}\int_{0}^{t}d\tau\, \Big(  \varepsilon^{3} \sup_{z\in \mathbb{R}^{3}} \tr \, \mathcal{W}^{(1)}_{z} \omega_{N,\tau}\Big) \alpha(\tau;t_{0}) \\
&= \alpha(0;t_{0}) + K T_{\ast}\int_{0}^{t} d\tau\,  \alpha(\tau; \tau) \alpha(\tau; t_{0})\;.
\end{split}
\end{equation}
Let:
\begin{equation*}
f(t) := \sup_{t_{0} \in [0,T_{\ast}]} \alpha(t;t_{0})\;.
\end{equation*}
Notice that
\begin{equation}\label{eq:notice}
\tr\, \mathcal{W}^{(1)}_{z} \omega_{N,t} \leq f(t)\;,
\end{equation}
see Eqs.~\eqref{eq:trWw-Unit} and \eqref{eq:manip-UWU},
hence our goal will be to derive an estimate for $f(t)$. Eq.~(\ref{eq:alpha}) implies:
\begin{equation}\label{eq:integraleq}
f(t) \leq f(0) + K T_{\ast} \int_{0}^{t} d\tau\, f(\tau)^{2}\;.
\end{equation}
The quantity $f(0)$ is bounded as follows:
\begin{equation*}
\begin{split}
f(0) &= \varepsilon^{3} \sup_{t_{0}\in [0,T_{\ast}]} \sup_{z\in \mathbb{R}^{3}} \tr\, \mathcal{W}^{(1)}_{z}(t_{0}) \omega_{N}  \leq C\;,
\end{split}
\end{equation*}
where the last bound follows from the assumption (\ref{eq:assump_density}) on the initial datum, since $T_{\ast}\leq T_{1}$. Next, let us denote by $g(t)$ the r.h.s.~of (\ref{eq:integraleq}), so that (\ref{eq:integraleq}) reads $f(t) \leq g(t)$. Also,
\begin{equation*}
g'(t) = K T_{\ast} f(t)^{2} \leq K T_{\ast} g(t)^{2}\;.
\end{equation*}
Equivalently,
\begin{equation*}
\frac{d}{dt} \Big( -\frac{1}{g(t)} - K T_{\ast} t \Big) \leq 0\;.
\end{equation*}
Since $g(0) = f(0)$, we have:
\begin{equation*}
g(t) \leq \frac{f(0)}{1 - f(0) K T_{\ast} t} \leq \frac{f(0)}{1 - f(0) K T_{\ast}^{2}} 
\end{equation*}
for $0\leq t\leq T_{\ast}$, choosing $T_{\ast}$ so that the denominator is positive. This bound, together with (\ref{eq:notice}), proves the final claim with, e.g., $T_{*} = \min\{\frac{1}{2}(f(0) K)^{-1/2}, T_{1}\}$. Notice that the constant $K$ depends on the potential $V$, and can be made arbitrarily small by taking $V$ small enough.
\end{proof}

\subsection{Proof of Theorem \ref{thm:propcomm}}\label{sec:propcomm}

In this section we shall consider initial data satisfying Assumptions \ref{ass:sc}, and we shall prove
the stability of the local bounds in the assumptions \eqref{eq:propcomm1} and \eqref{eq:H2} under the Hartree flow. 
The proof of \eqref{eq:propagaz2} follows straightforwardly by application of Proposition \ref{prp:propU} and of Lemma \ref{lem:mon}, in fact for $0\leq t\leq T $:
\begin{equation}\label{eq:proof41first}
\begin{split}
 \big\|\W_{z}^{(n)}  \, \omega_{N,t} \big\|_{\tr}  &\leq  \big\|\W_{z}^{(n)} U(t;0)  \, \omega_{N} \big\|_{\tr} 
\\
&\leq C \big\|\W_{z}^{(n)}(t)   \, \omega_{N} \big\|_{\tr} \leq C \veps^{-3} \;,
\end{split}
\end{equation}
where we used the invariance of the trace norm under unitary conjugation and \eqref{eq:H2}.

Our strategy for proving \eqref{eq:propagazione} also relies on controlling the regularized Hartree evolution of the localization operators by their free evolution. The proof will be divided in a few steps.
\medskip

\noindent{\underline{Part 1: Setting up the Gronwall estimate.}} Our goal is to control the following quantity 
\begin{equation}\label{eq:weipsw-start}
\sup _{s \in [t,T]}\sup_{q: |q|\leq \varepsilon^{-1}} \sup_{z\in \mathbb{R}^{3}} \frac{X_\Lambda(z)}{1+|q|} \Big \|\W_{z}^{(n)}(s-t) \,\big[e^{i q \cdot \hat{x}}, \omega_{N,t}\big] \Big \|_{\tr}
\end{equation}
using a Gronwall-type strategy. The claim in the theorem corresponds to the special case $s = t$. Let us introduce the modified Hartree dynamics $U_{q}(t;s)$ as in Proposition \ref{prp:propU}:
\begin{equation*}
i\varepsilon \partial_{t} U_{ q}(t;s) = (-\varepsilon^{2}\Delta + V * \rho_{t} + i\varepsilon^{2} q\cdot\nabla) U_{ q}(t;s)\;,\qquad U_{q}(s;s) = 1\;.
\end{equation*}
Following \cite[Section 5]{BPS}, we have:
\begin{equation}\label{eq:BPS}
i\varepsilon \partial_{t} U_{q}(t;s)^{*} \big[e^{iq \cdot \hat{x}}, \omega_{N,t}\big]U_{-q}(t;s) =i U_{ q}(t;s)^{*}\big\{ e^{iq \cdot \hat{x}}, \varepsilon q [ \varepsilon \nabla, \omega_{N,t} ]\big\} U_{-q}(t;s)\;.
\end{equation}
Therefore, taking $s=0$, we get:
\begin{equation*}
\begin{split}
U_{q}(t;0)^{*}\big[e^{iq \cdot \hat{x}}, \omega_{N,t}\big] U_{ -q}(t;0) &= \big[e^{iq \cdot \hat{x}}, \omega_{N,0}\big] \\&\quad - \frac{i}{\varepsilon} \int_{0}^{t}d\tau\, i \varepsilon \partial_{\tau} \Big( U_{ q}(\tau;0)^{*}\big[e^{iq\cdot \hat{ x}}, \omega_{N,\tau}\big] U_{ -q}(\tau;0)\Big)\;,
\end{split}
\end{equation*}
which gives, using (\ref{eq:BPS}):
\begin{equation}\label{eq:BPS2}
\begin{split}
\big[e^{iq \cdot \hat{x}}, \omega_{N,t}\big] &= U_{ q}(t;0) \big[e^{iq \cdot \hat{x}}, \omega_{N}\big] U_{-q}(t;0)^{*} \\
&\quad +\frac{1}{\varepsilon}\int_{0}^{t}d\tau\, U_{ q}(\tau;t)^{*}\big\{ e^{iq \cdot \hat{x}}, \varepsilon q \,[ \varepsilon \nabla, \omega_{N,\tau} ]\big\} U_{ -q}(\tau;t)\;.
\end{split}
\end{equation}
We shall now plug this identity into \eqref{eq:weipsw-start}, and estimate the various terms. Consider the one due to the first term on the r.h.s.~of (\ref{eq:BPS2}).
Using Proposition \ref{prp:propU}, Eq.~(\ref{eq:propU1}), Lemma \ref{lem:mon} and the invariance of the trace under unitary conjugation, we have:
\begin{equation}
\label{eq:bnd-hWU[]UWh}
\begin{split}
\Big \| \W_{z}^{(n)}(s-t) &U_{q}(t;0) \big[e^{iq\cdot \hat{x}}, \omega_{N}\big] U_{ -q}(t;0)^{*} \Big \|_{\tr} 
\leq C \Big \|  \W_{z + \veps q t}^{(n)}(s) \big[e^{iq\cdot \hat{x}}, \omega_{N}\big] \Big \|_{\tr}\;.
\end{split}
\end{equation}
We notice that for all $t\leq T$
\begin{equation*}
X_{\Lambda}(z) \leq C (1+\veps ^{4}|q|^{4})X_{\Lambda}(z+ \veps qt) \;.
\end{equation*}
Accordingly, we get:
\begin{equation}\label{eq:Pbd}
\begin{split}
& \sup_{z \in \mathbb{R}^{3}}\,  X_{\Lambda}(z) \Big \|  \W_{z + \veps q t}^{(n)}(s) \big[e^{iq\cdot \hat{x}}, \omega_{N}\big] \Big \|_{\tr} 
\\
& \quad  \leq C (1+\veps ^{4}|q|^{4}) \sup_{z \in \mathbb{R}^{3}}  X_{\Lambda}(z)\Big \|   \W_{z}^{(n)}(s) \big[e^{iq\cdot \hat{x}}, \omega_{N}\big] \Big \|_{\tr}
\\
& \quad  \leq C (1+\veps^{4} |q|^4)(1+|q|) \veps^{-2} \;,
\end{split}
\end{equation}
where the last step follows from assumption \eqref{eq:propcomm1}, since $t\leq s\leq T$. 
This bound, combined with \eqref{eq:bnd-hWU[]UWh}, shows that:
\begin{equation}\label{eq:firstend}
\begin{split}
\sup_{z \in \mathbb{R}^{3}}\, &X_{\Lambda}(z) \Big \| \W_{z}^{(n)} (s-t)U_{q}(t;0) \big[e^{iq\cdot \hat{x}}, \omega_{N}\big] U_{-q}(t;0)^{*} \Big \|_{\tr} \\
& \leq C(1 + |q|) \varepsilon^{-2}\;.
\end{split}
\end{equation}
Next, let us consider the contribution due to the second term in (\ref{eq:BPS2}), namely:
\begin{equation*}
\int_{0}^{t}d\tau\, \Big \|  \W_{z}^{(n)}(s-t) U_{ q}(\tau;t)^{*}\big\{ e^{iq\cdot \hat{x}}, q\cdot [ \varepsilon \nabla, \omega_{N,\tau} ]\big\} U_{ -q}(\tau;t) \Big \|_{\tr}\;.
\end{equation*}
By Proposition \ref{prp:propU} and Lemma \ref{lem:mon}, we have:
\begin{equation}\label{eq:anticomm}
\begin{split}
&\Big \|   \W_{z}^{(n)}(s-t) U_{q}(\tau;t)^{*}\big\{ e^{iq\cdot \hat{x}}, q\cdot [ \varepsilon \nabla, \omega_{N,\tau} ]\big\} U_{-q}(\tau;t) \Big \|_{\tr} 
\\
&
\leq C \Big \| \W_{z + \veps q(t- \tau)}^{(n)}(s- \tau) \big\{ e^{iq\cdot \hat{x}}, q\cdot [ \varepsilon \nabla, \omega_{N,\tau} ]\big\} \Big \|_{\tr}
\end{split}
\end{equation}
The two contributions to the anticommutator are estimated in the same way. For instance, consider:
\begin{equation*}
\begin{split}
&\Big \|   \W_{z +\veps q(t- \tau)}^{(n)}(s - \tau) e^{iq\cdot \hat{x}} q\cdot [ \varepsilon \nabla, \omega_{N,\tau} ] \Big \|_{\tr} \\
&\quad 
\leq \Big \|  \W_{z + \veps q(t - \tau) - 2q\varepsilon(s-\tau)}^{(n)}(s - \tau) \, q\cdot [ \varepsilon \nabla, \omega_{N,\tau} ] \Big \|_{\tr}\;,
\end{split}
\end{equation*}
where in the localization operator we used that $e^{- i q \cdot \hat{x}}\hat{x}(s-\tau) e^{ i q \cdot \hat{x}} = \hat{x}(s-\tau)+2 \veps q (s-\tau)$.
Since for $0\leq \tau \leq t\leq s\leq T$ 
\begin{equation}\label{eq:change-Xloc}
X_{\Lambda}(z) \leq C (1 + \veps^{4}|q|^{4}) X_{\Lambda}(z + \veps q(t - \tau) - 2q\varepsilon(s-\tau))\;,
\end{equation}
we have:
\begin{equation}
\label{eq:secondend}
\begin{split}
&\sup_{z\in \mathbb{R}^{3}} X_{\Lambda}(z) \Big \| \W_{z + \veps q(t- \tau)}^{(n)}(s - \tau) e^{iq\cdot \hat{x}} q\cdot [ \varepsilon \nabla, \omega_{N,\tau} ]\Big \|_{\tr} \\
&\quad \leq C (1 + \veps^{4}|q|^{4}) 
\sup_{z \in \mathbb{R}^{3}} X_{\Lambda}(z) \Big \|   \W_{z}^{(n)}(s - \tau) q\cdot [ \varepsilon \nabla, \omega_{N,\tau} ] \Big \|_{\tr}\;.
\end{split}
\end{equation}
A similar bound holds for the second contribution to the anticommutator in (\ref{eq:anticomm}). Hence, combining (\ref{eq:BPS2}), (\ref{eq:firstend}), (\ref{eq:secondend}), we obtain:
\begin{equation}\label{eq:obtained}
\begin{split}
\sup_{z \in \mathbb{R}^{3}}\, & X_{\Lambda}(z) \Big \| \W_{z}^{(n)}(s-t) \,\big[e^{i q \cdot \hat{x}}, \omega_{N,t}\big] \Big \|_{\tr} \\
&\leq C (1 + |q|) \varepsilon^{-2}+C(1 + \veps^{4}|q|^{4 })|q| \,
\int_{0}^{t}d\tau\, \sup_{z\in \mathbb{R}^{3}} X_{\Lambda}(z)\;
\Big \| \W_{z}^{(n)}(s-\tau)  [ \varepsilon \nabla, \omega_{N,\tau} ]  \Big \|_{\tr} 
\;.
\end{split}
\end{equation}
\noindent{\underline{Part 2: Control of the convolution.}} Our last task is to estimate the argument of the integral in Eq.~(\ref{eq:obtained}). As in \cite{BPS}, we start by writing:
\begin{equation*}
i\varepsilon \partial_{\tau} U(\tau;s)^{*} [\varepsilon \nabla, \omega_{N,\tau} ] U(\tau;s) = U_{\chi}(\tau;s)^{*} \big[ \omega_{N,\tau}, [V*\rho_{\tau}, \varepsilon \nabla]  \big] U(\tau;s)\;,
\end{equation*}
which gives the identity:
\begin{equation}\label{eq:inteq}
[\varepsilon \nabla, \omega_{N,\tau} ] = U(\tau;0) [\varepsilon \nabla, \omega_{N} ] U(\tau;0)^{*} - \frac{i}{\varepsilon} \int_{0}^{\tau}d\eta\, U(\eta;\tau)^{*} \big[ \omega_{N,\eta}, [V_{\chi}*\rho_{\eta}, \varepsilon \nabla ]\big] U(\eta;\tau)\;.
\end{equation}
We plug this identity in the integrand in (\ref{eq:obtained}). We get:
\begin{equation}\label{eq:commgrad}
\begin{split}
\Big \|  \W_{z}^{(n)}(s-\tau)  &[ \varepsilon \nabla,\omega_{N,\tau} ] \Big \|_{\tr}
\\
&
 \leq \Big \| \W_{z}^{(n)}(s-\tau)  U(\tau;0) [\varepsilon \nabla, \omega_{N} ] U(\tau;0)^{*} \Big \|_{\tr} \\
& 
\quad
+ \frac{1}{\varepsilon} \int_{0}^{\tau}d\eta\, \Big \|  \W_{z}^{(n)}(s-\tau)  U(\eta;\tau)^{*} \big[ \omega_{N,\eta}, [V* \rho_{\eta}, \varepsilon \nabla]  \big] U(\eta;\tau) \Big \|_{\tr}\;.
\end{split}
\end{equation}
Consider the first term on the r.h.s.~of (\ref{eq:commgrad}). We have, by Proposition \ref{prp:propU}, Lemma \ref{lem:mon} and the invariance of the trace under unitary conjugation:
\begin{equation*}
\begin{split}
\Big \| \W_{z}^{(n)}(s-\tau)  & U(\tau;0) [\varepsilon \nabla, \omega_{N} ] U(\tau;0)^{*} \Big \|_{\tr} \leq C\Big \|  \W_{z}^{(n)}(s) [\varepsilon \nabla,\omega_{N} ] \Big \|_{\tr}\;.
\end{split}
\end{equation*}
Hence, the contribution to the integrand in \eqref{eq:obtained} due to the first term on the r.h.s.~of (\ref{eq:commgrad}) is bounded by, for $t\leq s\leq T$:
\begin{equation*}
\begin{split}
\sup_{z \in \mathbb{R}^{3}} X_{\Lambda}(z) & \Big \| \W_{z}^{(n)}(s-\tau)  U(\tau;0) [\varepsilon \nabla, \omega_{N} ] U(\tau;0)^{*} \Big \|_{\tr} \\
&
\leq C\sup_{z \in \mathbb{R}^{3}} X_{\Lambda}(z) \Big\| \W_{z}^{(n)}(s) [\varepsilon \nabla, \omega_{N} ] \Big \|_{\tr} \\
&\leq K \varepsilon^{-2}\;,
\end{split}
\end{equation*}
where in the last step we used the assumption \eqref{eq:propcomm2}. Let us consider now the second term in (\ref{eq:commgrad}). Again using Proposition \ref{prp:propU} and Lemma \ref{lem:mon}:
\begin{equation*}
\begin{split}
\Big \| & \W_{z}^{(n)}(s-\tau)  U(\eta;\tau)^{*} [ \omega_{N,\eta}, [V*\rho_{\eta}, \varepsilon \nabla]  ] U(\eta;\tau)\Big \|_{\tr} 
\\
&\leq C \Big \| \W_{z}^{(n)}(s - \eta)   [ \omega_{N,\eta}, [V*\rho_{\eta}, \varepsilon \nabla]  ] \Big \|_{\tr} \\
&
\equiv C\varepsilon \Big \|  \W_{z}^{(n)}(s - \eta)   [ \omega_{N,\eta}, \nabla V*\rho_{\eta}] \Big \|_{\tr}\;,
\end{split}
\end{equation*}
Therefore, the integrand in (\ref{eq:obtained}) is bounded as:
\begin{equation}\label{eq:hh}
\begin{split}
&\sup_{z \in \mathbb{R}^{3}} X_{\Lambda}(z) \;
\Big \|  \W_{z}^{(n)}(s-\tau)  [ \varepsilon \nabla, \omega_{N,\tau} ]\Big \|_{\tr}
 \\
 &
\qquad \leq  
  K \veps^{-2} + C \int_{0} ^{\tau} d \eta \sup _{z \in \mathbb{R}^{3}} X_{\Lambda}(z)  \Big \| \W_{z}^{(n)}(s-\eta)   [ \omega_{N,\eta}, \nabla V*\rho_{\eta} ] \Big \|_{\tr} \;.
\end{split}
\end{equation}
To estimate the latter integrand, we write:
\begin{equation}\label{eq:W[w,DVr]}
\begin{split}
\Big \| \W_{z}^{(n)}&(s-\eta)   [ \omega_{N,\eta}, \nabla V*\rho_{\eta} ] \Big \|_{\tr} 
\\
& = \Big \|\int d y \rho_{\eta}(y) \W_{z}^{(n)}(s-\eta) [ \omega_{N,\eta}, \nabla V(\hat{x}-y)] \Big \|_{\tr}
\\
&\leq
\Big \|\int d y \rho_{\eta}(y) \W_{z}^{(n)}(s-\eta) \W_{y}^{(1)}  [ \omega_{N,\eta}, F_{y}(\hat{x})] \Big \|_{\tr}
\\
& \qquad\qquad
+ \Big \|\int d y \rho_{\eta}(y) \W_{z}^{(n)}(s-\eta) [ \omega_{N,\eta},\W_{y}^{(1)} ]F_{y}(\hat{x}) \Big \|_{\tr} \;.
\end{split}
\end{equation}
where we let $F_{y}(\hat{x}):=\big(\W_{y}^{(1)}\big)^{-1} \nabla V(\hat{x}-y) $. We estimate the first term on the right-hand side as:
\begin{equation}\label{eq:rWW[w,F]}
\begin{split}
&\Big \|\int d y \rho_{\eta}(y) \W_{z}^{(n)}(s-\eta) \W_{y}^{(1)}  [ \omega_{N,\eta}, F_{y}(\hat{x})] \Big \|_{\tr}
\\
& \qquad
 \leq
 \int d p  \big| \hat{F}(p)\big|
  \Big \|\int d y \rho_{\eta}(y) e^{-i p \cdot y}  \W_{z}^{(n)}(s-\eta)\W_{y}^{(1)}    [ \omega_{N,\eta}, e^{i p \cdot \hat{x}}] \Big \|_{\tr}
  \\
  &\qquad
  =\int d p  \big| \hat{F}(p)\big|
  \Big \| \W_{z}^{(n)}(s-\eta)\big(\W^{(1)} *\rho_\eta^{p}\big)    [ \omega_{N,\eta}, e^{i p \cdot \hat{x}}] \Big \|_{\tr}
  \\
  & \qquad \leq 
  \int d p  \big| \hat{F}(p)\big|
  \Big \| \W_{z}^{(n)}(s-\eta)\big(\W^{(1)} *\rho_\eta^{p}\big) \big[\W_{z}^{(n)}(s-\eta)\big]^{-1} \Big\|_{\op}  \Big\| \W_{z}^{(n)}(s-\eta) [ \omega_{N,\eta}, e^{i p \cdot \hat{x}}] \Big \|_{\tr}
  \\
  & \qquad \leq \bigg(\sup _{p \in \mathbb{R}^{3}}\Big \| \W_{z}^{(n)}(s-\eta)\big(\W^{(1)} *\rho_\eta^{p}\big) \big[\W_{z}^{(n)}(s-\eta)\big]^{-1} \Big\|_{\op}  \bigg) \,\cdot
\\
&\hspace{16em} \cdot  \,
  \int d p  \big| \hat{F}(p)\big|
  \Big\|  \W_{z}^{(n)}(s-\eta)   [ \omega_{N,\eta}, e^{i p \cdot \hat{x}}] \Big \|_{\tr}
\end{split}
\end{equation}
where we set $\rho_{\eta}^{p}(y):=\rho_{\eta}(y) e^{-i p \cdot y}$. To control the operator norm, we write
\begin{equation}\label{eq:wwrhow-1:start}
\begin{split}
& \Big \| \W_{z}^{(n)}(s-\eta)\big(\W^{(1)} *\rho_\eta^{p}\big) \big[\W_{z}^{(n)}(s-\eta)\big]^{-1} \Big\|_{\op}
\\
&\quad \leq \Big \| \W^{(1)} *\rho_\eta^{p} \Big\|_{\op} + \Big \| \W_{z}^{(n)}(s-\eta) \Big[\big(\W^{(1)} *\rho_\eta^{p}\big),|\hat{x}_{z}(s-\eta)|^{4n}  \Big]\Big\|_{\op} \;.
\end{split}
\end{equation}
where we set $\hat{x}_{z}(s-\eta):=\hat{x}(s-\eta)-z $. We write the commutator by moving the powers of $\hat{x}_{z}(s-\eta)$ to the left:
\begin{equation}\label{eq:wwrhow-1:comm}
\begin{split}
\Big \| &\W_{z}^{(n)}(s-\eta) \Big[\big(\W^{(1)} *\rho_\eta^{p}\big),|\hat{x}_{z}(s-\eta)|^{4n}  \Big]\Big\|_{\op}
\\
& \leq \sum_{\alpha:|\alpha|=4n} \sum _{\substack{\alpha ': |\alpha'| \geq 1, \\ \alpha' \leq \alpha }} {\alpha \choose \alpha'}
\Big \| \W_{z}^{(n)}(s-\eta) (\hat{x}_{z}(s-\eta))^{\alpha-\alpha'} \mathrm{ad}_{\hat{x}_{z}(s-\eta)}^{(\alpha')}\big(\W^{(1)} *\rho_\eta^{p}\big)\Big\|_{\op}  
\\
& \leq  \sum_{\alpha:|\alpha|=4n} \sum _{\substack{\alpha ': |\alpha'| \geq 1, \\ \alpha' \leq \alpha }} {\alpha \choose \alpha'}
\Big \| \mathrm{ad}_{\hat{x}_{z}(s-\eta)}^{(\alpha')}\big(\W^{(1)} *\rho_\eta^{p}\big)\Big\|_{\op}  
\;.
\end{split}
\end{equation}
We then estimate the latter operator norm as follows, for $0 \leq |\alpha' | \leq 4n$:
\begin{equation}\label{eq:wwrhow-1:op}
\begin{split}
\Big \| \mathrm{ad}_{\hat{x}_{z}(s-\eta)}^{(\alpha')}\big(\W^{(1)} *\rho_\eta^{p}\big)\Big\|_{\op} &=(2|s-\eta|)^{|\alpha'|} \sup _{z \in \mathbb{R}^{3}} \bigg | \int d y \, \partial _{z}^{\alpha'}\W^{(1)}(z-y) \,\rho _{\eta}(y) e^{-i p \cdot y} \bigg| 
\\
& \leq 
C \| (1+|\cdot|^{4} ) \partial^{\alpha'}\W^{(1)}\|_{\infty} \, \|\W^{(1)} *\rho_{\eta} \|_{\infty}
\\
& \leq C \;,
\end{split}
\end{equation}
where we used Proposition \ref{prp:density} and that $\W^{(1)}$ and $ \rho_{\eta} $ are positive. All in all, putting together the bounds \eqref{eq:rWW[w,F]}, \eqref{eq:wwrhow-1:start}, \eqref{eq:wwrhow-1:comm} and \eqref{eq:wwrhow-1:op}, we have:
\begin{equation}\label{eq:intrWWwF-partial}
\begin{split}
\sup _{z \in \mathbb{R}^{3}}X_{\Lambda}(z)& \Big \|\int d y \rho_{\eta}(y) \W_{z}^{(n)}(s-\eta) \W_{y}^{(1)}  [ \omega_{N,\eta}, F_{y}(\hat{x})] \Big \|_{\tr} 
\\
& \leq C
\int d p  \big| \hat{F}(p)\big|
  \sup _{z \in \mathbb{R}^{3}} X_{\Lambda}(z)\Big\|  \W_{z}^{(n)}(s-\eta)   [ \omega_{N,\eta}, e^{i p \cdot \hat{x}}] \Big \|_{\tr} \;.
\end{split}
\end{equation}
To control this integral, we split it into small and large momenta:
\begin{equation}\label{eq:FW[w,eipx]final}
\begin{split}
 &\int d p  \big| \hat{F}(p)\big|\sup _{z \in \mathbb{R}^{3}} X_{\Lambda}(z)
  \Big\|  \W_{z}^{(n)}(s-\eta)   [ \omega_{N,\eta}, e^{i p \cdot \hat{x}}] \Big \|_{\tr}
  \\
  & =  \int_{|p|\leq \veps^{-1}} d p  \big| \hat{F}(p)\big|\sup _{z \in \mathbb{R}^{3}} X_{\Lambda}(z)
  \Big\|  \W_{z}^{(n)}(s-\eta)   [ \omega_{N,\eta}, e^{i p \cdot \hat{x}}] \Big \|_{\tr} 
\\  
  & \qquad +  \int_{|p|> \veps^{-1}} d p  \big| \hat{F}(p)\big| \sup _{z \in \mathbb{R}^{3}} X_{\Lambda}(z) \Big(  \Big\|  \W_{z}^{(n)}(s-\eta)   \omega_{N,\eta}  \Big \|_{\tr} +
    \Big\|  \W_{z}^{(n)}(s-\eta) e^{i p \cdot \hat{x}}\omega_{N,\eta} \Big \|_{\tr}
\Big) \;.
\end{split}
\end{equation}
The first term on the right-hand side is estimated as follows:
\begin{equation}\label{eq:final-small}
\begin{split}
\int_{|p|\leq \veps^{-1}} d p  \big| \hat{F}(p)\big|& \sup _{z \in \mathbb{R}^{3}} X_{\Lambda}(z)
  \Big\|  \W_{z}^{(n)}(s-\eta)   [ \omega_{N,\eta}, e^{i p \cdot \hat{x}}] \Big \|_{\tr} 
 \\
  &\leq \|(1+|\cdot|)\hat{F} \|_{1} \sup _{p: |p|\leq \veps^{-1}} \frac{X_{\Lambda}(z)}{1+|p|}  \Big\|  \W_{z}^{(n)}(s-\eta)   [ \omega_{N,\eta}, e^{i p \cdot \hat{x}}] \Big \|_{\tr}
  \\
  &\leq C \sup _{p: |p|\leq \veps^{-1}} \frac{X_{\Lambda}(z)}{1+|p|}  \Big\|  \W_{z}^{(n)}(s-\eta)   [ \omega_{N,\eta}, e^{i p \cdot \hat{x}}] \Big \|_{\tr}  \;,
  \end{split}
\end{equation}
where we used the assumptions on the interaction potential \eqref{eq:assV}. 
To estimate the second term on the r.h.s.~of \eqref{eq:FW[w,eipx]final}, we follow the computations in \eqref{eq:change-Xloc} and \eqref{eq:secondend}, to write
\begin{equation}
\sup _{z \in \mathbb{R}^{3}} X_{\Lambda}(z) \Big\|  \W_{z}^{(n)}(s-\eta) e^{i p \cdot \hat{x}}\omega_{N,\eta} \Big \|_{\tr} \leq  C (1+\veps^{4}|p|^4)
\sup _{z \in \mathbb{R}^{3}} X_{\Lambda}(z) \Big\|  \W_{z}^{(n)}(s-\eta) \omega_{N,\eta} \Big \|_{\tr} \;.
\end{equation}
Then, by \eqref{eq:proof41first} and by the assumption on the potential \eqref{eq:assV}, we obtain: 
\begin{equation}\label{eq:final-large}
\begin{split}
&\int_{|p|> \veps^{-1}} d p\,  \big| \hat{F}(p)\big| \sup _{z \in \mathbb{R}^{3}} X_{\Lambda}(z) \Big(  \Big\|  \W_{z}^{(n)}(s-\eta)   \omega_{N,\eta} \Big \|_{\tr} +
    \Big\|  \W_{z}^{(n)}(s-\eta) e^{i p \cdot \hat{x}}\omega_{N,\eta} \Big \|_{\tr}
\Big)
\\
&\quad \leq C
\int_{|p|> \veps^{-1}} d p  \, \veps |p|(1+\veps^{4}|p|^4) \big| \hat{F}(p)\big|  \sup _{z \in \mathbb{R}^{3}} X_{\Lambda}(z) \Big\|  \W_{z}^{(n)}(s-\eta)   \omega_{N,\eta} \Big \|_{\tr} 
\\
& \quad \leq C \veps  \|(1+|\cdot|^{5})\hat{F} \|_{1} \sup _{z \in \mathbb{R}^{3}} X_{\Lambda}(z) \Big\|  \W_{z}^{(n)}(s-\eta)   \omega_{N,\eta} \Big \|_{\tr}
\\
& \quad \leq C \veps^{-2} \;.
\end{split}
\end{equation}
Therefore, putting together \eqref{eq:rWW[w,F]}, \eqref{eq:intrWWwF-partial}, \eqref{eq:FW[w,eipx]final}, \eqref{eq:final-small}, \eqref{eq:final-large}, we have:
\begin{equation}\label{eq:FW[w,eipx]final-est}
\begin{split}
\sup _{z \in \mathbb{R}^{3}}X_{\Lambda}(z)&
\Big \|\int d y \rho_{\eta}(y) \W_{z}^{(n)}(s-\eta) \W_{y}^{(1)}  [ \omega_{N,\eta}, F_{y}(\hat{x})] \Big \|_{\tr}
\\
&
\leq C \veps^{-2} + C \sup _{p: |p|\leq \veps^{-1}} \frac{X_{\Lambda}(z)}{1+|p|}  \Big\|  \W_{z}^{(n)}(s-\eta)   [ \omega_{N,\eta}, e^{i p \cdot \hat{x}}] \Big \|_{\tr} \;.
\end{split}
\end{equation}
To estimate the second term on the r.h.s.~of \eqref{eq:W[w,DVr]}, we proceed in a similar way:
\begin{equation*}
\begin{split}
&\Big \|\int d y \rho_{\eta}(y) \W_{z}^{(n)}(s-\eta) [ \omega_{N,\eta},\W_{y}^{(1)} ]F_{y}(\hat{x}) \Big \|_{\tr}
\\
& \qquad
 \leq
 \int d p  \big| \widehat{\W^{(1)}}(p)\big|
  \Big \|\int d y \rho_{\eta}(y) e^{-i p \cdot y}  \W_{z}^{(n)}(s-\eta)    [ \omega_{N,\eta}, e^{i p \cdot \hat{x}}]
  F_{y}(\hat{x}) \Big \|_{\tr}
  \\
  & \qquad = 
  \int d p  \big| \widehat{\W^{(1)}}(p)\big|
  \Big \| \W_{z}^{(n)}(s-\eta) [ \omega_{N,\eta}, e^{i p \cdot \hat{x}}] \big(F *\rho_\eta^{p}\big) \Big \|_{\tr}
  \\
  & \qquad \leq \Big(\sup _{p \in \mathbb{R}^{3}} \| F *\rho_\eta^{p}\|_{\op}  \Big)
  \int d p  \big| \widehat{\W}^{(1)}(p)\big|
  \Big\|  \W_{z}^{(n)}(s-\eta)   [ \omega_{N,\eta}, e^{i p \cdot \hat{x}}] \Big \|_{\tr} \;,
\end{split}
\end{equation*}
then estimate the operator norm as in \eqref{eq:wwrhow-1:op} and the integral as in \eqref{eq:FW[w,eipx]final}, using that $\widehat{\W^{(1)}}$ decays fast enough.

\medskip

\noindent{\underline{Part 3: Conclusion.}} The estimate \eqref{eq:W[w,DVr]} together with the bounds \eqref{eq:rWW[w,F]},  \eqref{eq:FW[w,eipx]final-est} imply:
\begin{equation}\label{eq:conclu}
\begin{split}
\sup _{z \in \mathbb{R}^{3}} X_{\Lambda}(z) & \Big \| \W_{z}^{(n)}(s-\eta)   [ \omega_{N,\eta}, \nabla V*\rho_{\eta} ] \Big \|_{\tr}
\\
&
 \leq C \veps^{-2} + C \sup _{p: |p|\leq \veps^{-1}} \frac{X_{\Lambda}(z)}{1+|p|}  \Big\|  \W_{z}^{(n)}(s-\eta)   [ \omega_{N,\eta}, e^{i p \cdot \hat{x}}] \Big \|_{\tr}
\end{split}
\end{equation}
Combining (\ref{eq:hh}) with the estimate (\ref{eq:conclu}), we have:
\begin{equation*}
\begin{split}
&\sup_{z \in \mathbb{R}^{3}} X_{\Lambda}(z) \;
\Big \|  \W_{z}^{(n)}(s-\tau)  [ \varepsilon \nabla, \omega_{N,\tau} ] \Big \|_{\tr} 
 \\
 &\quad
\leq  
  K \veps^{-2} + C \int_{0} ^{\tau} d \eta  \sup_{p:|p|\leq \varepsilon^{-1}} \sup_{z \in \mathbb{R}^{3}} \frac{X_{\Lambda}(z)}{1 + |p|} \Big \| \W_{z}^{(n)}(s-\eta)  [ \omega_{N,\eta},  e^{ip\cdot \hat x}] \Big \|_{\tr}\;.
 \end{split}
\end{equation*}
Plugging this estimate into (\ref{eq:obtained}) we find:
\begin{equation*}
\begin{split}
&\sup _{s\in [t,T]}\sup_{q:|q|\leq \varepsilon^{-1}}\sup_{z \in \mathbb{R}^{3}}\, \frac{X_{\Lambda}(z)}{1 + |q|} \Big \|  \W_{z}^{(n)}(s-t) \,\big[e^{i q \cdot \hat{x}}, \omega_{N,t}\big] \Big \|_{\tr} \\
&\quad \leq C  \varepsilon^{-2} + C \int_{0}^{t}d\tau\, \int_{0}^{\tau} d\eta\, \sup _{s\in [\eta,T]}\sup_{p:|p|\leq \varepsilon^{-1}} \sup_{z \in \mathbb{R}^{3}} \frac{X_{\Lambda}(z)}{1 + |p|} \Big \| \W_{z}^{(n)}(s-\eta)  [ \omega_{N,\eta},  e^{ip\cdot \hat x} ]  \Big \|_{\tr}\;,
\end{split}
\end{equation*}
where we used that $0\leq \eta\leq \tau  \leq t$. Hence, by the Gronwall lemma we finally get, for $0\leq t\leq T$:
\begin{equation*}
\sup _{s \in [t,T]}
\sup_{p:|p|\leq \varepsilon^{-1}}\sup_{z \in \mathbb{R}^{3}}\, \frac{X_{\Lambda}(z)}{1 + |p|} \Big \|  \W_{z}^{(n)}(s-t) \,\big[e^{i p \cdot \hat{x}}, \omega_{N,t}\big] \Big \|_{\tr} \leq C\varepsilon^{-2}\;.
\end{equation*}
This concludes the proof of Theorem \ref{thm:propcomm}.
\qed

\section{Propagation of the Semiclassical Structure: Pseudo-Relativistic case}\label{sec:proprel}

In this section, we show how to propagate the local semiclassical structure along the flow of the 
pseudo-relativistic Hartree equation 
\begin{equation}
i \veps \partial_t \omega_{N,t} = \big[h_{\mathrm{rel}}(t), \omega_{N,t} \big]
\end{equation}
with $h_{\mathrm{rel}}(t):=\sqrt{1-\veps^2 \Delta }+ V * \rho_t$, and $\rho_t(x):= \varepsilon^{3} \omega_{N,t}(x;x)$. With respect to the non-relativistic case, here we will be able to propagate the local semiclassical structure for all times. This allows us to prove the convergence of the many-body pseudo-relativistic dynamics to the pseudo-relativistic Hartree dynamics, Theorem \ref{thm:mainrel}. 
The following theorem is the analogue of Theorem \ref{thm:propcomm}.
\begin{theorem}[Propagation of the Local Semiclassical Structure.]\label{thm:proprel} Under the same assumptions of Theorem \ref{thm:mainrel}, we have 
\begin{equation}\label{eq:propscrel}
\begin{split} 
\sup _{z \in \R^{3}} X_{\Lambda}(z) \big\|\W_{z}^{(n)}  \, \omega_{N,t} \big\|_{\tr}  &\leq C \exp(Ct)  \veps^{-3} \\
\sup_{q: |q|\leq \varepsilon^{-1}} \sup_{z\in \mathbb{R}^{3}} \frac{X_\Lambda(z)}{1+|q|} \Big \|\W_{z}^{(n)}\,\big[e^{i q \cdot \hat{x}}, \omega_{N,t}\big] \Big \|_{\tr} &\leq C \exp(C\exp Ct) \varepsilon^{-3}\;.
\end{split}
\end{equation}
\end{theorem}
The main improvement with respect to the non-relativistic case is that in the pseudo-relativistic case we are able to rule out excessive concentration of the density globally in time, thanks to the boundedness of the velocity of the particles. We start by adapting the propagation estimate for the localization operators, Proposition \ref{prp:propU}.
\begin{proposition}[Bounds for the Evolution of the Localization Operator]\label{prop:W-rel-prop} Under the same assumptions of Theorem \ref{thm:mainrel}, consider the pseudo-relativistic Hartree evolution generator $U_{\mathrm{rel}}(t;s)$ defined by
\begin{equation}
i\veps \partial_t U_{\mathrm{rel}}(t;s)
= h_{\mathrm{rel}}(t)U_{\mathrm{rel}}(t;s)\;,\qquad U_{\mathrm{rel}}(t;t) = 1\;.
\end{equation}
Then, for all $k \geq 1$, there exists a constant $C$ such that for all $z \in \mathbb{R}^{3}$, $0\leq s\leq t $
\begin{equation}
U_{\mathrm{rel}}(t;s)^{*}\W_{z}^{(k)}(\hat x)U_{\mathrm{rel}}(t;s) \leq e^{C (t-s)} \W_{z}^{(k)}(\hat x) \;.
\end{equation}
\end{proposition}
The reader should compare the result with Eq.~(\ref{eq:propU1}). The reason why we are able to control $U_{\mathrm{rel}}(t)^*\W_z(\hat x)U_{\mathrm{rel}}(t) $ with $\W_z(\hat x)$, for all times, is the fact that the velocity operator $[\hat{x},\sqrt{1-\veps^2 \Delta}]$ is bounded.
\begin{proof}
We compute the derivative
\begin{equation}
\begin{split}
i \veps \partial_{t} U_{\mathrm{rel}}(t;s)^{*}\W_{z}^{(k)}U_{\mathrm{rel}}(t;s) & = U_{\mathrm{rel}}(t;s)^{*}\big[\W_{z}^{(k)},h_{\mathrm{rel}}(t) \big]U_{\mathrm{rel}}(t;s)
\\
& =
U_{\mathrm{rel}}(t;s)^{*}\big[\W_{z}^{(k)},\sqrt{1-\veps^{2}\Delta} \big]U_{\mathrm{rel}}(t;s) \;,
\end{split}
\end{equation}
so that, for any $\phi \in L^{2}(\mathbb{R}^{3})$ we have
\begin{equation}
\begin{split}
&\langle \phi,U_{\mathrm{rel}}(t;s)^{*}\W_{z}^{(k)}U_{\mathrm{rel}}(t;s) \phi  \rangle \leq 
\langle \phi, \W_{z}^{(k)} \phi  \rangle
\\
& \quad
+ \frac{C}{\veps} \Big\| \big(\W_{z}^{(k/2)}\big)^{-1} \big[\W_{z}^{(k)},\sqrt{1-\veps^{2}\Delta} \big]\big(\W_{z}^{(k/2)}\big)^{-1}\Big\|_{\op}\int_{s}^{t}d \tau  
 \langle \phi,U_{\mathrm{rel}}(\tau;s)^{*}\W_{z}^{(k)}U_{\mathrm{rel}}(\tau;s) \phi  \rangle
\end{split}
\end{equation}
where we used that $\big(\W_{z}^{(k/2)}\big)^{2} \leq C\W_{z}^{(k)}$. To control the operator norm, we write
\begin{equation}
\begin{split}
\Big\| &\big(\W_{z}^{(k/2)}\big)^{-1}\big[\W_{z}^{(k)},\sqrt{1-\veps^{2}\Delta} \big]\big(\W_{z}^{(k/2)}\big)^{-1}\Big\|_{\op} \\
&\leq \big\| \big(\W_{z}^{(k/2)}\big)^{-1} \W_{z}^{(k)} \big[ |\hat x - z|^{2k}, \sqrt{1-\veps^{2}\Delta} \big] \W_{z}^{(k)} \big(\W_{z}^{(k/2)}\big)^{-1} \big\|_{\text{op}} \\
&\leq C \Big\| \W_{z}^{(k/2)} \big[|\hat{x}-z|^{2k}, \sqrt{1-\veps^{2}\Delta}\big]\Big\|_{\op}
\\
&\leq C\sum_{\alpha: |\alpha | = 2k} \sum_{\substack{\alpha': |\alpha'|\geq 1 \\ \alpha' \leq \alpha}} \Big\| \W_{z}^{(k/2)} (\hat{x}-z)^{\alpha-\alpha'}\mathrm{ad}^{(\alpha')}_{\hat{x}}\Big(\sqrt{1-\veps^{2}\Delta}\Big)\Big\|_{\op} \;.
\end{split}
\end{equation}
Since $\mathrm{ad}^{(\alpha')}_{\hat{x}}\Big(\sqrt{1-\veps^{2}\Delta}\Big) \leq C\veps^{|\alpha'|}$ for $1 \leq|\alpha'|\leq 2k$, we obtain
\begin{equation*}
\langle \phi,U_{\mathrm{rel}}(t;s)^{*}\W_{z}^{(k)}U_{\mathrm{rel}}(t;s) \phi  \rangle \leq 
\langle \phi, \W_{z}^{(k)} \phi  \rangle + C\int_{s}^{t} d \tau 
\langle \phi,U_{\mathrm{rel}}(\tau;s)^{*}\W_{z}^{(k)}U_{\mathrm{rel}}(\tau;s) \phi  \rangle 
\end{equation*}
which implies the claim by the Gronwall lemma.
\end{proof}
As a corollary, Proposition \ref{prop:W-rel-prop} immediately implies absence of excessive concentration for the density, for all times.
\begin{corollary}[No-Concentration Bound] \label{crl-no-conc-rel} Under the same assumption of Theorem \ref{thm:mainrel}, we have:
\begin{equation}
\sup _{z \in \mathbb{R}^{3}} \tr\, \W^{(1)}_{z} \omega_{N,t} \leq C e^{Ct} \veps^{-3}\;.
\end{equation}
\end{corollary}
\begin{proof}
We have:
\begin{equation}
\begin{split}
\tr\, \W^{(1)}_{z} \omega_{N,t} = \tr\, U_{\mathrm{rel}}(t;0)^{*}\W_{z}^{(1)}U_{\mathrm{rel}}(t;0)\omega_{N} &\leq e^{Ct} \tr\, \W_{z}^{(1)} \omega_{N} \leq C e^{Ct} \varepsilon^{-3}\;,
\end{split}
\end{equation}
where the first inequality follows from Proposition \ref{prop:W-rel-prop} and from the positivity of $\omega_{N}$, while the last inequality follows from the assumption of Eq. (\ref{eq:assump_density}).
\end{proof}

\subsection{Proof of Theorem \ref{thm:proprel}}
The first bound is an immediate consequence of Proposition \ref{prop:W-rel-prop} and of the assumption (\ref{eq:H2}) on the initial datum. Let us now prove the second inequality. By using the Jacobi identity, we write:
\begin{equation}
i\varepsilon \partial_{t} \big[e^{iq\cdot \hat{x}}, \omega_{N,t}\big] = \big[h_{\mathrm{rel}}(t),\big[e^{iq\cdot \hat{x}}, \omega_{N,t}\big]\big] + \big[\omega_{N,t}, \big[ h_{\mathrm{rel}}(t),e^{iq\cdot \hat{x}} \big]\big]\;.
\end{equation}
Consider the second term on the right-hand side. It can be rewritten as
\begin{equation}
\begin{split}
\big[\omega_{N,t}, \big[ h_{\mathrm{rel}}(t),e^{iq\cdot \hat{x}} \big]\big] & =
\big[\omega_{N,t}, \big[\sqrt{1-\veps^2 \Delta } ,e^{iq\cdot \hat{x}} \big]\big]
\\
& =
\big[\omega_{N,t},e^{iq\cdot \hat{x}}  \big(\sqrt{1+\veps^{2}(-i\nabla + q)^{2} } -\sqrt{1-\veps^2 \Delta }\big)\big]
\\
& = 
\big[\omega_{N,t},e^{iq\cdot \hat{x}}  \big] \veps  A(q)
+ e^{iq\cdot \hat{x}} 
\big[\omega_{N,t}, \veps A(q)\big] 
\end{split}
\end{equation}
where we have introduced the operator 
\begin{equation}
A(q):=\int_{0}^{1} d s \frac{\veps(-i \nabla + s q) \cdot q }{\sqrt{1+\veps^{2}(-i\nabla + s q)^{2}}}\;.
\end{equation}
Let us introduce the modified dynamics
\begin{equation*}
i\varepsilon \partial_{t} U_{ \mathrm{rel};q}(t;s) = (h_{\mathrm{rel}}(t)+ \veps A(q)) U_{\mathrm{rel}; q}(t;s)\;,\qquad U_{\mathrm{rel};q}(s;s) = 1\;.
\end{equation*}
This allows us to write 
\begin{equation}
i\varepsilon \partial_{t} U_{\mathrm{rel}}(t;s)^{*} \big[e^{iq \cdot \hat{x}}, \omega_{N,t}\big] U_{\mathrm{rel};q}(t;s)= U_{\mathrm{rel}}(t;s)^{*}e^{iq\cdot \hat{x}} 
\big[\omega_{N,t}, \veps A(q)\big] U_{\mathrm{rel};q}(t;s)\;.
\end{equation}
Writing this equation in integral form we get:
\begin{equation}\label{eq:BPS2-rel}
\begin{split}
\big[e^{iq \cdot \hat{x}}, \omega_{N,t}\big] &= U_{\mathrm{rel}}(t;0) \big[e^{iq \cdot \hat{x}}, \omega_{N}\big] U_{\mathrm{rel};q}(t;0)^{*} -i\int_{0}^{t}d\tau\, U_{\mathrm{rel}}(\tau;t)^{*}e^{iq\cdot \hat{x}}\big[ \omega_{N,\tau},A(q)\big] U_{\mathrm{rel};q}(\tau;t)\;.
\end{split}
\end{equation}
We shall now plug this identity into $\big\|\W_{z}^{(n)}\,\big[e^{i q \cdot \hat{x}}, \omega_{N,t}\big] \big\|_{\tr}$, and estimate the various terms. The first term gives the contribution,
\begin{equation}
\Big\| \W_{z}^{(n)} U_{\mathrm{rel}}(t;0) \big[e^{iq \cdot \hat{x}}, \omega_{N}\big] U_{\mathrm{rel};q}(t;0)^{*}\Big\|_{\tr} 
\leq C_{t} \Big\| \W_{z}^{(n)}  \big[e^{iq \cdot \hat{x}}, \omega_{N}\big] \Big\|_{\tr} \;,
\end{equation}
where we used Proposition \ref{prop:W-rel-prop}, Lemma \ref{lem:mon} and the invariance of the trace under unitary conjugation. We then bound the term due to the integrand in \eqref{eq:BPS2-rel} as follows:
\begin{equation}\label{int-A-rel}
\begin{split}
\Big\|\W_{z}^{(n)}&U_{\mathrm{rel}}(\tau;t)^{*}e^{iq\cdot \hat{x}}\big[ \omega_{N,\tau},A(q)\big] U_{\mathrm{rel};q}(\tau;t) \Big\|_{\tr}
\\
&
\leq C_{t-\tau}
\Big\|\W_{z}^{(n)} \big[ \omega_{N,\tau},A(q)\big]\Big\|_{\tr}
\\
& \leq C_{t-\tau} \int_{0}^{1} d s\, \Big\|\W_{z}^{(n)} \big[ \omega_{N,\tau},\frac{\veps(-i \nabla + sq) \cdot q }{\sqrt{1+\veps^{2}(-i\nabla + s q)^{2}}}\big]\Big\|_{\tr}
\\
& \leq 
C_{t-\tau} \int_{0}^{1} d s
\Big\|\W_{z}^{(n)} \big[ \omega_{N,\tau}, \veps(-i \nabla) \cdot q \big] \frac{1}{\sqrt{1+\veps^{2}(-i\nabla + s q)^{2}}}\Big\|_{\tr}
\\
& 
\qquad \qquad \qquad
+C_{t-\tau} \int_{0}^{1} d s\, \Big\|\W_{z}^{(n)} \veps(-i \nabla + sq) \cdot q \Big[ \omega_{N,\tau}, \frac{1}{\sqrt{1+\veps^{2}(-i\nabla + s q)^{2}}} \Big] \Big\|_{\tr}\;,
\end{split}
\end{equation}
where $C_{\tau} = C\exp(C \tau)$. Since $\| (1+\veps^{2}(-i\nabla + s q)^{2})^{-1/2}\|_{\op}\leq 1$, the first term on the last line of \eqref{int-A-rel} is estimated by 
\begin{equation}\label{first-term-WA-rel}
\int_{0}^{1} d s\, \Big\|\W_{z}^{(n)} \big[ \omega_{N,\tau}, \veps(-i \nabla) \cdot q \big] \frac{1}{\sqrt{1+\veps^{2}(-i\nabla + s q)^{2}}}\Big\|_{\tr} \leq |q| \Big\|\W_{z}^{(n)} \big[ \omega_{N,\tau}, \veps \nabla \big] \Big\|_{\tr} \;.
\end{equation}
To bound the second term, we use the integral representation 
\begin{equation}
\frac{1}{\sqrt{B}} = \frac{1}{\pi} \int_{0}^{\infty} \frac{d \lambda}{\sqrt{\lambda}} \frac{1}{B+\lambda} \, ,
\end{equation}
valid for any self-adjoint $B>0$. Accordingly:
\begin{equation}\label{partialWA-rel}
\begin{split}
&
\Big\|\W_{z}^{(n)} \veps(-i \nabla + sq) \cdot q \Big[ \omega_{N,\tau}, \frac{1}{\sqrt{1+\veps^{2}(-i\nabla + s q)^{2}}} \Big] \Big\|_{\tr}
\\
& \leq 
 \frac{1}{\pi} \int_{0}^{\infty} \frac{d \lambda}{\sqrt{\lambda}}
\Big\|\W_{z}^{(n)} \frac{\veps(-i \nabla + sq) \cdot q}{1+\veps^{2}(-i\nabla + s q)^{2}+\lambda} \big[ \omega_{N,\tau}, \veps^{2}(-i\nabla + s q)^{2} \big] \frac{1}{1+\veps^{2}(-i\nabla + s q)^{2}+\lambda}\Big\|_{\tr} 
\\
&
\leq 
\sum_{j} \frac{1}{\pi} \int_{0}^{\infty} \frac{d \lambda}{\sqrt{\lambda}}
\Big\|\W_{z}^{(n)} \frac{\veps(-i \nabla +s q) \cdot q \; \veps(-i\nabla + s q)_{j}}{1+\veps^{2}(-i\nabla + s q)^{2}+\lambda} \big[ \omega_{N,\tau}, \veps \nabla_{j} \big] \frac{1}{1+\veps^{2}(-i\nabla + s q)^{2}+\lambda}\Big\|_{\tr} 
\\
&
\quad+ \sum_{j} \frac{1}{\pi} \int_{0}^{\infty} \frac{d \lambda}{\sqrt{\lambda}}
\Big\|\W_{z}^{(n)} \frac{\veps(-i \nabla +s q) \cdot q  }{1+\veps^{2}(-i\nabla + s q)^{2}+\lambda} \big[ \omega_{N,\tau}, \veps \nabla_{j} \big] \frac{\veps(-i\nabla + s q)_{j}}{1+\veps^{2}(-i\nabla + s q)^{2}+\lambda}\Big\|_{\tr} \;.
\end{split}
\end{equation}
Using the following bounds, for $|\alpha| \leq 4n$:
\begin{equation*}
\begin{split}
\Big\|\mathrm{ad}^{(\alpha)}_{\hat{x}}\bigg(\frac{\veps(-i \nabla +s q)_{j}}{1+\veps^{2}(-i\nabla + s q)^{2}+\lambda} \bigg) \Big\|_{\op} &\leq C (1+\lambda)^{-\frac{1+|\alpha|}{2}} \;,
\\
\Big\|\mathrm{ad}^{(\alpha)}_{\hat{x}}\bigg(\frac{\veps(-i\nabla + s q)_{j}\veps(-i\nabla + s q)_{k}}{1+\veps^{2}(-i\nabla + s q)^{2}+\lambda} \bigg)\Big\|_{\op} &\leq C (1+\lambda)^{-\frac{|\alpha|}{2}}\;,
\end{split}
\end{equation*}
we obtain
\begin{equation}\label{ad-W-D}
\begin{split}
\Big\|\W_{z}^{(n)} &\frac{\veps(-i \nabla +s q)_{j} }{1+\veps^{2}(-i\nabla + s q)^{2}+\lambda}\big(\W_{z}^{(n)}\big)^{-1}\Big\|_{\op}
\\
& 
\leq 1 + \sum_{\alpha: |\alpha| = 4n} \sum_{\alpha'\leq \alpha}\Big\|\W_{z}^{(n)} (\hat{x}-z)^{\alpha - \alpha'}\mathrm{ad}^{(\alpha')}_{\hat{x}}\bigg(\frac{\veps(-i \nabla +s q)_{j}}{1+\veps^{2}(-i\nabla + s q)^{2}+\lambda} \bigg)\Big\|_{\op}
\\
&
\leq (1+\lambda)^{-1/2}
\end{split}
\end{equation}
and also
\begin{equation}\label{ad-W-DD}
\begin{split}
\Big\|\W_{z}^{(n)} &\frac{\veps(-i \nabla +s q)_{j} \veps(-i\nabla + s q)_{k} }{1+\veps^{2}(-i\nabla + s q)^{2}+\lambda}\big(\W_{z}^{(n)}\big)^{-1}\Big\|_{\op} \leq C
\end{split}
\end{equation}
Putting together \eqref{partialWA-rel}, \eqref{ad-W-D} and \eqref{ad-W-DD} and using $\|(1+\veps^{2}(-i\nabla + s q)^{2}+\lambda)^{-1} \|_{\op}\leq (1+\lambda)^{-1}$, 
we obtain the estimate
\begin{equation}
\begin{split}
\int_{0}^{1} d s
\Big\|\W_{z}^{(n)} \veps(-i \nabla + sq) \cdot q \Big[ \omega_{N,\tau}, \frac{1}{\sqrt{1+\veps^{2}(-i\nabla + s q)^{2}}} \Big] \Big\|_{\tr}
\leq C |q| \Big\| \W_{z}^{(n)}\big[\omega_{N,\tau}, \veps \nabla\big] \Big\|_{\tr} \;,
\end{split}
\end{equation}
which, combined with \eqref{first-term-WA-rel} implies
\begin{equation}
\Big\|\W_{z}^{(n)} U_{\mathrm{rel}}(\tau;t)^{*}e^{iq\cdot \hat{x}}\big[ \omega_{N,\tau},A(q)\big] U_{\mathrm{rel};q}(\tau;t) \Big\|_{\tr} \leq C_{t-\tau} |q| \Big\| \W_{z}^{(n)}\big[\omega_{N,\tau}, \veps \nabla\big] \Big\|_{\tr} \;.
\end{equation}
All in all, we have thus proven that
\begin{equation} \label{XW[eiqx,w]-rel}
\begin{split}
\sup _{z \in \mathbb{R}^{3}} X_{\Lambda}(z) \Big\|\W_{z}^{(n)} \big[ e^{iq \cdot \hat{x}},\omega_{N,t}\big] \Big\|_{\tr} 
& \leq C_{t} \veps^{-2} +  |q| \int_{0}^{t} d \tau \,C_{t-\tau} \sup _{z \in \mathbb{R}^{3}} X_{\Lambda}(z)\Big\| \W_{z}^{(n)}[\omega_{N,\tau}, \veps \nabla] \Big\|_{\tr}\;.
\\
\end{split}
\end{equation}
Proceeding as in the proof of Theorem \ref{thm:propcomm}, we find
\begin{equation}\label{W[w,D]-rel}
\begin{split}
\Big\| \W_{z}^{(n)}[\omega_{N,\tau}, \veps \nabla] \Big\|_{\tr} & \leq C_{\tau}
\Big\| \W_{z}^{(n)}[\omega_{N}, \veps \nabla] \Big\|_{\tr}
 + \int_{0}^{\tau} d\eta  \,C_{\tau-\eta} \Big\|\W_{z}^{(n)}[\omega_{N,\eta}, \nabla V * \rho_{\eta}] \Big\|_{\tr} \;,
\end{split}
\end{equation}
where we used Proposition \ref{prop:W-rel-prop}, Lemma \ref{lem:mon} and the invariance of the trace under unitary conjugation. With respect to the non-relativistic case, notice the simplification introduced by the fact that localization operator is not time-evolved. Letting $F_{y}(\hat{x})$ be as below \eqref{eq:W[w,DVr]}, we have:
\begin{equation}\label{partialW[w,DVrho]rel}
\begin{split}
\Big\|&\W_{z}^{(n)}[\omega_{N,\eta}, \nabla V * \rho_{\eta}] \Big\|_{\tr} 
\\
&\leq 
\Big\|\int d y \rho_{\eta}(y) \W_{z}^{(n)}\W_{y}^{(1)} [\omega_{N,\eta}, F_{y}(\hat{x})] \Big\|_{\tr} 
+
\Big\|\int d y\rho_{\eta}(y) \W_{z}^{(n)} [\omega_{N,\eta},\W_{y}^{(1)}] F_{y}(\hat{x}) \Big\|_{\tr}
\\
&
\leq \sup _{p \in \mathbb{R}^{3}}\Big(\big\| \W^{(1)}* \rho_{\eta}^{p}\big\|_{\op} +\big\| F* \rho_{\eta}^{p}\big\|_{\op}\Big)
\int dp \, \Big(\big|\hat{F}(p)\big|+ \big|\widehat{\W^{(1)}}(p)\big|\Big) \Big\| \W_{z}^{(n)}[\omega_{N,\eta},e^{ip\cdot \hat{x}}]\Big\|_{\tr} \;,
\end{split}
\end{equation}
compare with \eqref{eq:W[w,DVr]} and \eqref{eq:rWW[w,F]}. By Corollary \ref{crl-no-conc-rel} we have that $\big\| \W^{(1)}* \rho_{\eta}^{p}\big\|_{\op} ,\big\| F* \rho_{\eta}^{p}\big\|_{\op}\leq C_{\eta}$ uniformly in $p$,
whereas the integral on the last line in \eqref{partialW[w,DVrho]rel} is controlled by splitting in $|p|\leq \veps^{-1}$ and $|p| >\veps^{-1}$ as was done in \eqref{eq:FW[w,eipx]final}. Accordingly, by using the assumptions on the potential we obtain
\begin{equation}\label{XW[w,DV-r]-rel}
\sup _{z \in \mathbb{R}^{3}} X_{\Lambda}(z)\Big\| \W_{z}^{(n)}[\omega_{N,\eta}, \nabla V * \rho_{\eta}] \Big\|_{\tr} \leq C_{\eta} \veps^{-2}+ C_{\eta} \sup _{p: |p| \leq \veps^{-1}} \frac{X_{\Lambda}(z)}{1+|p|} \Big\|\W_{z}^{(n)}[\omega_{N,\eta} , e^{ip\cdot \hat{x}}] \Big\|_{\tr} \;.
\end{equation}
Putting together the bounds \eqref{XW[eiqx,w]-rel}, \eqref{W[w,D]-rel} and \eqref{XW[w,DV-r]-rel}, we finally get:
\begin{equation}
\begin{split}
\sup _{q: |q| \leq \veps^{-1}}\frac{X_{\Lambda}(z)}{1+|q|}&
 \Big\|\W_{z}^{(n)}[\omega_{N,\eta} , e^{iq\cdot \hat{x}}] \Big\|_{\tr} 
\\
&
\leq C_{t}\veps^{-2} + C_{t} \int_{0}^{t} d \tau \int_{0}^{\tau} d \eta 
 \sup _{p: |p| \leq \veps^{-1}} \frac{X_{\Lambda}(z)}{1+|p|} \Big\|\W_{z}^{(n)}[\omega_{N,\eta} , e^{ip\cdot \hat{x}}] \Big\|_{\tr} \;.
 \end{split}
\end{equation}
The final claim, Eq.~(\ref{eq:propscrel}), follows by the application of Gronwall lemma.
\qed

\appendix

\section{Check of Assumption \ref{ass:sc}}\label{app:localsc}

Here we shall discuss examples of fermionic states satisfying Assumptions \ref{ass:sc}. Specifically, we shall consider the case of the free Fermi gas on $\mathbb{R}^{3}$ at positive density, and of coherent states. We expect these assumptions to hold true for a wide class of fermionic equilibrium states.

\subsection{The Free Fermi Gas}

For the sake of simplicity, here we prove the assumptions for the free Fermi on $\mathbb{R}^{3}$, at positive density. We expect similar estimates to hold true for a free Fermi gas in a large enough, finite periodic box $\Lambda \subset \mathbb{R}^{3}$. Let $\omega$ be the operator on $L^{2}(\mathbb{R}^{3})$ with integral kernel:
\begin{equation*}
\omega(x;y) = \int_{\R^{3}} \frac{d q}{(2 \pi)^{3}} \,\mathbf{1}_{|q| \leq  \veps^{-1}} e^{i q \cdot (x - y)}\;,
\end{equation*}
This operator describes the homogeneous free Fermi gas in $\mathbb{R}^{3}$ at density $\omega(x;x) =  (1/6 \pi^{2})\varepsilon^{-3}$. Clearly, $\omega = \omega^{2} = \omega^{*}$, and $\tr\, \omega = \infty$.

\medskip

\noindent{\underline{Check of (\ref{eq:propcomm1})}.} Being the state defined in $\Lambda = \mathbb{R}^{3}$, here we shall replace $X_{\Lambda}(z)$ with $1$. Simple computations show that the kernel of the operator $[e^{i p \cdot \hat{x}}, \omega]$ is given by
\begin{equation*}
[e^{i p \cdot \hat{x}}, \omega](x,y) =
e^{i p \cdot \frac{x+y}{2}} \int_{\R^{3}} \frac{d q}{(2 \pi)^{3}} \,\Big( \mathbf{1}_{|q- p/2| \leq \veps^{-1} } - \mathbf{1}_{|q+ p/2| \leq \veps^{-1}}\Big) e^{i q \cdot  (x-y)}\;.
\end{equation*}
Furthermore, $\big| [e^{i p \cdot \hat{x}}, \omega]\big|^{2} = \big| [e^{i p \cdot \hat{x}}, \omega]\big|$ and
\begin{equation*}
\big| [e^{i p \cdot \hat{x}}, \omega]\big|^{2}(x;y)  =  \int_{\R^{3}} \frac{d q}{(2 \pi)^{3}} \, \mathbf{1}_{S_{p}} e^{i q \cdot (x-y)}\;,
\end{equation*}
with $S_{p}$ the following set:
\begin{equation*}
S_{p}:= \big\{q \in \R^{3} \,\big| \,|q -  p| \leq \veps^{-1} \big\}\ominus\big\{q \in \R^{3} \, \big| \,|q | \leq \veps^{-1} \big\}\;.
\end{equation*}
Clearly,
\begin{equation*}
|S_{p}| \leq C |p| \varepsilon^{-2}\;.
\end{equation*}
Let $F_{p}(q)$ be a $C^{\infty}$ smoothing of the characteristic function $\chi_{S_{p}}(q)$, such that:
\begin{equation}\label{eq:smoothSp-}
\begin{split}
\chi_{S_{p}}(q) \leq F_{p}(q)\;,\qquad
F_{p}(q)\upharpoonright_{S_{p}} = 1\;, \qquad F_{p}(q) = 0 \quad \mathrm{if}\quad \mathrm{dist}(q,S_{p})\geq 1\;.
\end{split}
\end{equation}
One can check that the following holds:
\begin{equation}\label{eq:smoothSp}
\| F_{p} \|_{1} \leq C (1+|p|) \varepsilon^{-2}\;,
\qquad \quad
\| D^{k} F_{p}\|_{1} \leq C_{\alpha}  \varepsilon^{-2}\;, \quad \forall k > 0\;.
\end{equation}
Let $O_{p}$ be the operator with integral kernel:
\begin{equation*}
O_{p}(x;y) = \int_{\R^{3}} \frac{d q}{(2 \pi)^{3}}\, F_{p}(q) e^{i q \cdot  (x-y)}\;.
\end{equation*}
Then, $\big| [e^{i p \cdot \hat{x}}, \omega]\big|^{2}\leq O_{p}$. In particular, by Lemma \ref{lem:mon}
\begin{equation}\label{eq:WO0}
\big\| \mathcal{W}^{(n)}_{z}(t) [e^{i p \cdot \hat{x}}, \omega] \big\|_{\text{tr}} = \big\| \mathcal{W}^{(n)}_{z}(t) \big| [e^{i p \cdot \hat{x}}, \omega]\big|^{2} \big\|_{\text{tr}} \leq \big\| \mathcal{W}^{(n)}_{z}(t) O_{p} \big\|_{\text{tr}}\;.
\end{equation}
Since $O_{p}$ is invariant under free time evolution we then have, by the invariance under conjugation with unitary transformations:
\begin{equation*}
\big\| \mathcal{W}^{(n)}_{z}(t) O_{p} \big\|_{\text{tr}} = \big\| \mathcal{W}^{(n)}_{z} O_{p} \big\|_{\text{tr}} \;.
\end{equation*}
To bound the right-hand side, it is enough to consider $n=1$. From the inequality $\mathcal{W}_{z}^{(1)}(x)\leq C(1 + |x-z|^{2})^{-2}$, we have by Lemma \ref{lem:mon}
\begin{equation}\label{eq:WO}
\begin{split}
\big\| \mathcal{W}^{(1)}_{z} O_{p} \big\|_{\text{tr}} &\leq C\Big\| \frac{1}{(1 + |\hat x - z|^{2})^{2}} O_{p}  \Big\|_{\text{tr}} \\
&\leq C\Big\| \frac{1}{1 + |\hat x - z|^{2}} O_{p} \frac{1}{1 + |\hat x - z|^{2}}  \Big\|_{\text{tr}} + C\Big\| \frac{1}{1 + |\hat x - z|^{2}} \Big[ \frac{1}{1 + |\hat x - z|^{2}}, O_{p} \Big] \Big\|_{\tr}\\
&\equiv \text{I} + \text{II}\;.
\end{split}
\end{equation}
Consider the first term. Let:
\begin{equation*}
g_{z,p}(x) = \frac{e^{ip \cdot x}}{(1 + |x - z|^{2})^{2}}\;.
\end{equation*}
Clearly, $g_{z,p} \in L^{2}(\mathbb{R}^{3})$. Moreover, by the first bound in (\ref{eq:smoothSp}):
\begin{equation}\label{eq:Ismooth}
|\text{I}| \leq C\int dq\, F(q) \big\| |g_{z,p}\rangle \langle g_{z,p}| \big\|_{\mathrm{tr}} \leq K |p| \varepsilon^{-2}\;.
\end{equation}
Consider now the term $\text{II}$. We have:
\begin{equation*}
\begin{split}
\Big[ \frac{1}{1 + |\hat x - z|^{2}}, O_{p} \Big]  &= - \frac{1}{1 + |\hat x - z|^{2}} \Big[ |\hat x - z|^{2}, O_{p} \Big] \frac{1}{1 + |\hat x - z|^{2}} \\
& = - \sum_{i=1}^{3} \frac{1}{1 + |\hat x - z|^{2}} \Big[ (\hat x_{i} - z_{i})^{2}, O_{p} \Big] \frac{1}{1 + |\hat x - z|^{2}}\;.
\end{split}
\end{equation*}
We then have:
\begin{equation*}
\begin{split}
&\Big\| \frac{1}{1 + |\hat x - z|^{2}} \Big[ \frac{1}{1 + |\hat x - z|^{2}}, O_{p} \Big] \Big\|_{\text{tr}} \\&\leq \sum_{i=1}^{3} \Big\| \frac{1}{(1 + |\hat x - z|^{2})^{2}} \Big[ (\hat x_{i} - z_{i})^{2}, O_{p} \Big] \frac{1}{1 + |\hat x - z|^{2}} \Big\|_{\text{tr}} \\
& \leq \sum_{i=1}^{3} \Big\| \frac{(\hat x_{i} - z_{i})}{(1 + |\hat x - z|^{2})^{2}} \Big[ \hat x_{i}, O_{p} \Big] \frac{1}{1 + |\hat x - z|^{2}} \Big\|_{\text{tr}} + \sum_{i=1}^{3}\Big\| \frac{1}{(1 + |\hat x - z|^{2})^{2}} \Big[ \hat x_{i}, O_{p} \Big] \frac{(\hat x_{i} - z_{i})}{1 + |\hat x - z|^{2}} \Big\|_{\text{tr}} \\
&\equiv \text{II}_{1} + \text{II}_{2}\;.
\end{split}
\end{equation*}
Consider $\text{II}_{1}$. We have:
\begin{equation*}
|\text{II}_{1}| \leq C \sum_{i =1}^{3}\Big\| \frac{1}{1 + |\hat x - z|^{2}} \Big[ \hat x_{i}, O_{p} \Big] \frac{1}{1 + |\hat x - z|^{2}} \Big\|_{\text{tr}}\;.
\end{equation*}
Using that:
\begin{equation*}
\Big[ \hat x_{i}, O_{p} \Big] = i \int_{\R^{3}} \frac{d q}{(2 \pi)^{3}} \, \partial_{q_{i}} F_{p}(q) |e^{i q \cdot x} \rangle \langle e^{iq \cdot x}|\;,
\end{equation*}
and recalling the second bound in (\ref{eq:smoothSp}), we have:
\begin{equation}\label{eq:II1}
|\text{II}_{1}| \leq C \sup _{i =1,2,3}\int dq\, | \partial_{q_{i}} F_{p}(q) | \leq K \varepsilon^{-2}\;.
\end{equation}
Consider now the term $\text{II}_{2}$. We have:
\begin{equation*}
\begin{split}
&\sum_{i=1}^{3}\Big\| \frac{1}{(1 + |\hat x - z|^{2})^{2}} \Big[ \hat x_{i}, O_{p} \Big] \frac{(\hat x_{i} - z_{i})}{1 + |\hat x - z|^{2}} \Big\|_{\text{tr}} \\
&\leq \sum_{i=1}^{3}\Big\| \frac{1}{1 + |\hat x - z|^{2}} \Big[ \hat x_{i}, O_{p} \Big] \frac{1}{1 + |\hat x - z|^{2}} \frac{(\hat x_{i} - z_{i})}{1 + |\hat x - z|^{2}} \Big\|_{\text{tr}} \\
& \quad + \sum_{i=1}^{3}\Big\| \frac{1}{1 + |\hat x - z|^{2}} \Big[ \frac{1}{1 + |\hat x - z|^{2}}, \Big[ \hat x_{i}, O_{p} \Big]\Big] \frac{(\hat x_{i} - z_{i})}{1 + |\hat x - z|^{2}}\Big] \Big\|_{\text{tr}} \equiv \text{II}_{2;1} + \text{II}_{2;2}\;.
\end{split}
\end{equation*}
The first term is bounded as $\text{II}_{1}$:
\begin{equation}\label{eq:II21}
|\text{II}_{2;1}| \leq \sum_{i=1}^{3}\Big\| \frac{1}{1 + |\hat x - z|^{2}} \Big[ \hat x_{i}, O_{p} \Big] \frac{1}{1 + |\hat x - z|^{2}} \Big\|_{\text{tr}} \leq C |p|\varepsilon^{-2}\;.
\end{equation}
Finally, consider $\text{II}_{2;2}$. Writing:
\begin{equation*}
\begin{split}
-\Big[ \frac{1}{1 + |\hat x - z|^{2}}, \Big[ \hat x_{i}, O_{p} \Big]\Big]  &= \frac{1}{1 + |\hat x - z|^{2}} \Big[ |\hat x - z|^{2} , \Big[ \hat x_{i}, O_{p} \Big]\Big] \frac{1}{1 + |\hat x - z|^{2}} \\
&= \sum_{i=1}^{3} \frac{(\hat x_{i} - z_{i})}{1 + |\hat x - z|^{2}} \Big[ \hat x_{i} , \Big[ \hat x_{i}, O_{p} \Big]\Big] \frac{1}{1 + |\hat x - z|^{2}} \\
&\quad + \sum_{i=1}^{3} \frac{1}{1 + |\hat x - z|^{2}} \Big[ \hat x_{i} , \Big[ \hat x_{i}, O_{p} \Big]\Big] \frac{(\hat x_{i} - z_{i})}{1 + |\hat x - z|^{2}}\;,
\end{split}
\end{equation*}
it is clear that $\text{II}_{2;2}$ can be estimated in terms of a sum of terms bounded by:
\begin{equation}\label{eq:II22}
\Big\|\frac{1}{1 + |\hat x - z|^{2}} \Big[ \hat x_{i} , \Big[ \hat x_{i}, O_{p} \Big]\Big] \frac{1}{1 + |\hat x - z|^{2}} \Big\|_{\text{tr}} \leq C\int dq\, | \partial_{q_{i}}^{2} F_{p}(q) | \leq K|p| \varepsilon^{-2}\;,
\end{equation}
where the last inequality follows from (\ref{eq:smoothSp}). Hence, (\ref{eq:II1}), (\ref{eq:II21}), (\ref{eq:II22}) imply:
\begin{equation*}
|\text{II}| \leq C |p| \varepsilon^{-2}\;.
\end{equation*}
Combined with (\ref{eq:Ismooth}) and with (\ref{eq:WO}), we have:
\begin{equation*}
\big\| \mathcal{W}^{(1)}_{z} O_{p} \big\|_{\text{tr}} \leq C|p| \varepsilon^{-2}\;.
\end{equation*}
Recalling (\ref{eq:WO0}), this concludes the check of the assumption (\ref{eq:propcomm1}) for the free Fermi gas.

\medskip

\noindent{\underline{Check of (\ref{eq:propcomm2})}.} This assumption is trivially true for the free Fermi gas, since $[ \omega, \nabla ] = 0$.

\medskip

\noindent{\underline{Check of (\ref{eq:H2})}.} By stationarity of the free Fermi gas:
\begin{equation}\label{eq:pa2}
\big\| \omega\, \W_{z}^{(n)}(t) \big\|_{\tr} =  \big\| \omega\, \W_{z}^{(n)} \big\|_{\tr} \leq C\varepsilon^{-3}\;,
\end{equation}
where the last bound is proven as we did with the assumption (\ref{eq:propcomm1}), replacing $[e^{ip \cdot \hat x}, \omega]$ with $\omega$.
Finally, since we allow for the value $n=1$ in the localizer, assumption \eqref{eq:assump_density} in Proposition \ref{prp:density} immediately follows.

\subsection{Coherent States}

Let $\rho \in L^{1}(\mathbb{R}^{3}) \cap L^{\infty}(\mathbb{R}^{3})$, $\rho(r)\geq 0$, such that
\begin{equation}\label{eq:rhoass}
\int_{\mathbb{R}^{3}} dr\, \rho(r) = N\;,\qquad | \rho(r) | \leq C\varepsilon^{-3}\;,\qquad X_{\Lambda}(r) \rho(r)^{2/3} \leq C\varepsilon^{-2}\;,\qquad | \nabla_{r} \rho(r)^{1/3} | \leq C \varepsilon^{-1}\;,
\end{equation}
where recall that $X_{\Lambda}(r):= 1+\mathrm{dist}(r,\Lambda)^{4}$. 
The function $\rho(r)$ plays the role of density for the fermionic state that we are going to introduce. The second inequality in (\ref{eq:rhoass}) introduces a form of localization of $\rho(r)$ in the domain $\Lambda$, while the last one allows to bound derivatives of the local Fermi momentum, to be defined below. Here we shall consider coherent states, corresponding to the following reduced one-particle density matrix:
\begin{equation}\label{eq:coherentdef}
\omega_{N} = \frac{1}{(2\pi)^{3}}\int_{\mathbb{R}^{3} \times \mathbb{R}^{3}} dqdr M(q,r) \pi_{q,r}\;,\qquad \pi_{q,r} = | f_{q,r} \rangle \langle f_{q,r} |\;, 
\end{equation}
with $f_{q,r}(x) = e^{iq \cdot x} g(x-r)$, where the function $g$ is even, $\|g\|_{2} = 1$, smooth and fast decaying; for definiteness, we choose $g(x) = \frac{1}{(2\pi \delta^{2})^{3/4}} e^{-\frac{|x|^{2}}{2 \delta^{2}}}$ with $\delta > 0$ to be chosen later. We also set:
\begin{equation*}
M(q,r) = \mathbf{1}_{|q| \leq k_{F}(r)}\;,
\end{equation*}
where we choose the local Fermi momentum $k_{F}(r) = \kappa \rho(r)^{1/3}$, with $\kappa = (6 \pi^{2})^{1/3}$. With this choice:
\begin{equation*}
\tr\, \omega_{N} = \frac{1}{(2\pi)^{3}}\int dqdr M(q,r) = \int dr \rho(r) = N\;.
\end{equation*}
\medskip

\noindent{\underline{Closeness to a projection.}} The state $\omega_{N}$ is not an orthogonal projection. However, it can be viewed as an approximate projection, in the following sense. Consider the quantity $\tr\, \omega_{N} (1 - \omega_{N})$. Since $0\leq \omega_{N} \leq 1$, it satisfies the trivial bound $\tr\, \omega_{N} (1 - \omega_{N}) \leq N$. We claim that, for $\omega_{N}$ given by (\ref{eq:coherentdef}), for $\delta =\sqrt{\varepsilon}$,
\begin{equation}
\tr\, \omega_{N} (1 - \omega_{N}) \leq C \sqrt{\varepsilon}N\;.
\end{equation}
To prove this estimate, we proceed as follows. We write:
\begin{equation}
\tr\, ( \omega_{N} - \omega_{N}^{2} ) = \frac{1}{(2\pi)^{6}}\int dqdq' drdr' M(q,r) \big( M(q,r) - M(q',r') \big) |\langle f_{q,r}, f_{q',r'} \rangle|^{2}
\end{equation}
where we used the completeness of coherent states, and the fact that $M(q,r) = M(q,r)^{2}$. Next, notice that:
\begin{equation}
M(q,r) \big( M(q,r) - M(q',r') \big) = \chi(|q| \leq k_{F}(r)) \big( \chi(|q| \leq k_{F}(r)) - \chi(|q'| \leq k_{F}(r')) \big)\;,
\end{equation}
which implies:
\begin{equation}
\begin{split}
&\int dqdq' drdr' M(q,r) \big( M(q,r) - M(q',r') \big) |\langle f_{q,r}, f_{q',r'} \rangle|^{2} \\
&\qquad = \int dqdq' drdr' \chi(|q| \leq k_{F}(r)) \chi(|q'| > k_{F}(r')) |\langle f_{q,r}, f_{q',r'} \rangle|^{2}\;.
\end{split}
\end{equation}
We compute:
\begin{equation}
|\langle f_{q,r}, f_{q',r'} \rangle|^{2} = e^{-(r-r')^{2} / 2\delta^{2} - (q-q')^{2} \delta^{2} / 2}\;,
\end{equation}
and we consider the integral:
\begin{equation}\label{eq:qq'}
\int dqdq' \chi(|q| \leq k_{F}(r)) \chi(|q'| > k_{F}(r')) e^{- (q-q')^{2} \delta^{2} / 2}\;.
\end{equation}
By the regularity properties of the Fermi momentum (\ref{eq:rhoass}), we have:
\begin{equation}
k_{F}(r') = k_{F}(r) + k_{F}(r') - k_{F}(r) \geq k_{F}(r) - C \varepsilon^{-1} |r - r'|\;.
\end{equation}
Therefore, the expression in (\ref{eq:qq'}) is bounded above by:
\begin{equation}
\int dqdq' \chi(|q| \leq k_{F}(r)) \chi\big(|q'| > k_{F}(r) - C \varepsilon^{-1} |r - r'|\big) e^{- (q-q')^{2} \delta^{2} / 2}\;,
\end{equation}
which we further decompose as:
\begin{equation}\label{eq:further}
\begin{split}
&\int dqdq' \chi(|q| \leq k_{F}(r)) \chi\big(k_{F}(r) + \delta^{-1} > |q'| > k_{F}(r) - C \varepsilon^{-1} |r - r'|\big) e^{- (q-q')^{2} \delta^{2} / 2} \\
&\qquad  + \int dqdq' \chi(|q| \leq k_{F}(r)) \chi\big(k_{F}(r) + \delta^{-1} \leq |q'|\big) e^{- (q-q')^{2} \delta^{2} / 2}\;.
\end{split}
\end{equation}
The second term is easily estimated as:
\begin{equation}
\int dqdq' \chi(|q| \leq k_{F}(r)) \chi\big(k_{F}(r) + \delta^{-1} \leq |q'|\big) e^{- (q-q')^{2} \delta^{2} / 2} \leq C k_{F}(r)^{3}\;,
\end{equation}
which contributes to $\tr\,\omega_{N} (1 - \omega_{N})$ with a term bounded by:
\begin{equation}
C\int drdr'\, k_{F}(r)^{3} e^{-(r-r')^{2} / 2\delta^{2}} \leq C N \delta^{3}\;.
\end{equation}
Consider now the first term in (\ref{eq:further}). We estimate it as:
\begin{equation}
\begin{split}
\int dqdq' \chi(|q| \leq k_{F}(r))\chi\big(k_{F}(r) + \delta^{-1} > &|q'| > k_{F}(r) - C \varepsilon^{-1} |r - r'|\big) e^{- (q-q')^{2} \delta^{2} / 2} \\ &\leq C\delta^{-3} k_{F}(r)^{2} ( \delta^{-1} + C \varepsilon^{-1} |r - r'| )\;;
\end{split}
\end{equation}
this contributes to $\tr\,\omega_{N} (1 - \omega_{N})$ with a term bounded by:
\begin{equation}
\begin{split}
C\int drdr' \delta^{-3} k_{F}(r)^{2} &( \delta^{-1} + C \varepsilon^{-1} |r - r'| ) e^{-(r-r')^{2} / 2\delta^{2}} \\
& \leq K \int dr k_{F}(r)^{2} ( \delta^{-1} + C \varepsilon^{-1} \delta ) \\
&= K \int dr \frac{1}{X_{\Lambda}(r)} X_{\Lambda}(r) k_{F}(r)^{2} ( \delta^{-1} + C \varepsilon^{-1} \delta ) \\
&\leq C |\Lambda| \varepsilon^{-2/3} ( \delta^{-1} + C \varepsilon^{-1} \delta )\;,
\end{split}
\end{equation}
where we used the assumptions (\ref{eq:rhoass}). Putting everything together, and choosing $\delta = \sqrt{\varepsilon}$ we find:
\begin{equation}
\tr\,\omega_{N} (1 - \omega_{N}) \leq C N \sqrt{\varepsilon}
\end{equation}
as claimed.

\medskip

\noindent{\underline{Check of (\ref{eq:propcomm1}).}} We write:
\begin{equation*}
\begin{split}
[e^{i p \cdot \hat{x}}, \omega_{N}] = \frac{1}{(2\pi)^{3}} \int d q\, d r \, \Big[M(q-p/2,r)- M(q+p/2,r) \Big]
|f_{q+ p/2,r}\rangle \langle f_{q- p/2,r}| \;,
\end{split}
\end{equation*}
and notice that 
\begin{equation*}
\Big|M(q-p/2,r)- M(q+p/2,r) \Big| = \mathbf{1}_{S_{p}(r)}(q)
\end{equation*}
the set $S_{p}(r)$ being the symmetric difference of two Fermi balls of radius $k_{F}(r)$, shifted by $p$, i.e.,
\begin{equation}
\label{eq:Spr}
S_{p}(r):= \big\{q \in \R^{3} \,\big| \,|q -  p/2| \leq \kappa \rho(r)^{1/3} \big\}
\ominus
\big\{q \in \R^{3} \, \big| \,|q +  p/2| \leq \kappa \rho(r)^{1/3} \big\}
\end{equation}
with measure
\begin{equation}
\label{eq:Spcoh}
|S_{p}(r)| \leq C |p|  \rho(r)^{2/3} \;.
\end{equation}
We compute:
\begin{equation}\label{eq:commcoh}
\begin{split}
\Big\| \W^{(n)}_{z}(t) [e^{i p \cdot \hat{x}}, \omega_{N}]  \Big\|_{\text{tr}} &\leq \frac{1}{(2\pi)^{3}}\int dqdr\, \mathbf{1}_{S_{p}(r)}(q) \Big\| \W_{z}^{(n)}(t) |f_{q+ p /2,r} \rangle \langle f_{q- p /2,r}| \Big\|_{\tr} \\
&= \frac{1}{(2\pi)^{3}}\int dqdr\, \mathbf{1}_{S_{p}(r)}(q)  \big\| \W_{z}^{(n)}(t) f_{q+ p /2,r} \big\|_{2} \\
&=  \frac{1}{(2\pi)^{3}}\int dqdr\, \mathbf{1}_{S_{p}(r)}(q)  \big\| \W_{z}^{(n)} e^{i \varepsilon \Delta t} f_{q+ p /2,r} \big\|_{2} \\
&= \frac{1}{(2\pi)^{3}}\int dqdr\, \frac{\mathbf{1}_{S_{p}(r)}(q)}{1 + |z-r|^{4n}} \big\|(1 + |z-r|^{4n}) \W_{z}^{(n)} e^{i \varepsilon \Delta t} f_{q+ p /2,r} \big\|_{2}\;.
\end{split}
\end{equation}
We estimate:
\begin{equation}\label{eq:insertion-hatx-f}
\begin{split}
&\big\|(1 + |z-r|^{4n}) \W_{z}^{(n)} e^{i \varepsilon \Delta t} f_{q+ p /2,r} \big\|_{2} \\
&\quad \leq C \big\|(1 + |z-\hat x|^{4n}) \W_{z}^{(n)} e^{i \varepsilon \Delta t} f_{q+ p /2,r} \big\|_{2} + C \big\|(1 + |\hat x-r|^{4n}) \W_{z}^{(n)} e^{i \varepsilon \Delta t} f_{q+ p /2,r} \big\|_{2}\\
&\quad \leq C + C\big\|(1 + |\hat x-r|^{4n}) \W_{z}^{(n)} e^{i \varepsilon \Delta t} f_{q+ p /2,r} \big\|_{2}\;.
\end{split}
\end{equation}
We estimate the second term as, using the unitarity of the free dynamics:
\begin{equation}\label{eq:unita}
\begin{split}
\big\|(1 + |\hat x-r|^{4n}) \W_{z}^{(n)} e^{i \varepsilon \Delta t} f_{q+ p /2,r} \big\|_{2} &\leq  \big\|(1 + |\hat x(t)-r|^{4n})  f_{q+ p /2,r} \big\|_{2} \\
&\leq \big\|(1 + |\hat{x}(t)-r - 2\varepsilon (q + p/2)t |^{4n})  g(\cdot-r) \big\|_{2}\\
&\leq C(1 + \varepsilon^{4n} |q + p/2|^{4n} t^{4n} + (\varepsilon \delta^{-1} t)^{4n})\;.
\end{split}
\end{equation}
The last inequality follows from the smoothness of $g$, and from its fast decay at infinity. Therefore, going back to (\ref{eq:commcoh}), for $\veps|p|\leq 1$, using that $\varepsilon \delta^{-1} \leq C$:
\begin{equation}\label{eq:unita2}
\begin{split}
\Big\| \W^{(n)}_{z}(t) [e^{i p \cdot \hat{x}}, \omega_{N}]  \Big\|_{\text{tr}} &\leq \int dqdr\, \frac{\mathbf{1}_{S_{p}(r)}(q)}{1 + |z-r|^{4n}} C(1 + \varepsilon^{4n} |q + p/2|^{4n} t^{4n}) \\
&\leq C(1 + t^{4n}) \int dqdr\, \frac{\mathbf{1}_{S_{p}(r)}(q)}{1 + |z-r|^{4n}}\;,
\end{split}
\end{equation}
where in the last step we used that, by the compact support properties of the integral and by (\ref{eq:rhoass}) and \eqref{eq:Spr}, $\mathbf{1}_{S_{p}(r)}(q) = 0$, if $|q|\geq C\varepsilon^{-1}$. Recalling the bound (\ref{eq:Spcoh}):
\begin{equation*}
\Big\| \W^{(n)}_{z}(t) [e^{i p \cdot \hat{x}}, \omega_{N}]  \Big\|_{\text{tr}} \leq C(1 + t^{4n})  \int dr\, \frac{|p| \rho(r)^{2/3}}{1 + |z-r|^{4n}}\;;
\end{equation*}
hence,
\begin{equation*}
\begin{split}
\sup_{p:|p|\leq \varepsilon^{-1}} \sup_{z\in \mathbb{R}^{3}} \frac{X_{\Lambda}(z)}{1 + |p|} \Big\| \W^{(n)}_{z}(t) [e^{i p \cdot \hat{x}}, \omega_{N}]  \Big\|_{\text{tr}} \leq K(1 + t^{4n}) \sup_{z\in \mathbb{R}^{3}} X_{\Lambda}(z) \int dr\, \frac{\rho(r)^{2/3}}{1 + |z-r|^{4n}}\;.
\end{split}
\end{equation*}
To estimate the supremum, we proceed as follows, using that by the triangle inequality $X_{\Lambda}(z) \leq C X_{\Lambda}(r) (1+|z-r|^{4})$:
\begin{equation*}
\begin{split}
X_{\Lambda}(z) \int dr\, \frac{\rho(r)^{2/3}}{1 + |z-r|^{4n}} &= \int dr\, \frac{X_{\Lambda}(z)}{X_{\Lambda}(r)} X_{\Lambda}(r) \frac{\rho(r)^{2/3}}{1 + |z-r|^{4n}} \\
&\leq \int dr\, X_{\Lambda}(r) \frac{\rho(r)^{2/3}}{1 + |z-r|^{4(n-1)}} \\
&\leq C\varepsilon^{-2}\;,
\end{split}
\end{equation*}
where the last bound follows from the last assumption in (\ref{eq:rhoass}). This concludes the check of (\ref{eq:propcomm1}) for $n\geq 2$.

\medskip

\noindent{\underline{Check of (\ref{eq:propcomm2}).}} We start by writing:
\begin{equation*}
\begin{split}
[ \nabla, \omega_{N} ](x;y) &= (\nabla_{x} + \nabla_{y}) \omega_{N}(x;y) \\
&= \frac{1}{(2\pi)^{3}}\int_{\mathbb{R}^{3} \times \mathbb{R}^{3}} dqdr M(q,r) e^{iq(x-y)} \big( \nabla_{x} g(x-r) \overline{g(y-r)} + g(x-r) \nabla_{y} \overline{g(y-r)} \big)\\
&\equiv -\frac{1}{(2\pi)^{3}}\int_{\mathbb{R}^{3} \times \mathbb{R}^{3}} dqdr M(q,r) e^{iq(x-y)} \nabla_{r} g(x-r) \overline{g(y-r)}\;.
\end{split}
\end{equation*}
Integrating by parts, we get:
\begin{equation}
[ \nabla, \omega_{N} ](x;y) = \frac{1}{(2\pi)^{3}}\int_{\mathbb{R}^{3} \times \mathbb{R}^{3}} dqdr\, \delta(|q| - k_{F}(r)) (\nabla_{r} k_{F}(r)) e^{iq(x-y)}  g(x-r) \overline{g(y-r)}\;.
\end{equation}
We then have, proceeding as in (\ref{eq:commcoh})-(\ref{eq:unita2}):
\begin{equation*}
\begin{split}
&\Big\| \W^{(n)}_{z}(t) [ \varepsilon\nabla, \omega_{N} ]  \Big\|_{\text{tr}} \\
&\leq C\varepsilon \int_{|q| = k_{F}(r)} dqdr\, \frac{|\nabla_{r} k_{F}(r)|}{1 + |z-r|^{4n}} \big\|(1 + |z-r|^{4n}) \W_{z}^{(n)} e^{i \varepsilon \Delta t} e^{iq \cdot \hat{x}} g(\cdot-r) \big\|_{2} \\
&\leq C\varepsilon (1 + t^{4n}) \int dr\, \frac{k_{F}(r)^{2} |\nabla_{r} k_{F}(r)|}{1 + |z-r|^{4n}}\;.
\end{split}
\end{equation*}
Hence:
\begin{equation*}
\begin{split}
\sup_{z \in \mathbb{R}^{3}} X_{\Lambda}(z) \Big\| \W^{(n)}_{z}(t) [ \varepsilon\nabla, \omega_{N} ]  \Big\|_{\text{tr}} &\leq C\varepsilon (1 + t^{4n}) \sup_{z \in \mathbb{R}^{3}} X_{\Lambda}(z) \int dr\, \frac{k_{F}(r)^{2} |\nabla_{r} k_{F}(r)|}{1 + |z-r|^{4n}} \\
&= C\varepsilon (1 + t^{4n}) \sup_{z \in \mathbb{R}^{3}} \int dr\, \frac{X_{\Lambda}(z)}{X_{\Lambda}(r)} X_{\Lambda}(r) \frac{k_{F}(r)^{2} |\nabla_{r} k_{F}(r)|}{1 + |z-r|^{4n}} \\
&\leq C\varepsilon (1 + t^{4n}) \sup_{z \in \mathbb{R}^{3}} \int dr\, X_{\Lambda}(r) \frac{k_{F}(r)^{2} |\nabla_{r} k_{F}(r)|}{1 + |z-r|^{4(n-1)}} \\
&\leq K (1 + t^{4n}) \varepsilon^{-2}\;,
\end{split}
\end{equation*}
where the last step follows from the assumptions in (\ref{eq:rhoass}), recalling that $k_{F}(r) = \kappa \rho(r)^{1/3}$. This concludes the check of (\ref{eq:propcomm2}) for $n \geq 2$.

\medskip

\noindent{\underline{Check of (\ref{eq:H2}).}} To begin, we estimate:
\begin{equation*}
\begin{split}
\Big\| \W_{z}^{(n)}(t) \omega_{N} \Big\|_{\tr} &\leq \frac{1}{(2\pi)^{3}}\int dqdr\, M(q,r) \big\|  \W_{z}^{(n)}(t) f_{q,r} \big\|_{2} \\
&\leq \frac{1}{(2\pi)^{3}}\int dqdr\, M(q,r) \big\|\W_{z}^{(n)} e^{i\varepsilon \Delta t} f_{q,r} \big\|_{2}
\\
&
= \frac{1}{(2\pi)^{3}}\int dqdr\, \frac{M(q,r)}{1+|z-r|^{4n}} \big\|(1+|z-r|^{4n})\W_{z}^{(n)} e^{i\varepsilon \Delta t} f_{q,r} \big\|_{2}
\;,
\end{split}
\end{equation*}
where in the second last step we used the unitarity of the free dynamics. Next, following \eqref{eq:insertion-hatx-f} and \eqref{eq:unita}, we have:
\begin{equation*}
\begin{split}
\Big\| (1 + |z-r|^{4n}) \W_{z}^{(n)} e^{i\varepsilon \Delta t} f_{q,r} \Big\|_{2} &\leq C(1 + \varepsilon^{4n} |q|^{4n} t^{4n} + (\varepsilon \delta^{-1} t)^{4n}) \;.
\end{split}
\end{equation*}
To conclude, using that the integral is supported for $|q|\leq k_{F}(r) \leq C\varepsilon^{-1}$:
\begin{equation*}
\begin{split}
\sup_{z \in \mathbb{R}^{3}} X_{\Lambda}(z)\Big\| \W_{z}^{(n)}(t) \omega_{N} \Big\|_{\tr} &\leq  C \sup_{z \in \mathbb{R}^{3}} X_{\Lambda}(z) \int dqdr\, \frac{M(q,r)(1 +\veps^{4n} |q|^{4n} t^{4n} + (\varepsilon \delta^{-1} t)^{4n})}{1 + |z-r|^{4n}} \\
&\leq  C (1 + t^{4n}) \sup_{z \in \mathbb{R}^{3}} \int dr\, \frac{X_{\Lambda}(z)}{X_{\Lambda}(r)} X_{\Lambda}(r) \frac{\rho(r)}{1 + |z-r|^{4n}} \\
&\leq C(1 + t^{4n}) \sup_{z \in \mathbb{R}^{3}} \int dr\, X_{\Lambda}(r) \frac{\rho(r)}{1 + |z-r|^{4(n-1)}} \\
&\leq C (1 + t^{4n}) \varepsilon^{-3}\;,
\end{split}
\end{equation*}
where the last step follows from the assumptions (\ref{eq:rhoass}). This concludes the check of (\ref{eq:H2}), and establishes the validity of the local semiclassical structure for coherent states for $n\geq 2$. 

\medskip

\noindent{\underline{Check of (\ref{eq:assump_density}).}} To conclude, we check the validity of the assumption \eqref{eq:assump_density}. This follows from the previous computations, in fact:
\begin{equation*}
\begin{split}
\sup_{z \in \mathbb{R}^{3}} \tr\, \W^{(1)}_{z}(t) \omega_{N}  \leq C (1 + t^{4}) \sup _{z \in \mathbb{R}^{3}} \int dr \, \frac{\rho(r)}{1 + |z-r|^{4}}
\leq C (1+t^{4}) \veps^{-3} \;.
\end{split}
\end{equation*}

\medskip

\noindent
\textbf{Acknowledgements.} We thank Alessandro Giuliani for suggesting the relation of our result with the Kac limit. The work of L.F.~has been supported by the European Research Council through the ERC-AdG CLaQS and by the Swiss National Science Foundation grant number 200160. The work of M.P.~has been supported
by the European Research Council (ERC) under the European Union's Horizon 2020 research and innovation program ERC StG MaMBoQ, n.80290. The work of M.P. has been carried out under the auspices of the GNFM of INdAM. B.S.~acknowledges partial support from the NCCR SwissMAP, from the Swiss National Science Foundation through the Grant ``Dynamical
and energetic properties of Bose-Einstein condensates'' and from the European Research Council through the ERC-AdG CLaQS.

\end{document}